\documentclass[11pt]{article}
\usepackage[utf8]{inputenc}
\usepackage{amsfonts,amsmath,amsthm,amssymb,mathrsfs,graphicx,bm,verbatim,dsfont}
\usepackage{tabularx,lipsum,array,makecell,booktabs,caption,multirow,float,wasysym}
\usepackage[table]{xcolor}
\usepackage[hang,flushmargin]{footmisc} % alignment and indentation of footnotes
\usepackage[letterpaper, margin=1.0 in]{geometry}
\usepackage[auth-lg]{authblk} % author and affiliation blocks

\usepackage{adjustbox}
\usepackage{lscape}
\usepackage{dcolumn}
\newcolumntype{L}{D{.}{.}{2,4}}

\usepackage[american]{babel}
\usepackage[longnamesfirst,authoryear]{natbib}
\cleardoublepage
\normalbaselines % fixes spacing of bibliography

\usepackage{hyperref,url}
\hypersetup{colorlinks,allcolors=blue}
\usepackage[capitalise]{cleveref}
\graphicspath{{./figures/}}

\newtheorem{theorem}{Theorem}[section]

\newtheorem{lemma}{Lemma}[section]
\theoremstyle{definition}

\theoremstyle{definition}
\newtheorem{remark}{Remark}[section]

\newtheorem*{notation}{Notation}

\numberwithin{equation}{section}
\theoremstyle{definition}
\newtheorem{assumption}{Assumption}%[section]

\newtheorem{assumptionalt}{Assumption}[assumption]

\newenvironment{assumptionp}[1]{
  \renewcommand\theassumptionalt{#1}
  \assumptionalt
}{\endassumptionalt}

\begin{document}
\title{Revisiting Panel Data Discrete Choice Models with Lagged Dependent Variables\footnote{This paper uses unit record data from Household, Income and Labour Dynamics in Australia Survey (HILDA) conducted by the Australian Government Department of Social Services (DSS). The findings and views reported in this paper, however, are those of the author[s] and should not be attributed to the Australian Government, DSS, or any of DSS’ contractors or partners. DOI: 10.26193/YP7MNU.}\vspace{0.25cm}}

\author[1]{Christopher R. Dobronyi\footnote{E-mail address: \href{dobronyi@google.com}{dobronyi@google.com}.}}
\author[2]{Fu Ouyang\footnote{E-mail address: \href{f.ouyang@uq.edu.au}{f.ouyang@uq.edu.au}. }}
\author[3]{Thomas Tao Yang\footnote{E-mail address: \href{tao.yang@anu.edu.au}{tao.yang@anu.edu.au}.}}
\affil[1]{Google}
\affil[2]{School of Economics, University of Queensland}
\affil[3]{Research School of Economics, Australian National University}
\date{\today}
\maketitle

%\bigskip
\begin{abstract}
\noindent This paper revisits the identification and estimation of a class of semiparametric (distribution-free) panel data binary choice models with lagged dependent variables, exogenous covariates, and entity fixed effects. We provide a novel identification strategy, using an ``identification at infinity'' argument. In contrast with the celebrated \cite{honore-k}, our method permits time trends of any form and does not suffer from the ``curse of dimensionality''. We propose an easily implementable conditional maximum score estimator. The asymptotic properties of the proposed estimator are fully characterized. A small-scale Monte Carlo study demonstrates that our approach performs satisfactorily in finite samples. We illustrate the usefulness of our method by presenting an empirical application to enrollment in private hospital insurance using the Household, Income and Labor Dynamics in Australia (HILDA) Survey data.
\end{abstract}

\noindent%
\vfill

\newpage
 
\section{Introduction}

In this paper, we propose new identification and estimation methods for
panel data binary choice models with fixed effects and ``dynamics'' (lagged dependent variables). Specifically,
suppose that there are $n$ individuals and $T+1$ time periods, $\{0,1,...,T\}
$. In each time period $t\in \{1,...,T\}$, each individual $i$ makes a
choice $y_{it}\in \{0,1\}$ according to the following latent utility model:
\begin{equation}
y_{it}=\mathds{1}\{\alpha _{i}+\gamma y_{it-1}+x_{it}^{\prime }\beta +\varpi
z_{it}\geq \epsilon _{it}\},  \label{model}
\end{equation}%
where $\alpha _{i}$ is an entity fixed effect absorbing all relevant
time-invariant factors, $y_{it-1}$ is the lagged dependent variable, $%
(x_{it},z_{it})$ is a $(p+1)$-vector of time-varying covariates, and $%
\epsilon _{it}$ is an idiosyncratic error term. We separate $z_{it}$ from other covariates because, as it will be clear in the next section, we assign it a crucial role in the identification at infinity. In panel data literature, $%
y_{it-1}$ is often called ``state
dependence'', and $\alpha _{i}$ is referred to as
``unobserved heterogeneity'' or
``spurious'' state dependence (see \cite%
{Heckman1981b, Heckman1981a}). In model (\ref{model}), $%
(y_{it},x_{it},z_{it})$ along with the ``initial
status'' $y_{i0}$ are observed in the data, whereas $\alpha
_{i}$ and $\epsilon _{it}$ are not observable to the econometrician. Note that we do not specify model (\ref{model}) in the initial period 0. This paper studies the identification and
estimation of the preference parameter $\theta :=(\gamma ,\beta ,\varpi )\in
\mathbb{R}^{p+2}$ in ``short'' panel
settings, i.e., $n\rightarrow \infty $ and $T<\infty $. %Since one can only identify $\theta $ up to scale$,$ we normalize $\left\Vert \theta\right\Vert =1,$ where $\left\Vert \cdot \right\Vert $ is the Euclidean norm.

In line with the vast literature on panel data models with entity fixed effects, we do not impose any parametric restrictions on the distribution of $\alpha_i$ conditional on the initial choice $y_{i0}$ and observed covariates in model (\ref{model}). The prevalent methods for such models assume that $\epsilon_{it}$ are independently and identically distributed (i.i.d.) with a logistic distribution. \citet{ArellanoHonore2001}, \citet{honore2021identification}, and \citet{Hsiaobook} review various conditional likelihood approaches based on these parametric assumptions on $\epsilon _{it}$. Recent advances in the literature focus on constructing moment conditions for variants of dynamic Logit models. Representative works include \citet{honore-w}, \citet{dobronyi-gu-kim}, \citet{kitazawa2022transformations}, and \citet{dano2023transition}, among others.

Without making distributional assumptions on $\epsilon _{it}$, \citet{manski-87} establishes the semiparametric identification of model (\ref{model}) that includes covariates, but not $y_{it-1}$. \cite{honore-k} extend this approach to include both $y_{it-1}$ and covariates in the model, showing that model (\ref{model}) with $T\geq 3$ can be identified under exogeneity and serial dependence assumptions stronger than those in \citet{manski-87}. However, their proposed estimator requires element-by-element matching of observed covariates over time, which rules out covariates with non-overlapping supports over time (e.g., time trend or dummies) and has a convergence rate decreasing in the dimension of the covariate space. \cite{OuyangYang2024binary} demonstrates that this curse of dimensionality can be mitigated by imposing certain serial dependence conditions on the covariates and by observing an extra time period. To highlight the novelty and contributions of this paper, we present a thorough comparison of our method with \citet{honore-k} and \cite{OuyangYang2024binary} in Appendices \ref{appendix0_1} and \ref{appendix0_2}, respectively.

There are alternative semiparametric and nonparametric approaches to model (\ref{model}). \citet{hl} demonstrate that model (\ref{model}) can be point identified if $z_{it}$ satisfies certain exclusion restrictions. More recently, \citet{ChenEtal2019} revisit this method and discuss the sufficient conditions for such exclusion restrictions. \citet{Williams2019} studies the nonparametric identification of dynamic binary choice models that satisfy certain exclusion restrictions. In the absence of excluded regressors, \citet{arist} establishes informative partial identification of model (\ref{model}) under weak conditions. \citet{KhanEtal2020} offer a partial identification result under even milder restrictions and prove that point identification is attainable in many interesting scenarios.

\begin{comment}
There are alternative semiparametric and nonparametric approaches to model (\ref{model}).\footnote{
Our literature review here focuses on fixed effects methods. It is well known that model (\ref{model}) can be estimated by the random effects or correlated random coefficients approach. See, e.g., \cite{ArellanoCarrasco2003}, \cite{Wooldridge2005}, and \cite{honore-tamer}, among others. These approaches often allow the econometrician to calculate choice probabilities and marginal effects in addition to $\theta$ at the cost of imposing restrictions on the distribution of $y_{i0}$ (initial condition problem) and the statistical relation between observed covariates and $\alpha_{i}$.} \citet{hl} demonstrate that the model, either with or without $y_{it-1}$, can be point identified if $z_{it}$ satisfies certain exclusion restrictions, i.e., the existence of an excluded regressor, regardless of the endogeneity of the covariates. More recently, \cite{ChenEtal2019} revisit this method and discusses the sufficient conditions for the exclusion restriction to be satisfied in dynamic settings. See also \cite{Williams2019} for nonparametric identification of dynamic binary choice models. In the absence of excluded regressors, \citet{arist} establishes informative partial identification of model (\ref{model}) under weak conditions. \cite{KhanEtal2020} offer a unique partial identification result under even milder restrictions, proving that point identification is attainable in many interesting scenarios.
\end{comment}

This paper revisits the distribution-free identification and estimation of model (\ref{model}). We show that the overlapping support restrictions required by \cite{honore-k} can be removed if $z_{it}$ is a free-varying covariate with full support. Our identification employs an ``identification at infinity'' strategy, first introduced in \cite{chamberlain1986asymptotic} and \cite{heckman-90}, and then applied in more recent work such as \cite{tamer2003incomplete}, \cite{bajari2010identification}, \cite{wan2014semiparametric}, and \cite{ouyang2020semiparametric}, among others. The combination of this strategy and \citeauthor{manski-87}'s (\citeyear{manski-87}) insight yields an estimator in the spirit of \citeauthor{honore-k}'s (\citeyear{honore-k}) conditional maximum score (MS) estimator, but without the need to match observed covariates over time. As a result, our estimator can accommodate flexible time effects and escape from the curse of dimensionality, in contrast to \cite{honore-k}. Through extending \cite{KimPollard1990} and \cite{SeoOtsu2018}, we demonstrate that our estimator converges at a rate slower than cube-root-$n$, is independent of the number of observed covariates, and has a non-standard limiting distribution. The asymptotics share similarities with those in \cite{honore-k} and \cite{OuyangYang2024binary}, with an important difference: the rate of convergence for our estimator depends on unknown factors, while the convergence rates of their estimators are known. We evaluate the finite-sample performance and implementability of our proposed estimator using both simulated and real-world data.

The rest of this paper is organized as follows. Section \ref{sec:identification} establishes the identification of $\theta$, which serves as the basis for the MS estimator presented in Section \ref{Sec:Estimation}. We then derive asymptotic properties of the proposed estimator in Section \ref{Sec:Asymptotics}. Results of Monte Carlo experiments are reported in Section \ref{sec_simulation}. We present an empirical illustration using the HILDA data in Section \ref{sec_application}. Finally, Section \ref{Sec:Conlcusions} concludes the paper with a brief discussion on possible future research directions. All proofs, supplementary discussions, and additional simulation results are included in the Supplementary Appendix.

For ease of reference, we list the notations maintained throughout this
paper here.

\begin{notation}
We reserve letter $i\in \{1,...,n\}$ for indexing individuals, and letter $%
t\in \{1,...,T\}$ for indexing time periods. $\mathbb{R}^{k}$ is a $k$%
-dimensional Euclidean space equipped with the Euclidean norm $\Vert \cdot
\Vert $ and $\mathbb{R}_{+}^{k}:=\{x\in \mathbb{R}^{k}|x>0\}$. We use $%
P(\cdot )$ and $\mathbb{E}[\cdot ]$ to denote probability and expectation,
respectively. $\mathds{1}\{\cdot \}$ is an indicator function that equals
one when the event in the brackets occurs, and zero otherwise. Following a
substantial panel literature, we use the notation $\xi ^{t}$ to denote the
history of $\xi $ from period $1$ to period $t.$ For example, $x^{t}:=\left( x_{1},...,x_{t}\right) $ and  $y^{t}:=\left(y_{1},...,y_{t}\right)$. We use $\setminus$ to denote set difference. For example, $\left( x_{1},x_{2},...,x_{t}\right)\setminus x_{1}:= \left(x_{2},...,x_{t}\right) $.
For two random vectors, $u$ and $v$, the notation $u\overset{d}{=}v|\cdot $
means that $u$ and $v$ have identical distribution, conditional on $\cdot $,
and $u\perp v|\cdot $ means that $u$ and $u$ are independent conditional on $%
\cdot $. We use $\overset{p}{\rightarrow }$ and $\overset{d}{\rightarrow }$
to denote convergences in probability and in distribution, respectively. For
any (random) positive sequences $\{a_{n}\}$ and $\{b_{n}\}$, $a_{n}=O(b_{n})$
($O_{p}(b_{n})$) means that $a_{n}/b_{n}$ is bounded (in probability), $%
a_{n}=o(b_{n})$ ($o_{p}(b_{n})$) means that $a_{n}/b_{n}\rightarrow 0$ ($%
a_{n}/b_{n}\overset{p}{\rightarrow }0$). $a_{n}\lesssim b_{n}$ and $a_{n}\asymp b_{n}$ respectively mean
that there exist two constants $0<c_{1}\leq c_{2}<\infty$ such that
$c_{1}a_{n}\leq b_{n}$ and $c_{1}a_n\leq b_{n}\leq c_{2}a_{n}$. $a_{n}\ll b_{n}$ means $a_n = o(b_n)$.
\end{notation}

\section{Identification}\label{sec:identification}
Suppose in model (\ref{model}) $\epsilon_{it}$'s are i.i.d. over
time and independent of observed covariates $(x_{i}^{T},z_{i}^{T})$
and the initial choice $y_{i0}$ conditioning on the fixed effect
$\alpha_{i}$. The conditional probability of $y_{it}=1$ is equal
to:
\begin{equation}
P\left(y_{it}=1|\alpha_{i},y_{i}^{t-1},x_{i}^{T},z_{i}^{T}\right)=F_{\epsilon|\alpha}\left(\alpha_{i}+\gamma y_{it-1}+x_{it}^{\prime}\beta+\varpi z_{it}\right) \nonumber
\end{equation}
for each $i=1,\dots,n$ and $t=1,\dots,T$, where $F_{\epsilon|\alpha}(\cdot)$
denotes the cumulative distribution function (CDF) of $\epsilon_{it}$
conditional on $\alpha_{i}$. Consequently, the probability of observing
the choice history $y_{i}^{T}$ conditional on $(\alpha_{i},y_{i0},x_{i}^{T},z_{i}^{T})$
is expressed as:
\begin{align}
 & P\left(y_{i}^{T}|\alpha_{i},y_{i0},x_{i}^{T},z_{i}^{T}\right)\nonumber\\
= & P\left(y_{i2},y_{i3},...,y_{iT}|\alpha_{i},y_{i0},y_{i1},x_{i}^{T},z_{i}^{T}\right)P\left(y_{i1}|\alpha_{i},y_{i0},x_{i}^{T},z_{i}^{T}\right)\nonumber \\
= & P\left(y_{i3},...,y_{iT}|\alpha_{i},y_{i0},y_{i1},y_{i2},x_{i}^{T},z_{i}^{T}\right)P\left(y_{i2}|\alpha_{i},y_{i0},y_{i1},x_{i}^{T},z_{i}^{T}\right)P\left(y_{i1}|\alpha_{i},y_{i0},x_{i}^{T},z_{i}^{T}\right)\nonumber \\
= & \cdots=\prod_{t=1}^{T}P\left(y_{it}|\alpha_{i},y_{i}^{t-1},x_{i}^{T},z_{i}^{T}\right)\nonumber \\
= & \prod_{t=1}^{T}F_{\epsilon|\alpha}\left(\alpha_{i}+\gamma y_{it-1}+x_{it}^{\prime}\beta+\varpi z_{it}\right)^{y_{it}}\left[1-F_{\epsilon|\alpha}(\alpha_{i}+\gamma y_{it-1}+x_{it}^{\prime}\beta+\varpi z_{it})\right]^{1-y_{it}} \nonumber
\end{align}
for each individual $i=1,\dots,n$.

In what follows, we will restrict the illustration of our approach
to model (\ref{model}) with $T=3$ and $\varpi>0$ to ease the exposition.
The condition $\varpi>0$ implies that we must know that the covariate
$z_{it}$ is included in the model and that it has a positive effect
on the choice probability of $y_{it}=1$. Applying our method to longer panels is straightforward, and the case with $\varpi<0$ is
symmetric. In addition, we will omit the subscript $i$ in our notation
whenever the context makes clear that all variables pertain to each
individual. Finally, we assume a balanced panel for simplicity. Our methods
are applicable to models with unbalanced panels, provided the
unbalancedness is not due to endogenous attrition.
 
Consider two choice histories
\[
C =\{y_{0}=d_{0},y_{1}=0,y_{2}=d_{2},y_{3}=1\} \textrm{ and } D =\{y_{0}=d_{0},y_{1}=1,y_{2}=d_{2},y_{3}=0\},
\]
%\begin{align*}
%C & =\{y_{0}=d_{0},y_{1}=0,y_{2}=d_{2},y_{3}=1\},\\
%D & =\{y_{0}=d_{0},y_{1}=1,y_{2}=d_{2},y_{3}=0\},
%\end{align*}
where $d_{0},d_{2}\in\{0,1\}$. The conditional probability of the
choice history $C$ is equal to
\begin{align*}
 & P(C|\alpha,y_{0}=d_{0},x^{T},z^{T})\\
= & p_{0}(\alpha,x^{T},z^{T})^{d_{0}}(1-p_{0}(\alpha,x^{T},z^{T}))^{1-d_{0}}(1-F_{\epsilon|\alpha}(\alpha+\gamma d_{0}+x_{1}^{\prime}\beta+\varpi z_{1}))\\
 & \times F_{\epsilon|\alpha}(\alpha+x_{2}^{\prime}\beta+\varpi z_{2})^{d_{2}}(1-F_{\epsilon|\alpha}(\alpha+x_{2}^{\prime}\beta+\varpi z_{2}))^{1-d_{2}}F_{\epsilon|\alpha}(\alpha+\gamma d_{2}+x_{3}^{\prime}\beta+\varpi z_{3}),
\end{align*}
where $p_{0}(\alpha,x^{T},z^{T})$ denotes the conditional probability
of $y_{0}=1$. In a similar fashion,
\begin{align*}
 & P(D|\alpha,y_{0}=d_{0},x^{T},z^{T})\\
= & p_{0}(\alpha,x^{T},z^{T})^{d_{0}}(1-p_{0}(\alpha,x^{T},z^{T}))^{1-d_{0}}F_{\epsilon|\alpha}(\alpha+\gamma d_{0}+x_{1}^{\prime}\beta+\varpi z_{1})\\
 & \times F_{\epsilon|\alpha}(\alpha+\gamma+x_{2}^{\prime}\beta+\varpi z_{2})^{d_{2}}(1-F_{\epsilon|\alpha}(\alpha+\gamma+x_{2}^{\prime}\beta+\varpi z_{2}))^{1-d_{2}}(1-F_{\epsilon|\alpha}(\alpha+\gamma d_{2}+x_{3}^{\prime}\beta+\varpi z_{3})).
\end{align*}

Here, we take $d_{2}=1$ to illustrate, and the case with $d_{2}=0$ is
symmetric. Suppose the support of $z_{2}$ is unbounded above. Then, for $z_2>\sigma$, where $\sigma$ is a sufficiently large positive number, these probabilities satisfy
\begin{align}
& \frac{P(C|\alpha,y_{0}=d_{0},x^{T},z^{T})}{P(D|\alpha,y_{0}=d_{0},x^{T},z^{T})}  \nonumber\\
\approx & \frac{1-F_{\epsilon|\alpha}(\alpha+\gamma d_{0}+x_{1}^{\prime}\beta+\varpi z_{1})}{1-F_{\epsilon|\alpha}(\alpha+\gamma d_{2}+x_{3}^{\prime}\beta+\varpi z_{3})}  \times\frac{F_{\epsilon|\alpha}(\alpha+\gamma d_{2}+x_{3}^{\prime}\beta+\varpi z_{3})}{F_{\epsilon|\alpha}(\alpha+\gamma d_{0}+x_{1}^{\prime}\beta+\varpi z_{1})}. \label{eq:prob_ratio_1}
\end{align}
The idea of the above is to make negligible the effect of $y_{1}$
on $y_{2}$, i.e., $F_{\epsilon|\alpha}(\alpha+\gamma+x_{2}'\beta+\varpi z_{2})\approx F_{\epsilon|\alpha}(\alpha+x_{2}'\beta+\varpi z_{2})$
($\approx1$), by letting $z_2$ be sufficiently large.

Suppose $F_{\epsilon|\alpha}(\cdot)$ is strictly increasing. Then,
when $d_{2}=1$, equation (\ref{eq:prob_ratio_1}) implies that
\begin{align}
& \text{sgn}\left\{ P(C|\alpha,y_{0}=d_{0},x^{T},z^{T})-P(D|\alpha,y_{0}=d_{0},x^{T},z^{T})\right\} \nonumber \\
=& \text{sgn}\left\{ \gamma(d_{2}-d_{0})+(x_{3}-x_{1})^{\prime}\beta+\varpi(z_{3}-z_{1})\right\} \label{eq:iden_ineq1}
\end{align}
holds for $z_2>\sigma$ as $\sigma\rightarrow +\infty$, where $\text{sgn}\{\cdot\}$ is the
\textit{sign function}, which is equal to 1 if the expression inside
the brackets is strictly positive, to 0 if the expression inside the
brackets is zero, and to $-1$ if the expression inside the brackets
is strictly negative.

Equation (\ref{eq:iden_ineq1}) reveals that when $z_{2}$ is sufficiently large and $d_{2}=1$, the likelihood of observing event $C$ exceeds that of observing event $D$ if and only if $\gamma(d_{2}-d_{0})+(x_{3}-x_{1})^{\prime}\beta+\varpi(z_{3}-z_{1})>0$. In other words, the sign of $\gamma(d_{2}-d_{0})+(x_{3}-x_{1})^{\prime}\beta+\varpi(z_{3}-z_{1})$ determines the rank order of the conditional probabilities of events $C$ and $D$. Our focus on the subsample with $y_{3}\neq y_{1}$ aligns with \cite{manski-87} in forming his MS estimator. The distinction lies in our additional conditioning event of $z_{2}$ being large. It is worth noting that the same identification equation holds true when $-z_{2}$ is sufficiently large and $d_{2}=0$.

A natural way to build a population objective function based on equation
(\ref{eq:iden_ineq1}) is to define
\begin{align}
\bar{Q}_{1}(r,b,w):=\lim_{\sigma\rightarrow+\infty}\mathbb{E} & \left[\left(P(C|\alpha,y_{0}=d_{0},x^{T},z^{T})-P(D|\alpha,y_{0}=d_{0},x^{T},z^{T})\right)\right.\nonumber \\
 & \left.\times\text{sgn}\left(r(d_{2}-d_{0})+(x_{3}-x_{1})^{\prime}b+w(z_{3}-z_{1})\right)|z_{2}>\sigma\right]\label{eq:pop_obj}
\end{align}
with $w>0$\ for $d_{2}=1$. By a symmetric argument, define
\begin{align}
\bar{Q}_{2}(r,b,w):=\lim_{\sigma\rightarrow+\infty}\mathbb{E} & \left[\left(P(C|\alpha,y_{0}=d_{0},x^{T},z^{T})-P(D|\alpha,y_{0}=d_{0},x^{T},z^{T})\right)\right.\nonumber \\
 & \left.\times\text{sgn}\left(r(d_{2}-d_{0})+(x_{3}-x_{1})^{\prime}b+w(z_{3}-z_{1})\right)|z_{2}<-\sigma\right]\label{eq:pop_obj_2}
\end{align}
with $w>0$ for $d_{2}=0$. Note that equation (\ref{eq:iden_ineq1}) implies that $\bar{Q}_{1}(\gamma,\beta,\varpi)\geq\bar{Q}_{1}(r,b,w)$
and $\bar{Q}_{2}(\gamma,\beta,\varpi)\geq\bar{Q}_{2}(r,b,w)$ for
all $(r,b,w)\neq(\gamma,\beta,\varpi)$. Establishing that $\theta:=(\gamma,\beta,\varpi)$ is the unique maximum of either objective function (\ref{eq:pop_obj}) or (\ref{eq:pop_obj_2}) would confirm the point identification of these coefficients. The following conditions are sufficient for this.

\begin{assumptionp}{A} For all $\alpha$ and $s,t\in\mathcal{T}:=\{1,2,3\}$
($T=3$), the following conditions hold:\par \begin{enumerate}\par

\item[A1] (i) $\epsilon^{T}\perp(x^{T},z^{T},y_{0})|\alpha$,
(ii) $\epsilon_{s}\perp\epsilon_{t}|\alpha$, (iii) $\epsilon_{s}\overset{d}{=}\epsilon_{t}|\alpha$,
and (iv) conditional on $\alpha$, the CDF of $\epsilon_{t}$ is absolutely
continuous with support $\mathbb{R}$.\par

\item[A2] $z_{2}$ has unbounded support conditional on $(\alpha,y_{0},x^{T},z_{1},z_{3})$. \par

\item[A3] One of the elements in $\left(x_{3}-x_1,z_{3}-z_1\right)$,
denoted as $\xi_{31}$, has a bounded Lebesgue density that is positive
almost everywhere (a.e.) on $\mathbb{R}$ conditional on $(\alpha,x_{3}-x_1,z_{3}-z_1)\left\backslash \xi_{31}\right.$
and $\{z_{2}>\sigma\}\cup\{z_{2}<-\sigma\}$ as $\sigma\rightarrow +\infty$.
Moreover, the coefficient before $\xi_{31}$ is non-zero. \par

\item[A4] As $\sigma\rightarrow +\infty$, (i) the support $\mathcal{S}$
of $(y_{2}-y_{0},x_{3}-x_{1},z_{3}-z_{1})$ conditional on $\{z_{2}>\sigma\}$
or $\{z_{2}<-\sigma\}$ is not contained in any proper linear subspace
of $\mathbb{R}^{p+2}$, and (ii) the joint probability density function
(PDF) of $(x_{3}-x_{1},z_{3}-z_{1})$ conditional on $\{z_{2}>\sigma\}$
or $\{z_{2}<-\sigma\}$ is non-degenerate and uniformly bounded. \par

\item[A5] Let $\Theta$ be the set $\{\vartheta:=(r,b,w)\in\mathbb{R}^{p+2}\ |\ \Vert\vartheta\Vert=1,w>0\}$. $\theta$ is an interior
point of $\Theta$. \end{enumerate} \end{assumptionp}

Assumptions A1(i)--(iii) place the same restrictions on the joint
distribution of $(\alpha,\epsilon^{T},x^{T},z^{T})$ as \citet{honore-k},
which implies that the unobserved heterogeneity (entity fixed effects)
$\alpha$ picks up both the autocorrelation in the unobservables and
the dependence between explanatory variables and unobservables. As
a result, $\epsilon_{t}$ is independent of $(x^{T},z^{T},y^{t-1})$
conditional on $\alpha$ for all $t\in\mathcal{T}$. Assumption A1(iv)
is a regularity condition to guarantee that any possible sequence
of $y^{T}$ has a positive probability to occur.

Assumption A2 is a pivotal assumption that enables the ``identification at infinity'' approach, and when combined with Assumption A1, it establishes the identification equation
(\ref{eq:iden_ineq1}). It is clear from the derivation of equation (\ref{eq:iden_ineq1}) that relaxing this assumption may require additional restrictions on the parameter space $\Theta$,
the support of $x_{2}$, and the distribution of $(\epsilon,\alpha)$.

The support and continuity restrictions on
$\xi_{31}$ imposed by Assumption A3 are common for the family of
MS-type estimators, which are required to achieve the point
identification instead of a set identification. See, e.g., \citet{manski-75,manski-85,manski-87},
\citet{horowitz}, \citet{honore-k}, \citet{Fox2007}, \citet{ShiEtal2018},
\citet{YanYoo2019}, and \citet{khan2021inference}, among others.
Given the significance of Assumptions A2 and A3 in both our theoretical results and empirical application, we provide further discussion on them in Appendix \ref{appendix0_3}.

Assumption A4(i) is a familiar full-rank condition. Note that Assumptions
A3 and A4(ii) require $(x_{3}-x_{1},z_{3}-z_{1})$ to have sufficient
variation conditional on $\alpha$ and event $\{z_{2}>\sigma\}$ or
$\{z_{2}<-\sigma\}$. Assumption A5 applies the scale normalization
and restricts the search of $\theta$ in a compact set, which also
facilitates the asymptotic analysis of our estimator proposed in the
next section.

Additionally, we compare the key identification assumptions imposed in \citet{honore-k} and \citet{OuyangYang2024binary}, along with other aspects, with our method in Appendices \ref{appendix0_1} and \ref{appendix0_2}, respectively.

\begin{remark} \label{remark_A5}
Assumption A5 applies scale normalization by restricting $\vartheta$ to lie on a unit sphere. Alternatively, one can normalize one nonzero element of $\vartheta$, such as $w$ in this paper, to be 1. Following this convention, we express the parameter space as:\par\medskip

\noindent \textbf{Assumption A}5': $\Theta:=\{\vartheta:=(r,b,1)\in\mathbb{R}^{p+2}\}\cap\Xi$,
where $\Xi\subset\mathbb{R}^{p+2}$ is a compact set. $\theta$ is
an interior point of $\Theta$. \par\medskip
\noindent These two methods for scale normalization are essentially equivalent when $w \neq 0$. Therefore, researchers often use either of these methods based on their convenience in exposition or derivation. For instance, \cite{ShiEtal2018} use both methods in different sections. In this paper, we insist on Assumption A5 in Appendices \ref{appendixA} and \ref{appendixB}, as we derive the asymptotic properties of our estimator. This is because the two primary references for doing this, \cite{KimPollard1990} and \cite{SeoOtsu2018}, both normalize the parameter space to a unit sphere. Following the same convention facilitates our use of their established asymptotic theory and makes it easier for interested readers to review our proofs. When we apply our method to simulation studies and an empirical application in Sections \ref{sec_simulation} and \ref{sec_application}, we switch to the normalization defined in Assumption A5’, which reduces one parameter to estimate and avoids imposing restrictions on the optimization algorithm. We thank one anonymous referee for highlighting this point.
\end{remark}

Our identification results are stated in the following theorem and
the proof of which is provided in Appendix \ref{appendixA}.

\begin{theorem} \label{T:identify} Suppose Assumption A holds. Then,
$\theta$ is identified. \end{theorem}

\section{Estimation}\label{Sec:Estimation}

Applying the analogy principle, the population
objective functions (\ref{eq:pop_obj}) and (\ref{eq:pop_obj_2}) translate into MS estimation procedures (\cite{manski-75,manski-85,manski-87}).

Assume a random sample of $n$ observations is drawn from model (\ref{model}) that satisfies Assumption A.
Let $\vartheta :=(r,b,w)\in \mathbb{R}^{p+2}$ and $%
\sigma _{n}\rightarrow \infty $ as $n\rightarrow \infty $. When the support
of $z_{i2}$ is unbounded above, we propose the MS estimator $\hat{\theta}_{n}
$ of $\theta $ maximizing the following objective function over the
parameter space $\Theta $:
\begin{equation}
Q_{n1}(\vartheta ):=\frac{1}{n}\sum_{i=1}^{n}y_{i2}(y_{i3}-y_{i1})\cdot %
\mathds{1}\{z_{i2}>\sigma _{n}\}\cdot \mathds{1}%
\{r(y_{i2}-y_{i0})+(x_{i3}-x_{i1})^{\prime }b+w(z_{i3}-z_{i1})>0\}.
\label{eq:Qn1}
\end{equation}%
When the support of $z_{i2}$ is unbounded below, one can instead define $%
\hat{\theta}_{n}$ with objective function
\begin{equation}
Q_{n2}(\vartheta ):=\frac{1}{n}\sum_{i=1}^{n}(1-y_{i2})(y_{i3}-y_{i1})\cdot %
\mathds{1}\{z_{i2}<-\sigma _{n}\}\cdot \mathds{1}%
\{r(y_{i2}-y_{i0})+(x_{i3}-x_{i1})^{\prime }b+w(z_{i3}-z_{i1})>0\}.
\label{eq:Qn2}
\end{equation}%
If the support of $z_{i2}$ is unbounded both above and below, the objective
function can be a combination of (\ref{eq:Qn1}) and (\ref{eq:Qn2}) such as
\begin{equation}
Q_{n}(\vartheta ):=Q_{n1}(\vartheta )+Q_{n2}(\vartheta ).  \label{eq:Qn}
\end{equation}%
Note that (3.3) puts the same weight on $Q_{n1}(\vartheta)$ and $Q_{n2}(\vartheta)$,
which is a generic choice and probably not optimal in specific applications.
In some cases, it might be preferable to put more weight on one side
if additional information, such as restrictions on the error distribution,
suggests that the identification at infinity is more effective on
that side, especially if $z_{2}$ has a relatively heavier tail against
the error term.
Since $\text{sgn}(u)=2\cdot \mathds{1}\{u>0\}-1$ almost surely for any
continuous variable $u$, objective functions (\ref{eq:Qn1}) and (\ref{eq:Qn2}) are sample
analogues to monotone transformations of population functions (\ref{eq:pop_obj}) and (\ref%
{eq:pop_obj_2}), respectively.

It is clear from expressions (\ref{eq:Qn1}) and (\ref{eq:Qn2}) that the
effective sample size for the estimator $\hat{\theta}_{n}$ is controlled by
the tuning parameter $\sigma_{n}$, and being similar to %
\citeauthor{manski-87}'s (\citeyear{manski-87}) and \citeauthor{honore-k}'s (%
\citeyear{honore-k}) estimators, only ``switchers'' who change
choices in periods 1 and 3 are used in the estimation. Besides, estimating (identifying) $\gamma$ relies on the variation in $y_{i2} - y_{i0}$, which means that we need some observations with $y_{i2}\neq y_{i0}$ and some with $y_{i2}=y_{i0}$.

Our proposed estimator $\hat{\theta}_n$ has two advantages, compared with
\citeauthor{honore-k}'s (\citeyear{honore-k}) estimator: First, the
estimation only needs to condition on a single univariate covariate, rather
than a vector of covariates, and hence it does not encounter the curse of
dimensionality. This property makes the procedure proposed above more
practical when the number of covariates is large. More importantly, our
estimator does not require matching $(x_t,z_t)$ in different periods.
Consequently, it allows covariates with non-overlapping support over time,
such as age, time trends, time dummy variables, etc.

\section{Asymptotic Properties}

\label{Sec:Asymptotics}

\subsection{Consistency}

This section establishes the asymptotic properties of the MS estimator
proposed in Section \ref{Sec:Estimation}. Given that objective functions (\ref%
{eq:Qn1}) and (\ref{eq:Qn2}) are symmetric, it suffices to only investigate
the estimator $\hat{\theta}_{n}$ obtained from maximizing objective function
(\ref{eq:Qn1}) requiring the support of $z_{2}$ to be unbounded above. The
derivation for $\hat{\theta}_{n}$ associated with objective functions (\ref%
{eq:Qn2}) or (\ref{eq:Qn}) is analogous. Additionally, for the sake of simplicity, we focus on the case where $\xi_{31}=z_{31}$ in Assumption A3.

To ensure the consistency of $\hat{\theta}_n$, we need additional technical
conditions.

\begin{assumptionp}{B}
For all $t\in \mathcal{T}$, the following conditions hold:

\begin{enumerate}
\item[B1] The data $\{y_{i0},y_{i}^{T},x_{i}^{T},z_{i}^{T}\}_{i=1}^{n}$ are
i.i.d. across $i$.

\item[B2] $\sigma_{n}$ is a sequence of positive numbers such that as $%
n\rightarrow\infty$: (i) $\sigma_{n}\rightarrow\infty$, and (ii) $%
nP(z_{2}>\sigma_{n})/\log n\rightarrow\infty$.

\item[B3] Let $\Lambda (\vartheta ):=y_{2}(y_{3}-y_{1})\cdot \mathds{1}%
\{r(y_{2}-y_{0})+(x_{3}-x_{1})^{\prime }b+w(z_{3}-z_{1})>0\}$ for $\vartheta
\in \Theta $. Then (i) $\lim_{\sigma \rightarrow +\infty }\mathbb{E}[\Lambda
(\vartheta )|z_{2}>\sigma ]$ exists for all $\vartheta \in \Theta $, and
(ii) there exists an absolute constant $L$ such that
\begin{equation}
|\mathbb{E}[\Lambda (\vartheta _{1})|z_{2}>\sigma ]-\mathbb{E}[\Lambda
(\vartheta _{2})|z_{2}>\sigma ]|\leq L\Vert \vartheta _{1}-\vartheta
_{2}\Vert  \label{assumption_lip}
\end{equation}%
holds for all $\vartheta _{1},\vartheta _{2}\in \Theta $ and $\sigma >0$.
\end{enumerate}
\end{assumptionp}

Assumptions B2 imposes mild restrictions on the tuning parameter $\sigma _{n}
$. It is worth noting that Assumption B2(ii) indicates that the choice of $%
\sigma_{n}$ depends on the tail behavior of the distribution of $z_{2}$.
For example, if $z_{2}$ has a sub-exponential right tail with $P(z_{2} > \sigma_{n}) \asymp e^{-\sigma_{n}}$, then any $\sigma_n$ satisfying $1 \ll \sigma_n \leq (1-\varepsilon) \log(n)$, e.g., $\sigma_n = \log\log(n/\log n)$, meets Assumption B2(ii), for some $\varepsilon\in(1/4,1)$. However, when the distribution of $z_{2}$ has a
(too) thin right tail $P(z_{2}>\sigma _{n})\asymp e^{-e^{\sigma _{n}}}$, $\sigma _{n}=\log\log(n/\log n)$ gives $nP(z_{2}>\sigma _{n})/\log n=O(1)$, violating
Assumption B2(ii). Notably, $nP(z_{2}>\sigma _{n})$ essentially controls the
``effective sample size'' for our proposed
procedure. As demonstrated in Theorem \ref{T:limiting_dist}, the tail behavior of the distribution of $z_{2}$ and the
choice of $\sigma _{n}$ jointly determine the convergence rate of the
proposed estimator $\hat{\theta}_{n}$. Assumption B3(ii) is a Lipschitz
condition essential for proving the uniform convergence of the objective
function (\ref{eq:Qn1}) to its population analogue. We provide a set of more concrete sufficient conditions for it in Appendix \ref{appendix0_34}.

The theorem below states that the proposed procedure described in (\ref%
{eq:Qn1})--(\ref{eq:Qn}) %\footnote{Some straightforward adjustments to Assumption B are necessary for establishing similar results for $\hat{\theta}_n$ associated with (\ref{eq:Qn2}) or (\ref{eq:Qn}).}
gives a consistent estimator of $\theta$, whose
proof is left to Appendix \ref{appendixA}.

\begin{theorem}
\label{T:consistency} Suppose Assumptions A and B hold. Then, $\hat{\theta}_n%
\overset{p}{\rightarrow}\theta$.
\end{theorem}

\subsection{Asymptotic Distribution}\label{Sec:asymp_dist}
We proceed to study the asymptotic distribution of
the estimator $\hat{\theta}_n$. Before presenting additional technical
conditions and the main results, we introduce some new notation to
facilitate exposition:

\begin{itemize}
\item[-] $h_{n}:=P(z_{2}>\sigma_{n}|y_{2}=1)$.%\footnote{Note that under Assumptions A1 and A3, $P(z_{2}>\sigma_{n})\asymp P(z_{2}>\sigma_{n}|y^{T})$, and so $h_n$ satisfies Assumption B2(ii).}

\item[-] For generic vectors $\xi_t$ and $\xi_s$, denote $%
\xi_{ts}=\xi_t-\xi_s$.

\item[-] $\chi:=(y_{0},y^{T},x^{T},z^{T})$ and $%
\bar{\chi}:=(y_{20},x_{31},z_{31})$. %a sub-vector of $\chi$.

\item[-] $u(\vartheta ):=\mathds{1}\{ry_{20}+x_{31}^{\prime }b+wz_{31}>0\}$,
thus, $u(\vartheta )=\mathds{1}\{\bar{\chi}^{\prime }\vartheta>0\} $.

\item[-] $\bar{q}_{n1,\vartheta }(\bar{\chi}):=\mathbb{E}\left[ y_{31}\left(
u(\vartheta )-u(\theta )\right) |z_{2}>\sigma _{n},y_{2}=1,\bar{\chi}\right] $%
, and $\bar{q}_{1\vartheta }^{+}(\bar{\chi}):=\lim_{n\rightarrow \infty }\bar{q%
}_{n1,\vartheta }(\bar{\chi})$.

%\item[-] For discrete $y$ and any generic differentiable function $f$, the $%
%f^{\prime }\left( y\right) $ is defined as $\left. \frac{df\left( x\right) }{%
%dx}\right\vert _{x=y}.$ For example, suppose $f\left( x\right) =x^{2},$ $%
%f^{\prime }\left( y\right) =2y$ for any discrete $y.$

\item[-] $\kappa _{n}(\bar{\chi}):=\mathbb{E}\left[ y_{31}|z_{2}>\sigma
_{n},y_{2}=1,\bar{\chi}\right] $ and $\kappa
^{+}(\bar{\chi}):=\lim_{n\rightarrow \infty }\mathbb{E}\left[
y_{31}|z_{2}>\sigma _{n},y_{2}=1,\bar{\chi}\right] $. $\dot{\kappa}_{n}(\nu
):=\left. \frac{\partial \kappa _{n}(\bar{\chi})}{\partial \bar{\chi}}%
\right\vert _{\bar{\chi}=\nu }$ and $\dot{\kappa}^{+}(\nu ):=\left. \frac{%
\partial \kappa ^{+}(\bar{\chi})}{\partial \bar{\chi}}\right\vert _{\bar{\chi}=\nu
}$.%\footnote{In our case, $y$ only takes value 0 and 1. When we say $f(y)$ is differentiable, we mean $f(\cdot)$ is differentiable in $\cdot$, e.g., $f(x)=x^2$ is differentiable in $x$.}

\item[-] $F_{\bar{\chi}}(\cdot |z_{2}>\sigma _{n},y_{2}=1)$ ($%
f_{\bar{\chi}}(\cdot |z_{2}>\sigma _{n},y_{2}=1)$) denotes the joint CDF (PDF)
of $\bar{\chi}$ conditional on $\{z_{2}>\sigma _{n},y_{2}=1\}$. $%
F_{\bar{\chi}}^{+}(\cdot |y_{2}=1):=\lim_{n\rightarrow \infty
}F_{\bar{\chi}}(\cdot |z_{2}>\sigma _{n},y_{2}=1)$.
\end{itemize}

\begin{assumptionp}{C}
Suppose the following conditions hold. $\text{ }$

\begin{enumerate}
\item[C1] The proposed estimator $\hat{\theta}_{n}$ satisfies $Q_{n1}(\hat{%
\theta}_{n})\geq\sup_{\vartheta\in\Theta}Q_{n1}(\vartheta)-o_{p}%
((nh_{n})^{-2/3})$.

\item[C2] $P\left( z_{2}>\sigma |y_{2}=1,y_{31},\bar{\chi}\right) >0$ for all $%
\sigma >0$ and almost every $(y_{31},\bar{\chi})$.

\item[C3] (i) $F_{\bar{\chi}}^{+}(\cdot |y_{2}=1)$ is non-degenerate and has
an uniformly bounded PDF $f_{\bar{\chi}}^{+}(\cdot |y_{2}=1)$, and (ii) $%
\sup_{\nu }|f_{\bar{\chi}}(\nu |z_{2}>\sigma
_{n},y_{2}=1)-f_{\bar{\chi}}^{+}(\nu |y_{2}=1)|=o(1)$.

\item[C4] (i) $\kappa _{n}(\bar{\chi})$ and $\kappa ^{+}(\bar{\chi})$ are
differentiable in $\bar{\chi}$, and (ii) $\sup_{\nu }|\dot{%
\kappa}_{n}(\nu )-\dot{\kappa}^{+}(\nu )|=o(1)$.

\item[C5] (i) $\mathbb{E}[\bar{q}_{n1,\vartheta }(\bar{\chi})]$ and $\mathbb{E%
}[\bar{q}_{1\vartheta }^{+}(\bar{\chi})]$ are twice continuously
differentiable at $\vartheta $ in a small neighborhood of $\theta $, (ii) $%
(\alpha ,x^{T})$ has a compact support, and (iii) for any constant $%
\varsigma $, $\sup_{\alpha }P(\epsilon _{2}\geq \varsigma +\sigma
_{n}|\alpha )=o((nh_{n})^{-1/3})$.

\item[C6] $h_{n}\gtrsim n^{-1+\varepsilon }$ for some small positive $\varepsilon$.
\end{enumerate}
\end{assumptionp}

Assumption C1 is standard in the literature (see, e.g., \cite{KimPollard1990} and \cite{SeoOtsu2018}). This assumption implies that the maximization of $Q_{n1}(\vartheta)$ need not be exact, and any approximate maximizer close enough to the exact one will be enough for the asymptotic analysis.
Assumption C2 is an implication of Assumptions A1 and A2. We list it as a
separate condition here mainly because it is more directly related to our
proof process presented in Appendix \ref{appendixB}. Assumption C3
strengthens Assumption A4(ii). Assumption C4 requires the two conditional
probabilities $\kappa_{n}(\bar{\chi})$ and $\kappa ^{+}(\bar{\chi})$ to be
smooth enough, which, together with Assumption C3, is important for calculating the expected value of the limiting distribution of the
estimator $\hat{\theta}_{n}$.

The smoothness conditions imposed in Assumption C5(i) are standard in the literature as well.
Assumption C5(ii) is made to simplify the proof process and can be relaxed to allow for unbounded support, albeit with more tedious discussions. The essential requirement here is to exclude the scenario in which $\alpha + x_2^{\prime} \beta \rightarrow -\infty$ as $z_2 \rightarrow +\infty$.
Assumption C5(iii) essentially places a restriction on the relative tail
behavior of the observed regressor $z_t$ and unobserved error $\epsilon_t$. As shown in
the proof of Theorem \ref{T:limiting_dist}, this assumption ensures that the bias of the estimator $\hat{%
\theta}_{n}$ shrinks sufficiently fast. If this condition is violated, the bias term dominates the distribution, and inferences are not possible. It is worth noting that such
condition plays a crucial role in determining the rate of convergence of
estimators based on ``irregular identification'' strategies including the
``identification at infinity'' as a special case. See \cite%
{khan2010irregular} for an in-depth investigation on this issue.

Assumption C6, together with Assumption C5(iii), guides the selection of the tuning parameter $\sigma_{n}$. These two conditions are in the same spirit of Assumptions 8 and 8* in \cite{andrews}. On one hand, since $%
h_{n}=P(z_{2}>\sigma_{n}|y_{2}=1)$ controls the effective sample size of the
estimation procedure, Assumption C6 implies that $\sigma_{n}$ should not increase too rapidly as $n\rightarrow\infty$, ensuring enough effective observations to control the variance of $\hat{\theta}_{n}$. On the
other hand, Assumption C5(iii) suggests that $\sigma_{n}$ should grow
sufficiently fast as $n\rightarrow\infty$ to lower the bias of $\hat{\theta}%
_{n}$.

However, there is no way to determine the optimal $%
\sigma_{n}$ since this requires the knowledge of relative (unknown) tail
behavior of $z_{t}$ and $\epsilon_{t}$. This feature is well known to the ``identification at infinity'' type of estimators, see, e.g., \cite{andrews}. We suggest choices of $\sigma_n$ that satisfy both Assumptions C5(iii) and C6 in some special cases in Table \ref{T:sigman}. From the table, there are no valid $\sigma_n$ in case (I) when  $\lambda^{\prime}>\lambda$, and in case (III). The valid choices of $\sigma_n$, if exists, differ from case to case. As expected, we prefer the cases where $z_2$ possesses heavier tails than $\epsilon$. \cite{andrews} share similar results, for details, see their discussions after Assumption 8*.

For practice, we propose to take $\sigma_{n}=\sqrt{\log n/2.95}$. This choice of $\sigma_n$ is valid for case (I) with $\lambda^{\prime}=2$ and $\lambda'<\lambda$, and case (II) with $\lambda=2$. Moreover, with this choice of $\sigma_n$, Assumption B2(ii) is satisfied for $z_{2}$ with $P(z_{2} > \sigma_{n}) \gtrsim e^{-\left(1-\varepsilon\right) 2.95\sigma_{n}^{2}}$ for some $\varepsilon\in (1/4,1)$. We show the finite sample properties of our estimator with this choice of $\sigma_n$ by means of simulations in Section \ref{sec_simulation}. This chosen $\sigma_n$ works well (the bias does not dominate the distribution) even in the situation that belongs to case (I) with $\lambda^{\prime}>\lambda$, where no valid $\sigma_n$ exists.

\begin{table}[ptb]
\caption{Choice of $\sigma_n$ that satisfies both Assumptions C5(iii) and C6}
\small
\label{T:sigman}\centering
\begin{tabular}
[c]{l|c|c}\hline\hline
   & $ P\left(  \epsilon>t\right) \asymp \exp\left(  -t^{\lambda}\right)  $ & $P\left(  \epsilon>t\right) \asymp  t^{-\lambda}$\\\hline
 $\begin{array}
[c]{c}
P\left(  z_{2}>t\right) \\
\asymp \exp(  -t^{\lambda^{\prime}})
\end{array} $ & (I):
$
\begin{array}{lc}
\sigma_{n}=(c\log n)^{1/\lambda'}\textrm{ }\forall c\in(0,1-\varepsilon], & \textrm{if }\lambda'<\lambda\\
\sigma_{n}=(c\log n)^{1/\lambda'}\textrm{ }\forall c\in(1/4,1-\varepsilon], & \textrm{if }\lambda'=\lambda\\
\textrm{No }\sigma_{n}, & \textrm{if }\lambda'>\lambda
\end{array}

$ & (III): $\text{No }\sigma_{n}$\\\hline
 $\begin{array}
[c]{c}
P\left(  z_{2}>t\right) \\
\asymp t^{-\lambda^{\prime}}
\end{array}$ & (II): $\left( \log n/3\right)  ^{1/\lambda}%
<\sigma_{n}\lesssim n^{\frac{1-\varepsilon}{\lambda^{\prime}}}$ & (IV):
$n^{\frac{1}{3\lambda+\lambda^{\prime}}}\ll\sigma_{n}\lesssim
n^{\frac{1-\varepsilon}{\lambda^{\prime}}}$\\
\hline\hline
\multicolumn{3}{l}{Note: We focus on the right tail, and $\lambda^{\prime}, \lambda > 0$.} \\
\end{tabular}
\end{table}

The above conditions are sufficient to characterize the asymptotic distribution of the estimator obtained by maximizing (\ref{eq:Qn1})--(\ref{eq:Qn}), as presented in the following theorem, along with its proof in Appendix \ref{appendixB}.

\begin{theorem}
\label{T:limiting_dist} Suppose Assumptions A--C hold. Then, (i) $\hat{\theta}%
_{n}-\theta =O_{p}\left( (nh_{n})^{-1/3}\right) $, and (ii)
\begin{equation*}
\left( nh_{n}\right) ^{1/3}(\hat{\theta}_{n}-\theta )\overset{d}{\rightarrow
}\arg \max_{s\in \mathbb{R}^{p+2}}Z(s),
\end{equation*}%
where $Z(s)$ is a Gaussian process with continuous sample paths, expected
value $s^{\prime }Vs/2$, and covariance kernel $H(s_{1},s_{2})$ for $%
s_{1},s_{2}\in \mathbb{R}^{p+2}$. $V$ and $H(\cdot ,\cdot )$ are defined in
Lemma \ref{lemma_M1} and expression (\ref{eq:cov_kernel}), respectively.
\end{theorem}

Note that Theorem \ref{T:limiting_dist}  does not determine the exact rate of convergence of $\hat{\theta}_n$, as $h_n$ depends on the unknown tail probabilities of $z_2$. However, the lack of this knowledge does not render statistical inference infeasible. In Section 6, we will apply the $m$-out-of-$n$ bootstrap to conduct the inference in an empirical application. We choose this method for two reasons: it is comparatively easier to implement, and it provides an estimate of the convergence rate for our estimator.

In Remark \ref{remark_boot}, we discuss several sampling-based methods with the potential to enable statistical inference in the absence of knowledge of the exact convergence rate of the estimator.

\begin{remark}
It is worth noting that the rate of convergence of $\hat{\theta}_{n}$ depends
on the relative tail behavior of the distributions of $z_{t}$ and
$\epsilon_{t}$. To achieve a faster convergence rate of $\hat{\theta}_{n}$, it
is desirable for the distribution of $z_{2}$ to have heavier tails compared to
$\epsilon_{2}$. To see this, consider any eligible $\sigma_{n}$ satisfying
both Assumptions C5 and C6. Suppose $P\left(  \epsilon_{2}>\sigma_{n}\right)
\asymp h_{n}^{\upsilon}$ (the bias term) for some $\upsilon>0$, and
$h_{n}^{\upsilon}\asymp(nh_{n})^{-1/3}$ for the fastest possible convergence
rate $n^{-\upsilon/\left(  1+3\upsilon\right)  }$. Recall that $P\left(
z_{2}>\sigma_{n}\right)  =h_{n}.$  When $z_{2}$ and $\epsilon
_{2}$ have the same tail, $\upsilon=1$ and the convergence rate is $n^{-1/4}$.$\ $Loosely
speaking$,$ $\upsilon\ $increases as the tail of $z_{2}$ becomes thicker, and
decreases otherwise. Therefore, for any eligible $\sigma_{n}$, the convergence rate
of $\hat{\theta}_{n}$ increases in $\upsilon$ (as the tail of $z_{2}$ becomes
thicker) and approaches $n^{-1/3}$ for large $\upsilon$ (as the tail of $z_{2}$ becomes much
thicker than $ \epsilon_{2}$).
\end{remark}

\begin{remark}\label{remark_boot}
Theorem \ref{T:limiting_dist} indicates that the proposed estimator $\hat{\theta}_n$ has a slower than cube-root-$n$ rate of convergence and its
asymptotic distribution is not Gaussian. As a result, standard inference methods
based on asymptotic normality no longer work here. Smoothing the objective
function in the sense of \cite{andrews} and \cite{horowitz} (See also \cite{kyriazidou1997estimation} and \cite{charlier1997limited}) may yield a
faster rate and regain an asymptotically normal estimator. However, this
involves choosing additional kernel functions and tuning parameters for the
two indicator functions in objective functions (\ref{eq:Qn1}) and (\ref{eq:Qn2}). A more practical alternative may be to consider sampling-based
inference methods. It is known that the naive nonparametric bootstrap is
typically invalid under the cube-root asymptotics (\cite{AbrevayaHuang2005}). For the ordinary MS estimator, valid inference can be conducted using
subsampling (\cite{DelgadoEtal2001}), the $m$-out-of-$n$ bootstrap (\cite{LeePun2006}), the numerical bootstrap (\cite{HongLi2020}), and a model-based bootstrap with modified objective function (\cite{cattaneo2020bootstrap}), among other procedures. \cite{OuyangYang2024binary} show that \citeauthor{LeePun2006}'s (\citeyear{LeePun2006}), \citeauthor{HongLi2020}'s (\citeyear{HongLi2020}),
and \citeauthor{cattaneo2020bootstrap}'s (\citeyear{cattaneo2020bootstrap})
methods, with certain modifications, are valid for kernel weighted MS
estimators with asymptotics similar to Theorem \ref{T:limiting_dist}.
Similar methods might apply to the estimator proposed in this paper. However, extending these bootstrap methods to cases with unknown convergence rates requires significant effort and is beyond
the scope of this paper. We, therefore, defer this task to future studies.
\end{remark}

\section{Monte Carlo Experiments}\label{sec_simulation}

In this section, we investigate the finite-sample performance of the
proposed estimators by means of Monte Carlo experiments. We examine
two designs, each with a less favorable scenario for our estimator.
In these scenarios, $z_{t}$ has a thinner tail than $\epsilon_{t}$,
and there is no theoretically valid $\sigma_{n}$. These are the first
scenarios in both Designs 1 and 2 presented below. Despite these challenges,
our estimator performs reasonably well, as the bias term does not
appear to dominate the distribution.

We start by considering a benchmark design similar to that used in
\citet{honore-k}, but we add an additional covariate $z_{it}$ and
a time trend that \citet{honore-k} cannot handle. Specifically, this
design (referred to as Design 1) is specified as follows:
\begin{align*}
y_{i0} & =\mathds{1}\left\{ \alpha_{i}+\delta\times\left(0-2\right)+\beta_{1}x_{i0,1}+z_{i0}\geq\epsilon_{i0}\right\} ,\\
y_{it} & =\mathds{1}\left\{ \alpha_{i}+\delta\times\left(t-2\right)+\gamma y_{it-1}+\beta_{1}x_{it,1}+z_{it}\geq\epsilon_{it}\right\} ,\text{ \ }t\in\left\{ 1,2,3\right\} ,
\end{align*}
where we set $\gamma=\beta_{1}=1$ and $\delta=1/2$. Following the discussion in Remark \ref{remark_A5}, we normalize the coefficient on $z_{it}$ to 1 for all designs investigated in this section and Appendix \ref{appendixC}. We consider
two scenarios. For each, we let $x_{it,1}\overset{d}{\sim}N\left(0,1\right),$
$\epsilon_{it}\overset{d}{\sim}(\pi^{2}/3)^{-1/2}\cdot$Logistic$\left(0,1\right)$
(the variance of $\epsilon_{it}$ is 1)$,$\ and $\alpha_{i}=\left(x_{i0,1}+x_{i1,1}+x_{i2,1}+x_{i3,1}\right)/4,$\ but
we consider $z_{it}$ with different tail behaviors. $x_{\cdot,1},z_{\cdot},$
and $\epsilon_{\cdot}$ are independent of each other, and all covariates
are i.i.d. across $i$ and $t.$ In the first scenario, we set $z_{it}\overset{d}{\sim}N(0,1),$
and denote it as ``Norm''. In the second scenario, we set $z_{it}\overset{d}{\sim}\text{Laplace}(0,\sqrt{2}/2)$
(with zero mean and unit variance), and denote it as ``Lap''. Note
that the density function of the Laplace distribution decays like
$e^{-\left\vert x\right\vert /c}$ for some constant $c$ at its tail$,$
which is heavier than the tail of the normal density.

In the second design (referred to as Design 2), the setup is the same
as that in Design 1, except that we add one more covariate to examine how
our estimators perform in a higher dimensional design. Specifically,
\begin{align*}
y_{i0} & =\mathds{1}\left\{ \alpha_{i}+\delta\times\left(0-2\right)+\beta_{1}x_{i0,1}+\beta_{2}x_{i0,2}+z_{i0}\geq\epsilon_{i0}\right\} ,\\
y_{it} & =\mathds{1}\left\{ \alpha_{i}+\delta\times\left(t-2\right)+\gamma y_{it-1}+\beta_{1}x_{it,1}+\beta_{2}x_{it,2}+z_{it}\geq\epsilon_{it}\right\} ,\text{ \ }t\in\left\{ 1,2,3\right\} ,
\end{align*}
where we set $\gamma=\beta_{1}=\beta_{2}=1$ and $\delta=1/2$. Random
covariates are generated as$\ x_{it,1},x_{it,2}\overset{d}{\sim}N\left(0,\sqrt{2}/2\right)$,
$\epsilon_{it}\overset{d}{\sim}(\pi^{2}/3)^{-1/2}\cdot\text{Logistic}\left(0,1\right)$,
and $\alpha_{i}=\sum_{t=0}^{3}(x_{it,1}+x_{it,2})/4$. Similarly,
we consider two scenarios with the same $z_{it}$ as in design 1$.$
Again, $x_{\cdot,1}$, $x_{\cdot,2}$, $z_{\cdot}$, and $\epsilon_{\cdot}$
are independent of each other. To investigate only the impact of higher dimension, we
set the variance of $x_{it,1}+x_{it,2}$ in Design 2 to be the same
as that of $x_{it,1}$ in Design 1.

As discussed in Section \ref{Sec:asymp_dist}, we set
$\sigma_{n}$ as
\[
\sigma_{n}=\widehat{\text{std}\left(z_{i2}\right)}\sqrt{\log n^{*}/2.95},
\]
where $\widehat{\text{std}\left(z_{i2}\right)}$ is the sample standard
deviation of $z_{2}$, and $n^{*}$ is the number of ``switchers'',
that is, observations with $y_{3}\neq y_{1}$. The usage of $n^{*}$
is intended to provide better control over the tuning parameters,
based on the features of the data. In practice, one may normalize
$z_{it}$ to mean 0 and variance 1 and set $\sigma_{n}=\sqrt{\log n^{*}/2.95}$.
We consider sample sizes of $n=5000,10000$, and $20000$. All the simulation results presented in this section are based on
1000 replications of each sample size. We implement MS
estimations in R, using the differential evolution (DE)
algorithm to attain a global optimum of the
objective function. The DE algorithm, developed by \citet{storn1997differential},
is capable of searching for the global optimum of a real-valued function
with real-valued parameters, even if the function lacks continuity
or differentiability. This algorithm has been effectively employed
in calculating MS-type estimators in the literature, including \citet{Fox2007}
and \citet{YanYoo2019}. \citet{mullen2011deoptim} provides a comprehensive
introduction to the R package $\texttt{DEoptim}$, which implements
the DE algorithm. We report the mean bias (MBIAS) and the root mean square errors (RMSE) of the estimates for Designs
1 and 2 in Tables \ref{T:D1} and \ref{T:D2}, respectively.

We summarize the findings in Tables \ref{T:D1} and \ref{T:D2}. First,
the RMSEs of all parameters decrease as the sample size increases,
but they converge to zero slower than the parametric rate. Second,
the convergence rate is faster with a thicker-tailed $z_{i2}$, as evidenced by comparing the RMSEs from Norm to Lap. Third, the RMSE does not appear to increase
for $\gamma$ and $\delta$ as we have one more covariate from Design
1 to Design 2. This confirms our theoretical findings. Note that the RMSE increases a bit for $\beta_{1}$, but this is probably due to
the lower variance of $x_{\cdot,1}$ in Design 2. To investigate the
sensitivity of the results to $\sigma_{n},$ we consider $\sigma_{n}=0.9\cdot\widehat{\text{std}\left(z_{i2}\right)}\sqrt{\log n^{*}/2.95}$
and $\sigma_{n}=1.1\cdot\widehat{\text{std}\left(z_{i2}\right)}\sqrt{\log n^{*}/2.95}$
(we need larger $\sigma_{n}$ to be in line with the discussion in
Section \ref{Sec:asymp_dist}), and report the corresponding results
in Tables \ref{T:D1_robust} and \ref{T:D2_robust} in Appendix \ref{appendixC}.
We note that the results are not sensitive to the choices of the tuning
parameters.

\begin{comment}
In Appendix \ref{appendixC}, we examine the impact of auto-correlations
of regressors on the performance of our estimators. After removing
the time trend term, we construct a new design and we compare the
performance of our estimator with the estimators in \citet{honore-k}
and \citet{OuyangYang2024binary}.\footnote{These two competing methods are not applicable in the presence of time
trends and dummies.} We briefly summarize the results in Appendix \ref{appendixC}. With
certain degrees of autocorrelations, our estimator performs reasonably
well, yet does not perform as well as before. Our estimator's performance
is comparable to that of the semiparametric estimators in \citet{honore-k}
and \citet{OuyangYang2024binary}. Notably, these competing methods are not valid in scenarios involving time trends and dummies, which are prevalent in empirical applications. In such contexts, our approach offers a valuable alternative.
\end{comment}

In Appendix \ref{appendixC}, we report additional results from supplementary simulation studies. Firstly, we investigate the impact of auto-correlations of the regressors on the performance of our estimator. Additionally, we compare the performance of our estimator with those proposed by \citet{honore-k} and \citet{OuyangYang2024binary} in designs without the time trend term. We direct interested readers to Appendix \ref{appendixC} for a more detailed discussion. Here, we provide a brief summary of these results. Our estimator still performs reasonably well with certain degrees of auto-correlations, but as expected, not as well as in Designs 1 and 2, where regressors are serially independent. Our estimator's performance is comparable to that of the semiparametric estimators proposed by \citet{honore-k} and \citet{OuyangYang2024binary}. It is essential to highlight that these alternative methods are not applicable in scenarios involving time trends or dummies, which are common in empirical applications. In such contexts, our approach offers a valuable alternative.

\begin{table}[H]
\caption{Simulation Results of Design 1}
\label{T:D1}\centering %
\begin{tabular}{r|cc|cc|cc}
\hline  \hline
 & \multicolumn{2}{c|}{$\beta_{1}$} & \multicolumn{2}{c|}{$\gamma$} & \multicolumn{2}{c}{$\delta$}\tabularnewline
 & MBIAS  & RMSE  & MBIAS  & RMSE  & MBIAS  & RMSE \tabularnewline
\hline
       $n_1$  & 0.120 & 0.407 & -0.027  & 0.549  & 0.068  & 0.228 \\
 Norm  $n_2$  & 0.075 & 0.299 & -0.015  & 0.427  & 0.053  & 0.171 \\
       $n_3$  & 0.048 & 0.219 & -0.062  & 0.340  & 0.043  & 0.132 \\ \hline
       $n_1$  & 0.039 & 0.249 & -0.021  & 0.386  & 0.030  & 0.154 \\
 Lap   $n_2$  & 0.024 & 0.185 & -0.049  & 0.306  & 0.019  & 0.116 \\
       $n_3$  & 0.027 & 0.146 & -0.041  & 0.247  & 0.020  & 0.095 \\
\hline
\hline
\multicolumn{7}{l}{Note: $n_{1}=5000,n_{2}=10000,n_{3}=20000$.}\tabularnewline
\end{tabular}
\end{table}

\begin{table}[H]
\caption{Simulation Results of Design 2}
\label{T:D2}\centering
\begin{tabular}{r|cc|cc|cc|cc}
\hline  \hline
 & \multicolumn{2}{c|}{$\beta_{1}$} & \multicolumn{2}{c|}{$\beta_{2}$} & \multicolumn{2}{c|}{$\gamma$} & \multicolumn{2}{c}{$\delta$}\tabularnewline
 & MBIAS  & RMSE  & MBIAS  & RMSE  & MBIAS  & RMSE  & MBIAS  & RMSE \tabularnewline
\hline
       $n_1$  & 0.144 & 0.471 & 0.157 & 0.475 & 0.018  & 0.542  & 0.093  & 0.235 \\
 Norm  $n_2$  & 0.087 & 0.344 & 0.085 & 0.355 & -0.032  & 0.430  & 0.057  & 0.172 \\
       $n_3$  & 0.048 & 0.276 & 0.062 & 0.273 & -0.048  & 0.332  & 0.044  & 0.137 \\\hline
       $n_1$  & 0.072 & 0.314 & 0.078 & 0.313 & -0.009  & 0.395  & 0.046  & 0.154 \\
 Lap   $n_2$  & 0.023 & 0.216 & 0.036 & 0.238 & -0.025  & 0.303  & 0.023  & 0.111 \\
       $n_3$  & 0.025 & 0.182 & 0.022 & 0.178 & -0.035  & 0.236  & 0.021  & 0.091 \\
\hline
\hline
\multicolumn{9}{l}{Note: $n_{1}=5000;n_{2}=10000;n_{3}=20000$.}\tabularnewline
\end{tabular}
\end{table}

A final note is that when using observational data, the choice of $\sigma_n$ depends on the unknown tail behavior of the variable $z_2$. As there are no formal methods to determine the appropriateness of a specific $\sigma_n$, we suggest practitioners try different $\sigma_n$'s in estimation and check if the results are sensitive to different choices.

\section{Empirical Illustration} \label{sec_application}

In Australia, Medicare is the universal tax-funded public health insurance scheme that provides free access to public
hospitals. Medicare patients in public hospitals receive free treatment from
doctors nominated by hospitals and free (shared) accommodations. Patients
may opt to receive private care in either private or public hospitals (as
private patients) to have their choice of doctors and nurses, better
amenities (e.g., private rooms, family member accommodation, etc.), and
quicker access to treatment by avoiding long waiting time experienced by
many Medicare patients. Medicare does not cover private hospital care. On
top of a patient copayment, the cost is either afforded by private patients
themselves as out-of-pocket expenditure or covered by their private hospital
(insurance) cover (PHC), if any. Having PHC does not preclude using hospital
care as a Medicare patient. The institutional context for Australia's
Medicare and private health insurance schemes has been more thoroughly
described in the vast health economics literature, e.g., Section 2 of \cite%
{cheng2014measuring}. We refer interested readers to \cite%
{cheng2014measuring} and references therein for more detailed information.

In this section, we apply our MS estimator to analyze the state dependence
and the impacts of government incentives on the choice to purchase PHC,
using 10 waves (waves 11--20 corresponding to years 2011--2020) of the
Household, Income and Labor Dynamics in Australia (\href{https://melbourneinstitute.unimelb.edu.au/hilda}{HILDA}) Survey data. Since 2011, the HILDA survey has begun recording information about respondents' enrollment in PHC.  

We denote the dependent variable, $y_{it}$, as whether individual $i$ has PHC in year $t$. We are interested in the effects of ``Lifetime Health
Cover'' (LHC) policy, ``Medicare Levy
Surcharge'' (MLS), and the state persistence $(y_{it-1})$
on one's purchasing PHC.

The age dummy variable $\text{Above30}_{it}$ indicates if individual $i$ is
30 years old or above in year $t$, namely, $\text{Above30}_{it}:=\mathds{1}%
\{\text{Age}_{it}\geq30\}$. Following the insight of the (sharp)
``regression discontinuity'' design, its coefficient captures the
effects of Australia's LHC policy introduced in 2000 to encourage the uptake
of PHC. Loosely speaking, the LHC states that if an individual has not taken
out and maintained PHC from the year she turns 31, she will pay a 2\% LHC
loading on top of her premium for every year she is aged over 30 if she
decides to take out PHC later in life. If LHC is a strong incentive, we
would expect a significant ``jump" in the PHC enrollment rate at this age.

The MLS is a levy paid by Australian taxpayers who do not have PHC and earn
above a certain income threshold. In the sample years of our data, MLS
rates remain unchanged, while the thresholds have been raised yearly until
2014. It is worth noting that the 2014 rise in MLS thresholds was only 50\%
of previous years, and the thresholds have remained at the same level until
2022. The time dummy $D_{2014,t}$ is included in model (\ref{app_model}%
) to examine whether this change in MLS policy would affect people's
willingness to purchase PHC. Note that \citeauthor{honore-k}'s (%
\citeyear{honore-k}) estimators do not allow either age ($\text{Age}_{it}$)
or fixed time effects ($D_{2014,t}$) since they do not have overlapping
supports across time.

$I_{it}$ represents standardized annual household disposable income
using the entire sample in the survey, which serves as the continuous regressor
with rich enough support required for point identification (by Assumptions A2 and A3). The standardization is performed before dropping missing data by subtracting the sample mean from each individual value and then dividing the difference by the standard deviation. 

We also include a (location) dummy variable $\text{GCC}_{it}$ that indicates whether individual $i$ lives in a major city/greater capital city in year $t$. This variable is included to control the accessibility to private hospital services. Tables \ref{tab_def} and \ref{tab_sum} provide definitions and summary statistics of all aforementioned variables, respectively. Note that some observations are excluded due to missing information in other variables, so in Table \ref{tab_sum}, $I_{it}$ does not have an exact zero mean and unit standard deviation.

With all these covariates, we specify our empirical model as follows:
\begin{equation}
y_{it}=\mathds{1}\{\alpha _{i}+\gamma y_{it-1}+\delta D_{2014,t}+\beta _{1}%
\text{Abov30}_{it}+\beta _{2}\text{Age}_{it}+\beta _{3}\text{GCC}%
_{it}+  I_{it}\geq \epsilon _{it}\}, \label{app_model}
\end{equation}
where $\epsilon _{it}$ and $\alpha _{i}$ are, respectively, the usual idiosyncratic error and unobserved heterogeneity in fixed effects panel data models.

In our analysis, we restrict the coefficient on $I_{it}$ to be 1, following the same convention for scale normalization as in Section \ref{sec_simulation}. This choice warrants justification; that is, household income enters the model with a significant positive coefficient, as implicitly required by Assumption A5'. We provide the following rationale for this based on common sense and evidence from exploratory regression. Practitioners seeking to justify normalizing the coefficient on $z_{it}$ to 1 can adopt similar argumentation method.

Firstly, in Australia, Medicare provides free access to public hospitals, and Medicare patients in public hospitals receive free treatment and accommodations. However, people can purchase private hospital insurance to cover faster and more premium services. Taking up or maintaining private insurance coverage requires a household to have sufficient disposable income. Besides, as income increases, the marginal utility of saving or other consumption may eventually become lower than that of enhanced private health care. In addition, Australia's tax system also gives considerable financial incentives for high-income households to buy private insurance. Therefore, common sense suggests that income should play a positive and significant role in private insurance purchases.

Secondly, we conduct a simple probit regression using one wave of the data and included income as the only regressor. The estimate is positive and significant, with $p$-value smaller than $10^{-15}$. This result holds true across all data waves, confirming our argument. This finding aligns with the results of more in-depth structural analyses in the health economics literature, such as \cite{cheng2014measuring}.

Note that Assumption A3 can be demanding. To address this concern, we relax Assumption A3 to Assumption A3' for a model closely resembling the current application and demonstrate identification under this relaxed condition, as detailed in Appendix \ref{appendix0_32}. Additionally, we justify our use of \( I_{it} \) as \( z_{it} \) under this modified condition in Appendix \ref{appendix0_33}, specifically by showing the kernel density and summary statistics of \( I_{it+1} - I_{it-1} \). For a more detailed discussion on Assumptions A2 and A3 and their roles in this empirical application, we refer interested readers to Appendix \ref{appendix0_3}.

\begin{table}[tpb]
\caption{Definition of Variables}
\label{tab_def}{\small \centering
\begin{tabularx}{1.0\textwidth}[c]{l|X}
\hline\hline
Variable & Description\\\hline
    Private hospital cover ($y$) & 1 if has private hospital cover for the whole year, otherwise 0 \\\hline
        Standardized income ($I$) & Standardized household's financial year disposable income (in the 2011 Australian dollar) \\\hline
 Above 30 years old (Above30) & 1 if age 30 years old or above, otherwise 0 \\\hline
  Age & Age \\\hline
Major city or greater capital city (GCC) & 1 if lives in a major city or greater capital city, otherwise 0 \\\hline
    Year 2014 ($D_{2014}$) & 1 if in financial year 2014, otherwise 0 \\\hline\hline
\end{tabularx}
}
\end{table}

\begin{table}[tpb]
\caption{Summary Statistics}
\label{tab_sum}%\small
\centering
\begin{tabular}{lccccc}
\hline\hline
Variable & $n\times T$ & Mean & Std.Dev. & Min & Max \\
\midrule $y_{it}$ & 65,603 & 0.527 & 0.499 & 0 & 1 \\
$I_{it}$ & 65,603 & -0.128 & 0.946 & -1.401 & 13.118 \\
$\text{Above30}_{it}$ & 65,603 & 0.855 & 0.353 & 0 & 1 \\
$\text{Age}_{it}$ & 65,603 & 50.091 & 17.484 & 17 & 99 \\
$\text{GCC}_{it}$ & 65,603 & 0.588 & 0.492 & 0 & 1 \\
$D_{2014,t}$ & 65,603 & 0.136 & 0.343 & 0 & 1 \\ \hline\hline
\end{tabular}%
\end{table}

Let $x_{it}:=(\text{Above30}_{it},\text{Age}_{it},\text{GCC}_{it})$ and $%
\beta :=(\beta _{1},\beta _{2},\beta _{3})$. We estimate $\theta :=(\delta
,\gamma ,\beta)$ through maximizing the objective function
\begin{equation}
Q_{n1}(\vartheta ):=\frac{1}{n}\sum_{i=1}^{n}%
\sum_{t=2}^{T_{i}-1}y_{it}(y_{it+1}-y_{it-1})\cdot \mathds{1}\{I_{it}>\sigma
_{n}\}\cdot \mathds{1}\{u_{it}(\vartheta )>0\},  \label{app_Qn1}
\end{equation}
where $u_{it}(\vartheta
):=r(y_{it}-y_{it-2})+d(D_{2014,t+1}-D_{2014,t-1})+(x_{it+1}-x_{it-1})^{%
\prime }b+(I_{it+1}-I_{it-1})$ and $\vartheta :=\left(d, r,b\right)$. Objective function (\ref{app_Qn1}) extends (\ref{eq:Qn1}) for longer and unbalanced panels in which the number of waves being observed varies across
individuals $i$ ($=:T_{i}$). We select the tuning parameter $\sigma_{n}$ using the same approach as described in Section \ref{sec_simulation}. It is important to note that the distribution of $I_{it}$ exhibits a significantly longer right tail compared to its left tail (skewness=3.78). Consequently, for sufficiently large $\sigma_{n}$, the objective function (\ref{app_Qn1}) has a much larger number of observations to use than the objective function
\begin{equation}
Q_{n2}(\vartheta ):=\frac{1}{n}\sum_{i=1}^{n}%
\sum_{t=2}^{T_{i}-1}(1-y_{it})(y_{it+1}-y_{it-1})\cdot \mathds{1}%
\{I_{it}<-\sigma _{n}\}\cdot \mathds{1}\{u_{it}(\vartheta )>0\}, \notag \label{app_Qn2}
\end{equation}
which extends (\ref{eq:Qn2}) for left tail observations. In fact, in this application, we set $\sigma_{n}=1.478$, which exceeds the absolute value of the lower bound of $I_{it}$ ($=1.401$ as shown in Table \ref{tab_sum}), thereby effectively using only objective function (\ref{app_Qn1}) and observations satisfying $\{I_{it}>\sigma_{n}\}$. Previous versions of this paper explored smaller values of $\sigma_n$ that allowed for the inclusion of left-tail observations (i.e., $\{I_{it}<-\sigma_{n}\}$), yielding similar results.

\begin{comment}
Let $x_{it}:=(\text{Above30}_{it},\text{Age}_{it},\text{GCC}_{it})$ and $%
\beta :=(\beta _{1},\beta _{2},\beta _{3})$. We estimate $\theta :=(\delta
,\gamma ,\beta ,\varpi )$ through maximizing the objective function
\begin{equation}
Q_{n}(\vartheta ):=Q_{n1}(\vartheta )+Q_{n2}(\vartheta ),  \label{app_Qn}
\end{equation}%
where
\begin{equation}
Q_{n1}(\vartheta ):=\frac{1}{n}\sum_{i=1}^{n}%
\sum_{t=2}^{T_{i}-1}y_{it}(y_{it+1}-y_{it-1})\cdot \mathds{1}\{I_{it}>\sigma
_{n}\}\cdot \mathds{1}\{u_{it}(\vartheta )>0\}  \notag  \label{app_Qn1}
\end{equation}%
with $u_{it}(\vartheta
):=r(y_{it}-y_{it-2})+d(D_{2014,t+1}-D_{2014,t-1})+(x_{it+1}-x_{it-1})^{%
\prime }b+w(I_{it+1}-I_{it-1}),$ $\vartheta :=\left( r,d,b,w\right) $ , and%
\begin{equation}
Q_{n2}(\vartheta ):=\frac{1}{n}\sum_{i=1}^{n}%
\sum_{t=2}^{T_{i}-1}(1-y_{it})(y_{it+1}-y_{it-1})\cdot \mathds{1}%
\{I_{it}<-\sigma _{n}\}\cdot \mathds{1}\{u_{it}(\vartheta )>0\},  \notag
\label{app_Qn2}
\end{equation}%
extending (\ref{eq:Qn1}) and (\ref{eq:Qn2}) respectively for longer and
unbalanced panels in which the number of waves being observed varies across
individuals $i$ ($=:T_{i}$). Since $\theta $ can only be identified up to
scale, we restrict the search of $\hat{\theta}_{n}$ on the unit sphere.
\end{comment}

By construction, procedure (\ref{app_Qn1}) only uses the subsample of
individuals who can be observed for at least four consecutive waves. After
dropping observations with missing values, our sample consists of 14,880
individuals satisfying this criterion. The panel is unbalanced with $3\leq
T_i\leq 9$ using the notation in previous sections. In total, we have 65,603
observations, among which about 7.36\% observations are ``switchers'' that
are useful for either ours or \citeauthor{honore-k}'s (\citeyear{honore-k})
estimators. As in Section \ref{sec_simulation}, we use $n^{*}$ to denote the number of ``switchers''.

We choose $\sigma _{n}=c\cdot \widehat{\text{std}(I_{it})}\sqrt{\log n^{*}/2.95}$ with $%
c=1.0$ and $1.1$ to implement our MS estimation and
report the results in Table \ref{tab_res}. We provide summary statistics for the sub-sample of switchers with $I_{it}>\sigma_n$ in Table \ref{tab_sum_eff} of Appendix \ref{appendix0_33}. We also conducted estimations using $\sigma_n$ with $c=0.5, 0.7,$ and $0.9$. While these results show patterns similar to Table \ref{tab_res}, they highlight the bias-variance trade-off inherent in choosing the tuning parameter, a common challenge in semiparametric methods. These additional results and their discussion are included in Appendix \ref{appendix0_33}.

In addition to the estimates of $\theta$, we also try calculating the 90\% and 95\% confidence intervals
(CIs) for $\theta $ using the $m$-out-of-$n$ bootstrapping. Here we sample $n$ individuals (clusters) to create the bootstrap sample. The main
difficulty in implementing this (or alternative sampling-based) method is
that Theorem \ref{T:limiting_dist} does not give an analytical convergence
rate for the estimator $\hat{\theta}$ due to the unknown tail probabilities
of $I_{it}$. We apply the method proposed in Remark 3 of \cite{LeePun2006}
to solve this problem; that is, assume $\hat{\theta}_n$ has convergence rate of $%
n^{\lambda }$ and obtain an estimate $\hat{\lambda}$ of $\lambda $ using a
double $m$-out-of-$n$ bootstrapping procedure with two bootstrap sample
sizes $m_{1}=n^{\rho _{1}}$ and $m_{2}=n^{\rho _{2}}$ for $\rho _{1},\rho
_{2}\in (0,1)$. The 90\% and 95\% CI reported in Table \ref{tab_res} are
calculated with $B=500$ bootstrap replications, $m=n^{7/8}$, and $\hat{%
\lambda}=0.309$ (obtained with $\rho _{1}=6/7$ and $\rho _{2}=7/8$).

\begin{table}[htbp]
\caption{Estimates of Preference Coefficients}
\label{tab_res}\centering
\begin{tabular}{clLrcrc}
\hline\hline
                   & Variable & \multicolumn{1}{c}{Estimate} &  \multicolumn{2}{c}{[90\% Conf.Int.]} &  \multicolumn{2}{c}{[95\% Conf.Int.]}  \\
    \midrule
        \multirow{5}[2]{*}{$c = 1.0$} & $y_{it-1}$ & 5.275^{\ast\ast} & 0.714 & 20.664 & 0.278 & 21.400 \\
          & $\textrm{Above30}_{it}$ & 5.190^{\ast\ast} & 0.319 & 20.286 & 0.016 & 21.465 \\
          & $\textrm{Age}_{it}$  & -0.465 & -11.384 & 7.298 & -11.953 & 9.039 \\
          & $\textrm{GCC}_{it}$ & -0.317 & -11.140 & 9.127 & -11.645 & 9.738 \\
          & $D_{2014,t}$ & -1.548 & -13.747 & 5.664 & -14.276 & 6.648 \\
\hline
\multirow{5}[2]{*}{$c = 1.1$} & $y_{it-1}$ & 5.140^{\ast\ast} & 1.008 & 20.042 & 0.370 & 20.927 \\
          & $\textrm{Above30}_{it}$ & 5.227^{\ast\ast} & 1.000 & 19.524 & 0.455 & 20.581 \\
          & $\textrm{Age}_{it}$ & -0.505 & -10.909 & 7.203 & -11.470 & 8.683 \\
          & $\textrm{GCC}_{it}$   & -1.770 & -13.725 & 4.950 & -14.299 & 6.573 \\
          & $D_{2014,t}$ & -1.578 & -13.274 & 5.303 & -13.598 & 6.582 \\
\hline\hline
    \end{tabular}
\end{table}

We can see from Table \ref{tab_res} that the estimation results are similar
for the two tuning parameters. Therefore, the following discussion on the empirical
results will be mainly based on the estimates obtained with $c=1.0$. The insignificant coefficient indicates that living in GCC may not affect people's
willingness to buy PHC.  The significant positive coefficient
on $y_{it-1}$ demonstrates the strong state persistence of PHC, which explains
why we can only observe a small percentage of switchers in the data.
Surprisingly, people's decision to buy PHC is hardly influenced by age. For
the two policy variables, the large positive coefficient on $\text{Above30}%
_{it}$ confirms that the LHC policy is a strong incentive for people to buy
PHC, while the change in MLS income threshold does not exhibit a strong
impact represented by the coefficient on $D_{2014,t}$. An intuitive
explanation for the latter is that although MLS promoted PHC purchases when
it was introduced in 1997--1998, the subsequent adjustments of its income
threshold only affected a small group of people whose incomes were near the
threshold.

We end this section with some remarks. First, our approach is more suitable
for data with a relatively large proportion of ``switchers'' which make up
the effective sample for the estimator. Second, to implement our method, the
model should have a continuous covariate with large support and ideally weak
dependence on other included covariates. Third, in the
absence of knowledge (or at least a good estimate) of the free-varying
variable's tail probabilities, the asymptotics of our estimator derived in
Section \ref{Sec:asymp_dist} cannot provide a ``rule of thumb'' for choosing
optimal tuning parameter $\sigma_{n}$. Perhaps a practical way is to try
different $\sigma_{n}$'s, use \citeauthor{LeePun2006}'s (%
\citeyear{LeePun2006}) proposed method (or other similar methods) to
estimate the convergence rates, and pick the $\sigma_{n}$ that gives the
fastest (estimated) rate. The last remark is for the $m$-out-of-$n$
bootstrap inference. The choice of the bootstrap sample size $m$ is the key
issue. Remark 1 of \cite{LeePun2006} provides some existing data-driven
methods. However, none of them can confirm an (asymptotically) optimal
choice of $m$ in nonstandard M-estimation like ours. Theoretical research on
this topic is necessary, but this is beyond the scope of the current paper.

\section{Conclusions}\label{Sec:Conlcusions}
This paper proposes new identification and estimation methods for a class of distribution-free dynamic panel data binary choice models that is first studied in \cite{honore-k}. We show that in the presence of a free-varying continuous covariate with unbounded support, an ``identification at infinity'' strategy in the spirit of \cite{chamberlain1986asymptotic} enables the point identification of the model coefficients without the need of element-by-element matching of covariates over time, in contrast to the method proposed in \cite{honore-k}. This property makes our methods more practical for models with many covariates or important covariates whose support may not overlap over time. Our identification arguments motivate a conditional maximum score estimator that is proven to be consistent and with the convergence rate independent of the model dimension. However, the asymptotic distribution of the proposed estimator is non-Gaussian, in line with well-established literature on cube-root asymptotics. We suggest valid bootstrap methods for conducting statistical inference. The results of a Monte Carlo study demonstrate that our estimator performs adequately in finite samples. Lastly, we use the HILDA data to investigate the demand for private hospital insurance in Australia.

This paper leaves some open questions for future research. For instance, although we suggest several theoretically feasible bootstrap inference methods in Section \ref{Sec:Asymptotics}, their asymptotic validity, finite-sample performance, and implementability (e.g., choice of tuning parameters) are not examined. Alternatively, one can also investigate whether it is possible to achieve a faster rate of convergence and obtain an asymptotically normal distribution by combining \citeauthor{horowitz}'s (\citeyear{horowitz}) and \citeauthor{andrews}'s (\citeyear{andrews}) methods to smooth the sample objective function.

\nocite{HILDA,HILDA2}

\bibliography{references}

\appendix

\part*{\centering \protect\Large Supplementary Appendix}\label{appendix}
This supplementary appendix is organized into four sections. In Appendix \ref{appendix0}, we present supplementary discussions that are omitted from the main text due to space limitations. Specifically, We compares our method with the methods proposed in \cite{honore-k} (HK) and \cite{OuyangYang2024binary} (OY) in Appendices \ref{appendix0_1} and \ref{appendix0_2}, respectively. We also provide further discussions on on Assumption A2, Assumption A3, and their role in the empirical application in Appendices \ref{appendix0_31}, \ref{appendix0_32}, and \ref{appendix0_33}, respectively. Moreover, Appendix \ref{appendix0_34} offers a set of sufficient conditions for Assumption B3.  Appendix \ref{appendixA} proves Theorems \ref{T:identify} (identification) and \ref{T:consistency} (consistency). Appendix \ref{appendixB} derives the
convergence rate and asymptotic distribution results summarized in Theorem \ref{T:limiting_dist}. Results of supplementary simulation studies are collected in Appendix \ref{appendixC}.

As discussed in the paper, all proofs presented in
these appendices are for the model with $T=3$ and the estimator
\begin{equation}
\hat{\theta}_{n}:=\arg\max_{\vartheta\in\Theta}Q_{n1}(\vartheta)  \notag
\end{equation}
with $Q_{n1}(\vartheta)$ defined in (\ref{eq:Qn1}) (corresponding to
identification equation (\ref{eq:iden_ineq1})).\footnote{In what follows, we may scale $Q_{n1}(\vartheta)$ with different factors
such as $P\left(z_{2}>\sigma_{n}\right)$ and $P\left(z_{2}>\sigma_{n}|y_2=1%
\right)$ in the proofs of different results to ease the derivation. Scaling $%
Q_{n1}(\vartheta)$ with these factors does not affect the value of $\hat{%
\theta}_n$.} The estimators obtained with longer panels or from maximizing
objective functions (\ref{eq:Qn2}) or (\ref{eq:Qn}) are of the same
structure, and hence the generalization is straightforward. In what follows,
we will use compact notation $\xi_{ts}$ for generic vectors $\xi_t$ and $%
\xi_s$ to denote $\xi_t-\xi_s$.

\section{Some Supplementary Discussions}\label{appendix0}

\subsection{Comparing with HK}\label{appendix0_1}
It is clear that the derivation of (\ref{eq:prob_ratio_1}) relies
on the assumption that the marginal distribution of $z_{2}$, conditional
on $\alpha$, $y_{0}$, and all other covariates, has an unbounded
support. HK identify model (\ref{model}) under the restriction
that $(x_{32},z_{32})$ has a support in some open neighborhood
of zero, conditional on $\alpha$, $y_{0}$, and all other covariates.\footnote{The $x_{t}$ in \citet{honore-k} corresponds to our $(x_{t},z_{t})$.}
Specifically, HK consider the following choice histories for the
model with $T=3$:
\begin{align*}
A & =\{y_{0}=d_{0},y_{1}=0,y_{2}=1,y_{3}=d_{3}\},\\
B & =\{y_{0}=d_{0},y_{1}=1,y_{2}=0,y_{3}=d_{3}\},
\end{align*}
where $d_{0},d_{3}\in\{0,1\}$. Applying arguments similar to those
used for obtaining (\ref{eq:prob_ratio_1}) yields
\begin{align}
&\frac{P(A|\alpha,y_{0}=d_{0},x^{T},z^{T},x_{2}=x_{3},z_{2}=z_{3})}{P(B|\alpha,y_{0}=d_{0},x^{T},z^{T},x_{2}=x_{3},z_{2}=z_{3})} \nonumber \\
= & \frac{1-F_{\epsilon|\alpha}(\alpha+\gamma d_{0}+x_{1}^{\prime}\beta+\varpi z_{1})}{1-F_{\epsilon|\alpha}(\alpha+\gamma d_{3}+x_{2}^{\prime}\beta+\varpi z_{2})}  \times\frac{F_{\epsilon|\alpha}(\alpha+\gamma d_{3}+x_{2}^{\prime}\beta+\varpi z_{2})}{F_{\epsilon|\alpha}(\alpha+\gamma d_{0}+x_{1}^{\prime}\beta+\varpi z_{1})}.
\end{align}
Then, this expression implies
\begin{align}
 & \text{sgn}\left\{ P(A|\alpha,y_{0}=d_{0},x^{T},z^{T},x_{2}=x_{3},z_{2}=z_{3})-P(B|\alpha,y_{0}=d_{0},x^{T},z^{T},x_{2}=x_{3},z_{2}=z_{3})\right\} \nonumber \\
= & \text{sgn}\left\{ \gamma(d_{3}-d_{0})+x_{21}^{\prime}\beta+\varpi z_{21}\right\},\label{eq:HK}
\end{align}
based on which the point identification can be established. Through comparing (\ref{eq:HK}) for HK and (\ref{eq:iden_ineq1}) for our estimator,  we see that both approaches impose restrictive conditions on the observed covariates. To achieve point identification, our method requires Assumption A3, and HK require that there exists at least one relevant, continuous element of $(x_{21},z_{21})$ can vary freely on a large support conditional on $\{x_{2}=x_{3},z_{2}=z_{3}\}$, $\alpha$, and all other covariates. Note that conditioning on
$\{x_{2}=x_{3},z_{2}=z_{3}\}$ in (\ref{eq:HK}) essentially excludes covariates with non-overlapping support over time (e.g., time trend or time dummies for controlling fixed time effects) and may suffer from the curse of dimensionality when there are many relevant covariates. Our method avoids these
two issues because we do not match $(x_{t},z_{t})$ over time.

We end this discussion by summarizing the advantages and disadvantages of our method, compared with
HK.

\noindent $\mathbf{Pros}$
\begin{enumerate}
\item Our method allows for general forms of time effects, while
HK do not.
\item Our estimator can be applied to models with many covariates and has a fixed rate of convergence,
while the applicability and the rate of convergence of HK's estimator deteriorate as the number of covariates grows.
\end{enumerate}
$\mathbf{Cons}$
\begin{enumerate}
\item Our method lacks a practical guide for choosing the tuning parameter
$\sigma_{n}$, which relies on unknown tail distribution of $z_2$.
This is a common feature for all estimators built on the ``identification at infinity''.
\item The convergence rate of our estimator is generally unknown. The rate of convergence of HK's estimator is determined by the choice of tuning parameter and the number of continuous covariates.
\end{enumerate}

\subsection{Comparing with OY}\label{appendix0_2}
To point identify model (\ref{model}), OY require $T\geq4$ and the process
of $\{x_{t},z_{t}\}$ to be serially independent and strictly stationary,
conditional on $\alpha$, in addition to the conditions stated in HK’s Theorem 4. The identification proceeds in two steps. First, $(\beta,\varpi)$
can be identified based on the following identification equation:
\begin{align}
 & \textrm{sgn}\{P(y_{3}=1|x_{1},z_1,x_{3},z_3,y_{0}=y_{2}=y_{4},\alpha)-P(y_{1}=1|x_{1},z_1,x_{3},z_3,y_{0}=y_{2}=y_{4},\alpha)\}\nonumber \\
= & \textrm{sgn}\{x_{31}'\beta+\varpi z_{31}\}.\label{eq:OY}
\end{align}
After identifying $(\beta,\varpi)$ from (\ref{eq:OY}), one can further use (\ref{eq:HK}), with
$\{x_{2}=x_{3},z_{2}=z_{3}\}$ replaced by $\{x_{3}'\beta+\varpi z_{3}=x_{2}'\beta+\varpi z_{2}\}$,
to identify $\gamma$. OY's approach also avoids the curse of dimensionality that HK's method faces, as it only needs to match $y_{t}$ and the index $x_{t}'\beta+\varpi z_{t}$ over time to identify $(\beta,\varpi)$ and $\gamma$, respectively. However, unlike our method, OY assume $\left(x_{t},z_{t}\right)$ to be serially independent conditional on $\alpha$, which rules out general time effects. Note that if we assume $(x_{t},z_{t})$ are i.i.d. conditional on $\alpha$ over $t$, Assumptions A2 and A3 can be significantly simplified to the statement that ``\textit{$z_{2}$ has unbounded support conditional on $(\alpha,x_{2})$}''. Specifically, Assumption A3 can be met by directly setting $\xi_{31} = z_{31}$.

We conclude this discussion by summarizing the advantages of our method over OY.
\noindent $\mathbf{Pros}$
\begin{enumerate}
\item Our method allows for general forms of time effects, while OY do not.
\item Our method needs $T=3$, as a minimum, while OY's identification requires $T \geq 4$.
\item Our method allows for general forms of serial dependence of the covariatess, while OY impose specific restrictions on their serial dependence, as stated in their Assumptions SI and SD.
\end{enumerate}
The disadvantages of our method compared with OY are the same as those listed
in the last section. We omit them to avoid repetition.

\subsection{Discussions on Assumptions A and B}\label{appendix0_3}

\subsubsection{On the Role of Assumption A2 in Identification}\label{appendix0_31}
Assumption A2 plays a pivotal role in applying the identification
at infinity strategy. From the derivation of equation (\ref{eq:iden_ineq1}), we can
see that this assumption guarantees that $z_{2}$ can be large enough
(assuming $d_{2}=1$ and $\varpi>0$) for the direct effects of $y_{1}$
on $y_{2}$ in events C and D to be almost negligible, i.e., $F_{\epsilon|\alpha}(\alpha+\gamma+x_{2}'\beta+\varpi z_{2})\approx F_{\epsilon|\alpha}(\alpha+x_{2}'\beta+\varpi z_{2})$
($\approx1$). It is worth noting that Assumption A2 is a sufficient
but not necessary condition. For example, if $F_{\epsilon|\alpha}(\cdot)$
has bounded support, the two indexes in it need only be close enough
to the boundary of the support, and $z_{2}$ may not have to be very
large. However, we can see from this example that relaxing Assumption
A2 may require additional restrictions on the parameter space $\Theta$,
the support of $x_{2}$, and the distribution of $(\epsilon,\alpha)$.
This paper deliberately avoids imposing distributional restrictions
on unobservables. The price for this generality is the need for strong
assumptions on observed variables, i.e., $z_{2}$, ensuring the applicability
of our method in the most general cases.

\subsubsection{On the Role of Assumption A3 in Identification}\label{appendix0_32}
Assumption A3 serves as the ``full (large) support'' condition commonly employed by maximum score estimators. Specifically, Assumptions A1 and A2 establish the identification equation (\ref{eq:iden_ineq1}), indicating that the sign of $\gamma d_{20}+x_{31}^{\prime}\beta+\varpi z_{31}$ determines the rank order of the conditional probabilities of events C and D. This implies that the true parameter $\theta$ maximizes the population criterion defined in (\ref{eq:pop_obj}). However, any other parameter $\vartheta\neq\theta$ can also be the maximizer of (\ref{eq:pop_obj}) as long as it satisfies (\ref{eq:iden_ineq1}) almost surely. Assumption A3 rules out such a possibility, ensuring that for any $\vartheta\neq\theta$, $rd_{20}+x_{31}'b+wz_{31}$ and $\gamma d_{20}+x_{31}'\beta+\varpi z_{31}$ can have opposite signs with positive probability, as seen in our proof of Theorem \ref{T:identify}. This guarantees that any $\vartheta\neq\theta$ will have a positive probability of violating the identification equation (\ref{eq:iden_ineq1}), establishing the point identification of $\theta$. Therefore, as argued in \cite{manski-87}, this condition ``prevents a local failure of the (point) identification.''

It is essential to note that the $\xi_{31}$ in Assumption A3 does not necessarily have to be $z_{31}$. Any element in $(x_{31},z_{31})$ that can vary freely on a large support conditioning on other covariates, $\alpha$,
and $\{z_{2}>\sigma\}\cup\{z_{2}<-\sigma\}$, suffices for the point identification of $\theta$. However, in situations where no such continuous covariate exists, and all covariates except $z_t$ are, for example, binary, we have no choice but to require $z_t$ to satisfy both Assumptions A1 and A2. An immediate example for such a case is the empirical application presented in Section \ref{sec_application}.

In this application, other than the coefficient before $y_{t-1}$, we only have the coefficient before a dummy variable to be statistically significant. By ignoring all the insignificant covariates in the model, we are left with a univariate, binary regressor $x_t$, and $z_t$ is the only continuous variable (then, $\xi_{31}=z_{31}$ in Assumption A3). Thus, the model reduces to
\begin{equation}
y_{t}=\mathds{1}\left\{ \gamma y_{t-1}+\beta x_{t}+\varpi z_{t}\geq\varepsilon_{it}\right\}. \label{eq:A3_1}
\end{equation}
In what follows, we assume (\ref{eq:A3_1}) is the true model, and use it as an illustrating example to show how we can relax the full support condition imposed in Assumption A3. In fact, if we have prior knowledge about $x_t$ and the parameter space $\Theta$, like in (\ref{eq:A3_1}), we do not need the support of $z_{31}$ to be the whole
real line. Specifically, we can achieve the point identification of model (\ref{eq:A3_1}) under Assumption A3' below, which is weaker than the Assumption A3 stated in the main text for a general model.

\bigskip

\noindent \textbf{Assumption A}3':  The distribution of $z_{31}$, conditional on $(\alpha,x_{31})$
and $\{z_{2}>\sigma\}\cup\{z_{2}<-\sigma\}$ as $\sigma\rightarrow+\infty$, has support that contains an interval $\left[-K,L\right]$,
such that $K,L\geq0$ and $\max\left\{ K,L\right\} \geq\frac{\max\left\{ |\beta|,|\gamma|\right\} }{|\varpi|}+\delta$
for some small positive $\delta$.

\bigskip

We show the identification under Assumptions A1, A2, A3', A4, and
A5 in the following.
\begin{proof}[Proof of Theorem \ref{T:identify} for (\ref{eq:A3_1})]
Without loss of generality, $K\geq\frac{\max\left\{ \beta,\gamma\right\} }{\varpi}+\delta$,
and $\beta,\gamma,\varpi>0$; other cases are similar. Note that $d_{20}$
and $x_{31}$ can take only three values, $-1,0,\text{\text{and} }1$,
and all cases happen with positive probabilities. For any $r/w\neq\gamma/\varpi$,
take $d_{20}=1$ and $x_{31}=0$, then
\[
P\left[\left\{ -\frac{r}{w}<z_{31}<-\frac{\gamma}{\varpi}\right\} \cup\left\{ -\frac{\gamma}{\varpi}<z_{31}<-\frac{r}{w}\right\} \right]>0,
\]
due to the assumption that the support of $z_{31}$ contains $\left[-\frac{\gamma}{\varpi}-\delta,0\right]$
which is a subset of $[-K,0]$. For any $b\neq\beta,$ take $d_{20}=0$
and $x_{31}=1$, then
\[
P\left[\left\{ -\frac{b}{w}<z_{31}<-\frac{\beta}{\varpi}\right\} \cup\left\{ -\frac{\beta}{\varpi}<z_{31}<-\frac{b}{w}\right\} \right]>0,
\]
due to the assumption that the support of $z_{31}$ contains $\left[-\frac{\beta}{\varpi}-\delta,0\right]$
which is a subset of $[-K,0]$. Then, the identification of model (\ref{eq:A3_1}) follows by combining these results with the arguments used in the proof of Theorem \ref{T:identify}, presented in Appendix \ref{appendixA}, for a general model.
\end{proof}

A final note is that the full support condition imposed in Assumption A3 is sufficient but not necessary. Such condition simplifies the proof and exposition of identification and is widely adopted in maximum-score literature. However, our Assumption A3' provides an example where we can relax such full support conditions to a large enough bounded support in special cases where we possess more knowledge about the covariates and the parameter space of the model.

\subsubsection{On the Role of Assumption A in Empirical Application}\label{appendix0_33}
In our empirical application, we use the standardized household income as the $z_t$ in Assumptions A2 and A3’. In any four consecutive waves $(t-2,t-1,t,t+1)$, $I_t$ and $I_{t+1}-I_{t-1}$ play the roles of $z_2$ and $z_{31}$ in these assumptions, respectively. To apply the identification at infinity strategy, we need $I_t$ to satisfy Assumption A2, meaning that it can reach sufficiently large values in the sample. Moreover, for point identification, we also need $I_{t+1}-I_{t-1}$ to meet Assumption A3’, which is specific to our application where the only significant covariates (besides income) are the lagged dependent variable and the age dummy. We provide summary statistics to support the validity of these assumptions in the context of our application. In particular, Table \ref{tab_sum_eff} reproduces Table \ref{tab_sum}, but only uses observations for ``switchers with $I_t>\sigma_n$'', which are the effective sample for our estimation. Here, $\sigma_n=1.478$ is the tuning parameter that we use to obtain the results in the top panel of Table \ref{tab_res}.

\begin{table}[tpb]
\caption{Summary Statistics (Effective Sample)}
\label{tab_sum_eff}%\small
\centering
\begin{tabular}{lccccc}
\hline\hline
Variable & $n\times T$ & Mean & Std.Dev. & Min & Max \\
\midrule $y_{it}$ & 2160 & 0.509 & 0.500 & 0 & 1 \\
$\text{Above30}_{it}$ & 2160 & 0.678 & 0.468 & 0 & 1 \\
$\text{Age}_{it}$ & 2160 & 40.080 & 14.372 & 18 & 83 \\
$\text{GCC}_{it}$ & 2160 & 0.761 & 0.427 & 0 & 1 \\
$D_{2014,t}$ & 2160 & 0.143 & 0.350 & 0 & 1 \\
$I_{t+1}-I_{t-1}$ & 2160 & -0.126 & 1.524 & -11.877 & 10.325 \\\hline\hline
\end{tabular}%
\end{table}

\begin{figure}[H]
\centering
\includegraphics[width=0.75\textwidth]{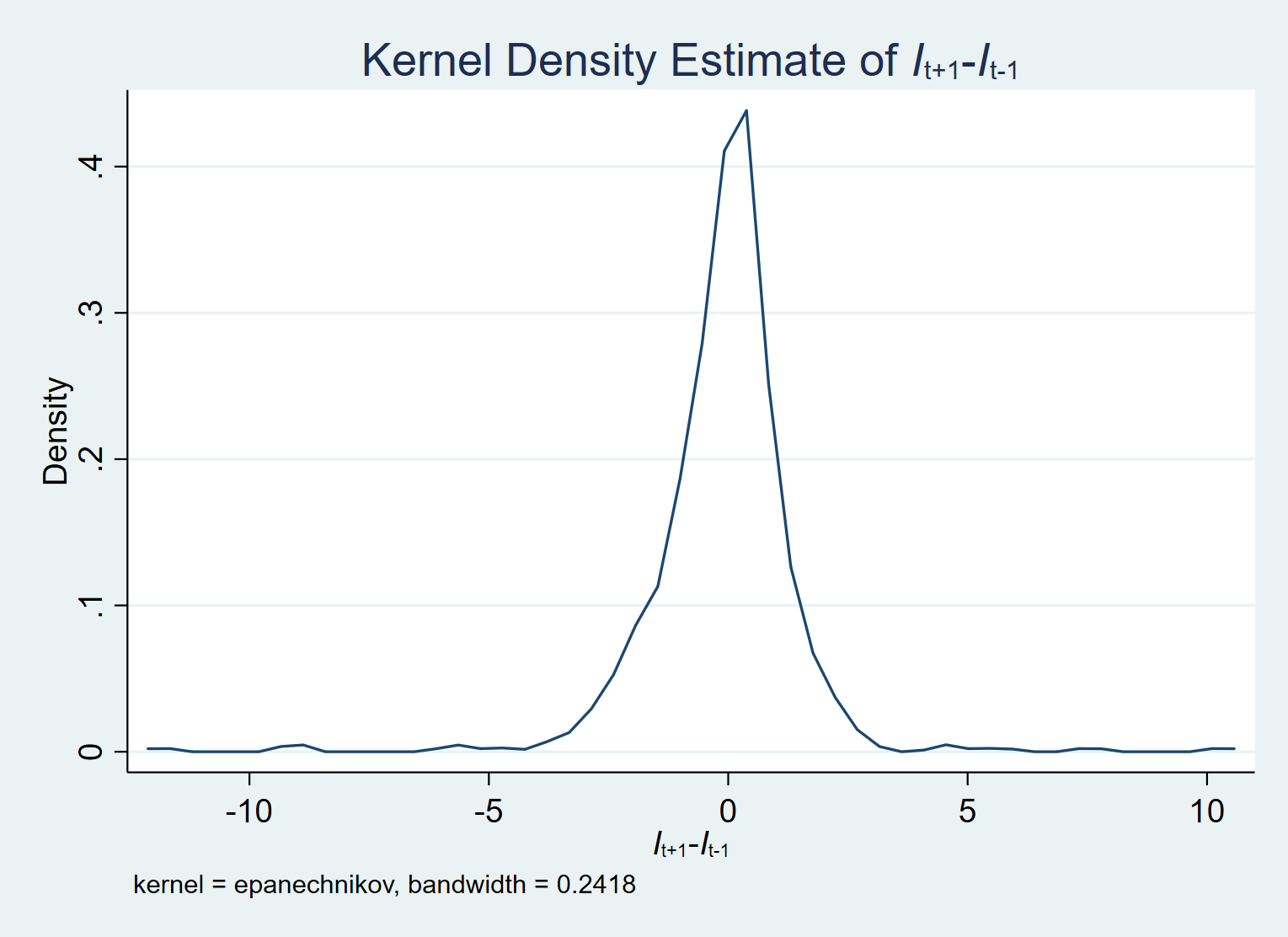}
\caption{\label{figA1}Kernel Density Estimate of $I_{t+1}-I_{t-1}$}
\end{figure}

We can see that only 2160 observations are left once the sample is restricted to ``switchers with $I_{t}>\sigma_{n}$''. This small sample size and the slow convergence rate of the estimator account for the wide confidence intervals (CI) in Table \ref{tab_res}. Moreover, this sub-sample has a lower average age and a higher proportion of GCC residents than the full sample in Table \ref{tab_sum_eff}. We omit the summary statistics for income because this sub-sample is defined with an income threshold. Instead, we calculate the range of $I_{t+1}-I_{t-1}$, which is $[-11.877,10.325]$. Moreover, to provide a better sense
of its variation in this sub-sample, we also calculate its $\{5\%,10\%,25\%,50\%,75\%,90\%,95\%\}$
quantiles, which are 
\[
\{-2.272,-1.710,-0.708,0.043,0.567,1.232,1.749\},
\]
respectively, and visualize its PDF via kernel density estimation as illustrated in Figure \ref{figA1}. These imply that the ``$z_{31}$'' in our application has a large support. In fact, to estimate the model coefficients, our approach needs observations of either low-income households that start purchasing PHC due to an increase in income or high-income households that stop maintaining PHC due to an income decline.

One important implication of Assumption A is that we should use observations with large absolute values of $I_{t}$. For our proposed estimator, this is controlled by the tuning parameter (threshold) $\sigma_{n}$. Recall that in the main text, we use $\sigma_{n} = c \cdot \widehat{\textrm{std}(I_{it})}\sqrt{\log n^{*} / 2.95}$ with $c = 1.0$ and $1.1$. We present Table \ref{tab_res2}, which summarizes empirical results obtained using smaller $\sigma_{n}$, i.e., by taking $c = 0.5,0.7$, and $0.9$. These estimations also use a small proportion of observations in the left tail. However, considering the ``bell'' shape of the PDF of $I_{t}$, a smaller $\sigma_{n}$ includes significantly more observations in the estimation.

In general, a larger $\sigma_{n}$ results in a smaller bias but a wider confidence interval (CI), while a smaller $\sigma_{n}$ leads to a larger bias but a narrower CI, reflecting the typical bias-variance trade-off. The results in Table A.2 illustrate this. Although these results exhibit a pattern similar to those in Table 6—with $y_{it-1}$ and $\textrm{Above30}_{it}$ showing much stronger (positive) effects compared to other covariates—the coefficient estimates appear biased toward zero, with CIs skewed to the right of the origin and including zero. We observe narrower CIs for coefficients estimated with $c=0.7$ and $0.9$ as expected, but wider CIs with $c=0.5$. One possible explanation is that a too-small $\sigma_n$ fails to achieve ``identification at infinity'', likely causing our method to estimate the coefficients of $I_t$ biased toward zero, especially during resampling. Recall that we perform scale normalization by setting the coefficient of $I_t$ to 1 (equivalent to dividing the other coefficients by the $I_t$ coefficient). Consequently, $m$-out-of-$n$ resampling is more likely to generate extreme values, resulting in longer CIs.

\begin{table}[htbp]
\caption{Estimates of Preference Coefficients with Smaller $\sigma_n$}
\label{tab_res2}\centering
\begin{tabular}{clLrcrc}
\hline\hline
                   & Variable & \multicolumn{1}{c}{Estimate} &  \multicolumn{2}{c}{[90\% Conf.Int.]} &  \multicolumn{2}{c}{[95\% Conf.Int.]}  \\
    \midrule
        \multirow{5}[2]{*}{$c = 0.5$} & $y_{it-1}$ & 2.598 &	-5.993 &	 10.512 & -6.632 & 14.059 \\
          & $\textrm{Above30}_{it}$ & 2.186 &	-7.161 & 10.420 & -7.634 & 12.606 \\
          & $\textrm{Age}_{it}$  & 0.099 &	-11.691 & 5.933 & -12.417 & 8.355 \\
          & $\textrm{GCC}_{it}$ & 0.139 &	-11.505 & 4.864 & -12.328 & 7.765 \\
          & $D_{2014,t}$ & 0.207 &	-11.225 & 5.019 & -11.884 & 7.121 \\
\hline
\multirow{5}[2]{*}{$c = 0.7$} & $y_{it-1}$ & 3.524 & -1.072 & 11.118 & -1.459 & 13.035 \\
          & $\textrm{Above30}_{it}$ & 3.838 &	-3.122 & 9.738 & -3.608 & 10.797 \\
          & $\textrm{Age}_{it}$ & -0.455 &	 -8.779 & 4.575 & -9.249	& 5.956 \\
          & $\textrm{GCC}_{it}$   & -1.110 & -9.764 & 2.404 & -10.320 & 3.993 \\
          & $D_{2014,t}$ & 1.091 &	-5.727 & 7.831 & -6.152 & 9.105 \\
          \hline
\multirow{5}[2]{*}{$c = 0.9$} & $y_{it-1}$ & 4.248^{\ast} & 0.090 & 14.646 & -0.339 & 15.668\\
          & $\textrm{Above30}_{it}$ & 4.691 & -2.805 & 11.963 & -3.042 & 12.665 \\
          & $\textrm{Age}_{it}$ & -0.473 &	 -8.757	& 5.530 &	-9.222 & 6.734 \\
          & $\textrm{GCC}_{it}$   & -1.820 & -11.585	& 3.583 &	-11.796 & 4.601 \\
          & $D_{2014,t}$ & -2.508 & -12.530 & 	1.900 &	-12.806 & 2.559 \\
\hline\hline
    \end{tabular}
\end{table}

\subsubsection{Sufficient Conditions for Assumption B3}\label{appendix0_34}
We provide the following sufficient conditions for Assumption B3:
\begin{enumerate}
\item[1] The limit of the joint distribution of $(x^{T},z_{1},z_{3})$,
conditional on $z_{2}>\sigma$, as $\sigma\rightarrow +\infty$, exists.
\item[2] For any element $v_{31}$ in $(x_{31},z_{31})$, the first moment
of $\vert v_{31}\vert$, conditional on $z_{2}>\sigma$ and $(x_{31},z_{31})\setminus v_{31}$,
is uniformly bounded for all $\sigma>0$ and $(x_{31},z_{31})\setminus v_{31}$.
\end{enumerate}
It is clear that Condition 1 implies Assumption B3(i). Note that if we remove the condition $z_{2}>\sigma$, Assumption B3(ii) becomes a classic Lipschitz condition. In what follows, we
demonstrate that Condition 2 is sufficient for Assumption B3(ii).
To ease our exposition, we assume $y_{0}$ is fixed and illustrate
the claim with the following expression:
\[
\tilde{\Lambda}(\vartheta):=y_{2}\mathds{1}[r(y_{2}-y_{0})+x_{31}'b+wz_{31}>0].
\]
Note that $(y_{3}-y_{1})\tilde{\Lambda}(\vartheta)$ is the $\Lambda(\vartheta)$
defined in Assumption B3.

Apply the law of iterated expectation to write
\begin{equation}
\mathbb{E}[\tilde{\Lambda}(\vartheta)|z_{2}>\sigma]=\mathbb{E}[\mathbb{E}[\tilde{\Lambda}(\vartheta)|x_{31},z_{31},z_{2}>\sigma]|z_{2}>\sigma].\label{eq:B3_1}
\end{equation}
Note that the inner expectation can be further expressed as:
\begin{align}
 & \mathbb{E}[\tilde{\Lambda}(\vartheta)|x_{31},z_{31},z_{2}>\sigma]\nonumber \\
= & \mathbb{E}[\tilde{\Lambda}(\vartheta)|y_{2}=1,x_{31},z_{31},z_{2}>\sigma]P(y_{2}=1|x_{31},z_{31},z_{2}>\sigma)\nonumber \\
 & +\mathbb{E}[\tilde{\Lambda}(\vartheta)|y_{2}=0,x_{31},z_{31},z_{2}>\sigma]P(y_{2}=0|x_{31},z_{31},z_{2}>\sigma)\nonumber \\
= & \mathbb{E}[\tilde{\Lambda}(\vartheta)|y_{2}=1,x_{31},z_{31},z_{2}>\sigma]P(y_{2}=1|x_{31},z_{31},z_{2}>\sigma)=\mathds{1}[\xi'\vartheta>0]\varphi(x_{31},z_{31},\sigma),\label{eq:B3_2}
\end{align}
where $\xi:=(1-y_{0},x_{31},z_{31})$ and $\varphi(x_{31},z_{31},\sigma):=P(y_{2}=1|x_{31},z_{31},z_{2}>\sigma)$. Then, substituting (\ref{eq:B3_2}) into (\ref{eq:B3_1}) yields
\begin{equation}
\mathbb{E}[\tilde{\Lambda}(\vartheta)|z_{2}>\sigma]=\int\mathds{1}[\xi'\vartheta>0]\varphi(x_{31},z_{31},\sigma)f_{\sigma}(x_{31},z_{31})dx_{31}dz_{31},\label{eq:B3_3}
\end{equation}
where $f_{\sigma}\left(x_{31},z_{31}\right):= f\left(\left.x_{31},z_{31}\right|z_{2}>\sigma\right)$.

Applying standard results from classic differential geometry to equation (\ref{eq:B3_3}) yields
\begin{align}
 & \left|\frac{\partial}{\partial\vartheta}\mathbb{E}[\tilde{\Lambda}(\vartheta)|z_{2}>\sigma]\right| \nonumber \\
= & \left|\int\xi\cdot\mathds{1}[\xi'\vartheta=0]\varphi(x_{31},z_{31},\sigma)f_{\sigma}(x_{31},z_{31})d\Delta_{0}\right|\leq\int\vert\xi\vert\cdot\mathds{1}[\xi'\vartheta=0]f_{\sigma}(x_{31},z_{31})d\Delta_{0}, \label{eq:B3_4}
\end{align}
where $\varDelta_{0}$ denotes the surface measure on $\left\{ \xi:\xi'\vartheta=0\right\}$. Since $\left\{ \xi:\xi'\vartheta=0\right\} $ is a hyperplane in $\mathbb{R}^{p+2}$,
$\varDelta_{0}=\left\Vert \vartheta\right\Vert$ and has a constant
density (see, e.g., Example 6.4 of \cite{KimPollard1990}). Under Assumption A5, the parameter space $\Theta$ of $\vartheta$ is a compact set, and thus $\Delta_{0}$ is (totally) bounded. It
is evident that Assumption B3(ii) is satisfied if $\left|\frac{\partial}{\partial\vartheta}\mathbb{E}[\tilde{\Lambda}(\vartheta)|z_{2}>\sigma]\right|$ is uniformly bounded. Recall that $\xi$ is defined as $(1-y_{0},x_{31},z_{31})$.
Then, by (\ref{eq:B3_4}), each element of $\left|\frac{\partial}{\partial\vartheta}\mathbb{E}[\tilde{\Lambda}(\vartheta)|z_{2}>\sigma]\right|$ is essentially bounded by a (weighted) first conditional moment of the corresponding element of $|(1-y_{0},x_{31},z_{31})|$ over $\left\{ \xi'\vartheta=0\right\} $. Then, Condition 2 suffices and the desired result follows.

\section{Proofs of Theorems 2.1 and 4.1 (Identification and Consistency)}\label{appendixA}
In this appendix, we prove Theorems \ref{T:identify} (identification) and \ref{T:consistency} (consistency). For ease of exposition, all proofs are
based on identification equation (\ref{eq:iden_ineq1}) and sample objective
function (\ref{eq:Qn1}). The same arguments apply to sample objective function (\ref{eq:Qn2}) and its corresponding identification equation.

\begin{proof}[Proof of Theorem \ref{T:identify}]

It suffices to prove the identification of $\theta$ based on equation
(\ref{eq:iden_ineq1}). According to (\ref{eq:iden_ineq1}), for any $(\alpha,y_{0},x^{T},z_{1},z_{3})$, we
can find a sufficiently large $\sigma$ such that the sign of $\gamma d_{20}+x_{31}'\beta+\varpi z_{31}$
matches the sign of $P(C|\alpha,y_{0}=d_{0},x^{T},z^{T})-P(D|\alpha,y_{0}=d_{0},x^{T},z^{T})$,
given $z_{2}>\sigma$. Note that the probability difference is independent
of the preference coefficients. Hence, the true parameter $\theta$
maximizes
\[
\bar{Q}_{1}(\vartheta):=\lim_{\sigma\rightarrow+\infty}\mathbb{E}\left[\left(P(C|\alpha,y_{0}=d_0,x^{T},z^{T})-P(D|\alpha,y_{0}=d_0,x^{T},z^{T})\right)\cdot\text{sgn}\left(rd_{20}+x_{31}^{\prime}b+wz_{31}\right)|z_{2}>\sigma\right],
\]
where we let $\sigma\rightarrow+\infty$ to ensure that (\ref{eq:iden_ineq1}) holds
for all $(\alpha,y_{0},x^{T},z_{1},z_{3})$.

To further show that $\theta$ attains a unique maximum, consider
any $\vartheta\in\Theta$ such that $\bar{Q}_{1}(\vartheta)=\bar{Q}_{1}(\theta)$.
We want to show that $\vartheta=\theta$ must hold under Assumption
A.

Under Assumption A3, assuming $\xi_t$ is an element of $x_t$ w.l.o.g., we can write, for the true parameter $\theta$,
\[
\gamma y_{t-1}+x_t'\beta+\varpi z_t=\gamma y_{t-1}+c^{*}\xi_t+\tilde{x}_t'\tilde{\beta}+\varpi z_t,
\]
and for any $\vartheta\in\Theta$,
\[
r y_{t-1}+x_t'b+w z_t=r y_{t-1}+c \xi_t+\tilde{x}_t'\tilde{b}+w z_t,
\]
where $\tilde{x}_t=x_t\left\backslash \xi_t\right.$, $\tilde{\beta}=\beta\left\backslash c^{*}\right.$, and $\tilde{b}=b\left\backslash c \right.$. Then, note that if
\begin{align*}
\lim_{\sigma\rightarrow+\infty}P\left[\left\{ \frac{rd_{20}+\tilde{x}_{31}^{\prime}\tilde{b}+wz_{31}}{-c}<\xi_{31}<\frac{\gamma d_{20}+\tilde{x}_{31}^{\prime}\tilde{\beta}+\varpi z_{31}}{-c^{*}}\right\} \cup\right.\\
\left.\left.\left\{ \frac{\gamma d_{20}+\tilde{x}_{31}^{\prime}\tilde{\beta}+\varpi z_{31}}{-c^{*}}<\xi_{31}<\frac{rd_{20}+\tilde{x}_{31}^{\prime}\tilde{b}+wz_{31}}{-c}\right\} \right|z_{2}>\sigma\right] & >0,
\end{align*}
$\vartheta$ and $\theta$ yield different values of the $\text{sgn}(\cdot)$
function in $\bar{Q}_{1}(\cdot)$ with strictly positive probability
under Assumption A3, and hence $\bar{Q}_{1}(\vartheta)<\bar{Q}_{1}(\theta)$.
This observation implies that for all $\vartheta\in\Theta$ satisfying
$\bar{Q}_{1}(\vartheta)=\bar{Q}_{1}(\theta)$, we must have
\begin{align*}
\lim_{\sigma\rightarrow+\infty}P\left[\left\{ \frac{rd_{20}+\tilde{x}_{31}^{\prime}\tilde{b}+wz_{31}}{-c}<\xi_{31}<\frac{\gamma d_{20}+\tilde{x}_{31}^{\prime}\tilde{\beta}+\varpi z_{31}}{-c^{*}}\right\} \cup\right.\\
\left.\left.\left\{ \frac{\gamma d_{20}+\tilde{x}_{31}^{\prime}\tilde{\beta}+\varpi z_{31}}{-c^{*}}<\xi_{31}<\frac{rd_{20}+\tilde{x}_{31}^{\prime}\tilde{b}+wz_{31}}{-c}\right\} \right|z_{2}>\sigma\right] & =0,
\end{align*}
which, by Assumption A3, is equivalent to
\[
\lim_{\sigma\rightarrow+\infty}P\left[\left.\left(r-\gamma\right)d_{20}+\tilde{x}_{31}^{\prime}(\tilde{b}-\tilde{\beta})+\left(c-c^{*}\right)\xi_{31}+\left(w-\varpi\right)z_{31}=0\right|z_{2}>\sigma\right]=1.
\]
Then the desired result follows from Assumption A4.
\end{proof}

We next prove the consistency of $\hat{\theta}_{n}$. For ease of
illustration, we work with the following sample and population objective
functions with a bit abuse of notation:
\begin{equation}
Q_{n1}(\vartheta ):=\frac{1}{nP\left( z_{2}>\sigma _{n}\right) }%
\sum_{i=1}^{n}\mathds{1}\{z_{i2}>\sigma _{n}\}\cdot \Lambda _{i}(\vartheta )
\label{eq:appendix_obj1}
\end{equation}%
and
\begin{equation}
Q_{1}(\vartheta ):=\lim_{\sigma \rightarrow +\infty }\mathbb{E}[\Lambda
(\vartheta )|z_{2}>\sigma ],
\end{equation}
where $\Lambda _{i}(\vartheta ):=y_{i2}y_{i31}\cdot \mathds{1}%
\{ry_{i20}+x_{i31}^{\prime }b+wz_{i31}>0\}$. Note that (\ref%
{eq:appendix_obj1}) and (\ref{eq:Qn1}) have the same maximum, and hence are
equivalent.

\begin{proof}[Proof of Theorem \ref{T:consistency}]
We prove the consistency of $\hat{\theta}_n$ via verifying the four
sufficient conditions for applying Theorem 2.1 of \cite{NeweyMcFadden1994}:
(S1) $\Theta$ is compact, (S2) $\sup_{\vartheta\in\Theta}\vert
Q_{n1}(\vartheta)-Q_{1}(\vartheta)\vert=o_{p}(1)$, (S3) $Q_{1}(\vartheta)$
is continuous in $\vartheta$, and (S4) $Q_{1}(\vartheta)$ is uniquely
maximized at $\theta$.

The compactness condition (S1) is satisfied by Assumption A5. The
identification condition (S4) follows from Theorem \ref{T:identify}. To see
this, let $\bar{\chi}^{\prime }\vartheta :=ry_{20}+x_{31}^{\prime }b+wz_{31}$
and note that by definition, when $d_{2}=1$ in (\ref{eq:iden_ineq1}), we
have
\begin{align}
& Q_{1}(\vartheta )=\lim_{\sigma \rightarrow +\infty }\mathbb{E}\left[
y_{2}y_{31}\cdot \mathds{1}\{\bar{\chi}^{\prime }\vartheta >0\}|z_{2}>\sigma %
\right]  \notag \\
=& \lim_{\sigma \rightarrow +\infty }\sum_{d_{0}}\mathbb{E}\left[
y_{2}y_{31}\cdot \mathds{1}\{\bar{\chi}^{\prime }\vartheta >0\}|z_{2}>\sigma
,y_{0}=d_{0}\right] P(y_{0}=d_{0}|z_{2}>\sigma )  \notag \\
=& \lim_{\sigma \rightarrow +\infty }\sum_{d_{0}}\mathbb{E}\left\{ \mathbb{E}%
\left[ y_{2}y_{31}\cdot \mathds{1}\{\bar{\chi}^{\prime }\vartheta
>0\}|\alpha ,y_{0}=d_{0},x^{T},z^{T}\right] |z_{2}>\sigma
,y_{0}=d_{0}\right\} P(y_{0}=d_{0}|z_{2}>\sigma )  \notag \\
=& \lim_{\sigma \rightarrow +\infty }\sum_{d_{0}}\mathbb{E}\left\{ \mathbb{E}%
\left[ \left( \mathds{1}\{C\}-\mathds{1}\{D\}\right) \cdot \mathds{1}\{\bar{%
\chi}^{\prime }\vartheta >0\}|\alpha ,y_{0}=d_{0},x^{T},z^{T}\right]
|z_{2}>\sigma ,y_{0}=d_{0}\right\} P(y_{0}=d_{0}|z_{2}>\sigma )  \notag \\
=& \lim_{\sigma \rightarrow +\infty }\sum_{d_{0}}\mathbb{E}\left\{ \left[
P(C|\alpha ,y_{0},x^{T},z^{T})-P(D|\alpha ,y_{0},x^{T},z^{T})\right] \cdot %
\mathds{1}\{y_{0}=d_{0}\}\cdot \mathds{1}\{\bar{\chi}^{\prime }\vartheta
>0\}|z_{2}>\sigma \right\}  \notag \\
=& \lim_{\sigma \rightarrow +\infty }\mathbb{E}\left\{ \left[ P(C|\alpha
,y_{0},x^{T},z^{T})-P(D|\alpha ,y_{0},x^{T},z^{T})\right] \cdot \mathds{1}\{%
\bar{\chi}^{\prime }\vartheta >0\}|z_{2}>\sigma \right\} =\frac{1}{2}\bar{Q}%
_{1}(\vartheta )+c,  \label{eq:iden_equiv}
\end{align}%
where $\bar{Q}_{1}(\vartheta )$ is defined in the proof of Theorem \ref{T:identify} and $c$ is an absolute constant.

We now verify the uniform convergence condition (S2). Let $\mathcal{F}%
_{n}:=\{\mathds{1}\{z_{2}>\sigma _{n}\}\cdot \Lambda (\vartheta )|\vartheta
\in \Theta \}$, which is clearly a sub-class of the fixed class $\mathcal{F}%
:=\{\mathds{1}\{z_{2}>\sigma \}\cdot \Lambda (\vartheta )|\vartheta \in
\Theta ,\sigma >0\}$. First, note that the collection of right-sided
half-intervals, $\mathcal{C}:=\{(\sigma ,+\infty )|\sigma \in \mathbb{R}%
_{+}\}$, is a Vapnik-\v{C}ervonenkis (VC) class with VC-index (or called
VC-dimension) $V(\mathcal{C})=2$.\footnote{%
Here we use the same definition of VC-index as \cite{vdVaartWellner1996},
\cite{kosorok2008}, and \cite{gine_nickl_2015}. Some more recent literature,
such as \cite{VershyninHDP} and \cite{WainwrightHDS}, uses a slightly
different definition (i.e., $V(\mathcal{C})-1$ in our notation).} Then by
Lemma 9.8 of \cite{kosorok2008}, the class $\mathcal{F}_{1}:=\{\mathds{1}%
\{z_{2}\in C\}|C\in \mathcal{C}\}$ of indicator functions is VC-subgraph
with envelope $F_{1}=1$ and VC-index $V(\mathcal{F}_{1})=V(\mathcal{C})=2$.
Next applying Lemma 9.8 of \cite{kosorok2008}, Lemmas 2.6.15, 2.6.18 (vi),
and Problem 2.6.12 of \cite{vdVaartWellner1996} to obtain that the class $%
\mathcal{F}_{2}:=\{\Lambda (\vartheta )|\vartheta \in \Theta \}$ is a
VC-subgraph class of functions with envelope $F_{2}=1$, whose VC-index $V(%
\mathcal{F}_{2})\leq 2p+7$. Put all these results together and apply Theorem
9.15 of \cite{kosorok2008} (or equivalently Theorem 2.6.7 of \cite%
{vdVaartWellner1996}) to conclude that the class $\mathcal{F}$ has bounded
uniform entropy integral (BUEI) with envelope $F:=F_{1}F_{2}=1$, i.e., the
covering number of $\mathcal{F}$ satisfies
\begin{equation}
\sup_{\mu }N(\epsilon ,\mathcal{F},L_{2}(\mu ))\leq A(\mathcal{F})\left(
\frac{1}{\epsilon }\right) ^{2(V(\mathcal{F}_{1})+V(\mathcal{F}_{2})-2)}=A(%
\mathcal{F})\left( \frac{1}{\epsilon }\right) ^{2(2p+7)},\text{ }\forall
\epsilon \in (0,1),\footnote{
Here we essentially prove that $\mathcal{F}$ is an Euclidean (manageable) class of
functions with envelope $F=1$ in the sense of \cite{PakesPollard1989} and
\cite{pollard1989}.}  \label{eq:BUEI}
\end{equation}%
with constant $A(\mathcal{F})$ only depending on $\mathcal{F}$, where the
supremum is taken over all probability measures $\mu $.\footnote{%
More precisely, $\mathcal{F} $ can be showed to be BUEI after taking logs, square
roots, and then integrating both sides of (\ref{eq:BUEI}) with respect to $%
\epsilon $.}

Next, note that
\begin{align}
& \sup_{\mathcal{F}_{n}}\mathbb{E}\left[ \left\vert \mathds{1}\{z_{2}>\sigma
_{n}\}\cdot \Lambda (\vartheta )\right\vert \right] =\sup_{\mathcal{F}_{n}}%
\mathbb{E}\left[ \mathds{1}\{z_{2}>\sigma _{n}\}\cdot \left\vert \Lambda
(\vartheta )\right\vert \right]  \notag \\
=& \sup_{\mathcal{F}_{n}}\int \mathds{1}\{z_{2}>\sigma _{n}\}\cdot \mathbb{E}%
\left[ \left\vert \Lambda (\vartheta )\right\vert |z_{2}=z\right]
f_{z_{2}}(z)dz\leq \sup_{\mathcal{F}_{n}}\int \mathds{1}\{z_{2}>\sigma
_{n}\}f_{z_{2}}(z)dz=O\left( \delta _{n}\right) ,  \label{eq:bounded1}
\end{align}%
where $f_{z_{2}}(\cdot )$ denotes the PDF of $z_{2}$ and $\delta
_{n}:=P(z_{2}>\sigma _{n})$.

With (\ref{eq:BUEI}) and (\ref{eq:bounded1}), apply Lemma 5 of \cite%
{honore-k} (see also Theorem 37 in Chapter 2 of \cite{Pollard1984}) to
obtain
\begin{equation*}
\sup_{\mathcal{F}_{n}}\left\vert \frac{1}{n}\sum_{i=1}^{n}\mathds{1}%
\{z_{i2}>\sigma _{n}\}\cdot \Lambda _{i}(\vartheta )-\mathbb{E}\left[ %
\mathds{1}\{z_{2}>\sigma _{n}\}\cdot \Lambda (\vartheta )\right] \right\vert
=O_{p}\left( \sqrt{\frac{\delta _{n}\log n}{n}}\right) =o_{p}\left( \delta
_{n}\right) ,
\end{equation*}%
where the last equality follows by Assumption B2. This then implies that
\begin{equation*}
\sup_{\mathcal{F}_{n}}\left\vert Q_{n1}(\vartheta )-\mathbb{E}\left[ \Lambda
(\vartheta )|z_{2}>\sigma _{n}\right] \right\vert =\sup_{\mathcal{F}%
_{n}}\left\vert Q_{n1}(\vartheta )-\frac{\mathbb{E}\left[ \mathds{1}%
\{z_{2}>\sigma _{n}\}\cdot \Lambda (\vartheta )\right] }{P\left(
z_{2}>\sigma _{n}\right) }\right\vert =o_{p}\left( 1\right) .
\end{equation*}

The remaining task for verifying (S2) is to show $\sup_{\mathcal{F}%
_{n}}\left\vert \mathbb{E}[\Lambda (\vartheta )|z_{2}>\sigma
_{n}]-Q_{1}(\vartheta )\right\vert =o(1)$ and apply triangle inequality.
Recall that $\Theta $ is a compact subset of $\mathbb{R}^{p+2}$. Then for
any $\varepsilon >0$, there exists a finite $\varepsilon /3L$-net $\mathcal{N%
}_{\Theta }\text{($\varepsilon $)}:=\{\vartheta _{1},...,\vartheta
_{N_{\varepsilon }}\}$ of $\Theta $ such that every $\vartheta \in \Theta $
is within a distance $\varepsilon /3L$ of some $\vartheta _{j}\in \mathcal{N}%
_{\Theta }\text{($\varepsilon $)}$, i.e., $\forall \vartheta \in \Theta $, $%
\exists \vartheta _{j}\in \mathcal{N}_{\Theta }\text{($\varepsilon $)}$ such
that $\Vert \vartheta -\vartheta _{j}\Vert \leq \varepsilon /3L$. The
smallest possible cardinality $N_{\varepsilon }$ of $\mathcal{N}_{\Theta }%
\text{($\varepsilon $)}$ can be the covering number of $\Theta $, denoted by
$N(\Theta ,\Vert \cdot \Vert ,\varepsilon /3L)$.\footnote{%
By Lemma 4.2.8 and Corollary 4.2.13 of \cite{VershyninHDP}, we have
\begin{equation*}
N_{\varepsilon }=N(\Theta ,\Vert \cdot \Vert ,\varepsilon /3L)\leq
N(B_{1}^{p+2},\Vert \cdot \Vert ,\varepsilon /6L)\leq \left(
12L/\varepsilon +1\right) ^{p+2},
\end{equation*}%
where $B_{1}^{p+2}$ is an Euclidean ball in $\mathbb{R}^{p+2}$ with radius $%
1 $ containing $\Theta $ under Assumption A5.} Since $\mathbb{E}[\Lambda
(\vartheta )|z_{2}>\sigma _{n}]\rightarrow Q_{1}(\vartheta )$ (and so $\{%
\mathbb{E}[\Lambda (\vartheta )|z_{2}>\sigma _{n}]\}$ is a Cauchy sequence),
for each $\vartheta _{j}$, we can find a positive integer $M_{j}$ so that $|%
\mathbb{E}[\Lambda (\vartheta _{j})|z_{2}>\sigma _{n}]-\mathbb{E}[\Lambda
(\vartheta _{j})|z_{2}>\sigma _{m}]|<\varepsilon /3$ for all $n,m>M_{j}$. We
next verify the Cauchy criterion for uniform convergence, i.e., for all $%
\vartheta \in \Theta $ and $n,m>M:=\max \{M_{1},...,M_{N_{\varepsilon }}\}$,
\begin{align*}
& |\mathbb{E}[\Lambda (\vartheta )|z_{2}>\sigma _{n}]-\mathbb{E}[\Lambda
(\vartheta )|z_{2}>\sigma _{m}]| \\
\leq & |\mathbb{E}[\Lambda (\vartheta )|z_{2}>\sigma _{n}]-\mathbb{E}%
[\Lambda (\vartheta _{j})|z_{2}>\sigma _{n}]|+|\mathbb{E}[\Lambda (\vartheta
_{j})|z_{2}>\sigma _{n}]-\mathbb{E}[\Lambda (\vartheta _{j})|z_{2}>\sigma
_{m}]| \\
& +|\mathbb{E}[\Lambda (\vartheta _{j})|z_{2}>\sigma _{m}]-\mathbb{E}%
[\Lambda (\vartheta )|z_{2}>\sigma _{m}]|\leq \varepsilon /3+\varepsilon
/3+\varepsilon /3=\varepsilon ,
\end{align*}%
where the last inequality follows from Assumption B3. Furthermore, for all $%
n>M$ and $\vartheta \in \Theta $,
\begin{equation*}
|\mathbb{E}[\Lambda (\vartheta )|z_{2}>\sigma _{n}]-Q_{1}(\vartheta
)|=\lim_{m\rightarrow \infty }|\mathbb{E}[\Lambda (\vartheta )|z_{2}>\sigma
_{n}]-\mathbb{E}[\Lambda (\vartheta )|z_{2}>\sigma _{m}]|\leq \varepsilon .
\end{equation*}%
As $\varepsilon $ is arbitrary, we establish $\sup_{\mathcal{F}%
_{n}}\left\vert \mathbb{E}[\Lambda (\vartheta )|z_{2}>\sigma
_{n}]-Q_{1}(\vartheta )\right\vert =o(1)$, and thus (S2).

The remaining task is to verify the continuity condition (S3). Note that $%
Q_{1}(\vartheta )$ can be expressed as the sum of terms of the form
\begin{align*}
& \lim_{\sigma \rightarrow +\infty }P\left(
y_{1}=0,y_{2}=y_{3}=1,wz_{31}>-r-x_{31}^{\prime }b|z_{2}>\sigma \right) \\
=& \lim_{\sigma \rightarrow +\infty }\int \int_{-(
r+x_{31}^{\prime }b)/w }^{+\infty
}P(y_{1}=0,y_{2}=y_{3}=1|x_{31},z_{2}>\sigma
,z_{31}=z)f_{z_{31}}(z|x_{31},z_{2}>\sigma )dzdF_{x_{31}|z_{2}>\sigma }
\end{align*}%
with $F_{x_{31}|z_{2}>\sigma }$ representing the CDF of $x_{31}$ conditional
on $\{z_{2}>\sigma \}$. Then we see $Q_{1}(\vartheta )$ is continuous in $%
\vartheta $ if $\lim_{\sigma \rightarrow +\infty
}[P(y_{1}=0,y_{2}=y_{3}=1|x_{31},z_{2}>\sigma
,z_{31}=z)f_{z_{31}}(z|x_{31},z_{2}>\sigma )]$ is continuous in $z$, which
is secured by Assumptions A1 and A3.
\end{proof}

\section{Proof of Theorem \ref{T:limiting_dist} (Asymptotic Distribution)}\label{appendixB}

This appendix derives the convergence rate and asymptotic distribution of
the proposed MS estimator $\hat{\theta}_{n}$. Throughout this section, we
use $\mathbb{P}_{n}$ to denote empirical measure, $\mathbb{P}$ to denote
expectation, and $\mathbb{G}_{n}$ to denote empirical process in the sense
of Section 2.1 of \cite{vdVaartWellner1996}, following the notation
conventions in the empirical processes literature. Our proof strategy is
along the lines of \cite{KimPollard1990} and \cite{SeoOtsu2018}, and so we
deliberately keep notation similar to those used in these papers. To ease
exposition, we use the abbreviation
\begin{equation*}
q_{n1,\vartheta }(\chi ):=h_{n}^{-1}y_{2}y_{31}\cdot \mathds{1}%
\{z_{2}>\sigma _{n}\}\cdot \left[ u(\vartheta )-u(\theta )\right]
\end{equation*}%
for all $\vartheta \in \Theta $, where $\chi :=(y_{0},y^{T},x^{T},z^{T})$.
Note that the estimator $\hat{\theta}_{n}$ obtained from objective function (%
\ref{eq:appendix_obj1}) can be equivalently obtained from maximizing $%
\mathbb{P}_{n}q_{n1,\vartheta }$ w.r.t. $\vartheta $.

Lemmas \ref{lemma_M23}--\ref{lemma_M1} below verify technical conditions
(similar to Assumption M of \cite{SeoOtsu2018}) as required by Lemma \ref%
{lemma_rate} to derive the rate of convergence of $\hat{\theta}_{n}$.

\begin{lemma}[$L_2(P)$-Norm and Envelope Condition]
\label{lemma_M23} Suppose Assumptions A--C hold. Then

\begin{enumerate}
\item[(i)] There exist positive constants $C_{1}$ and $C_{2}$ such that
\begin{equation*}
h_{n}\mathbb{P}\left[ \left( q_{n1,\vartheta _{1}}(\chi )-q_{n1,\vartheta
_{2}}(\chi )\right) ^{2}\right] \geq C_{1}\Vert \vartheta _{1}-\vartheta
_{2}\Vert ^{2},
\end{equation*}%
for all $n$ large enough and $\vartheta _{1},\vartheta _{2}\in \left\{
\Theta :\Vert \vartheta -\theta \Vert \leq C_{2}\right\} $.

\item[(ii)] There exists a positive constant $C_{3}$ such that
\begin{equation*}
\mathbb{P}\left[ \sup_{\vartheta _{1}\in \{\Theta :\Vert \vartheta
_{1}-\vartheta _{2}\Vert <\varepsilon \}}h_{n}|q_{n1,\vartheta _{1}}(\chi
)-q_{n1,\vartheta _{2}}(\chi )|^{2}\right] \leq C_{3}\varepsilon ,
\end{equation*}%
for all $n$ large enough, $\varepsilon >0$ small enough, and $\vartheta _{2}$
in a neighborhood of $\theta $.
\end{enumerate}
\end{lemma}

\begin{proof}[Proof of Lemma \protect\ref{lemma_M23}]
Let $e_{n}(\chi ):=h_{n}^{-1}y_{2}y_{31}\mathds{1}\{z_{2}>\sigma _{n}\}$. By
Assumptions A and C2, we have
\begin{align}
h_{n}\mathbb{P}\left[ e_{n}(\chi )^{2}|\bar{\chi}\right] =& \int_{\sigma
_{n}}^{\infty }h_{n}^{-1}f_{z_{2}}(z|y_{2}=1,y_{1}\neq y_{3},\bar{\chi}%
)dzP\left( y_{2}=1,y_{1}\neq y_{3}|\bar{\chi}\right)   \notag \\
=& \frac{P\left( z_{2}>\sigma _{n}|y_{2}=1,y_{1}\neq y_{3},\bar{\chi}\right)
}{P\left( z_{2}>\sigma _{n}|y_{2}=1\right) }\cdot P\left( y_{2}=1,y_{1}\neq
y_{3}|\bar{\chi}\right) \geq c_{1}.  \label{eq:lb1-1}
\end{align}%
hold for some constant $c_{1}>0$ almost surely. Then for any $\vartheta
_{1},\vartheta _{2}\in \Theta $, we can write
\begin{align*}
h_{n}\mathbb{P}\left[ \left( q_{n1,\vartheta _{1}}(\chi )-q_{n1,\vartheta
_{2}}(\chi )\right) ^{2}\right] & =h_{n}\mathbb{P}\left\{ \mathbb{P}\left[
e_{n}(\chi )^{2}|\bar{\chi}\right] \cdot \left( u(\vartheta
_{1})-u(\vartheta _{2})\right) ^{2}\right\}  \\
& \geq c_{1}\mathbb{P}\left[ \left( u(\vartheta _{1})-u(\vartheta
_{2})\right) ^{2}\right] \geq c_{2}\Vert \vartheta _{1}-\vartheta _{2}\Vert
^{2}
\end{align*}%
for some constant $c_{2}>0$, where the first inequality uses (\ref{eq:lb1-1}%
) and the last inequality follows from the same argument (p. 214) to Example
6.4 of \cite{KimPollard1990}. This proves Lemma \ref{lemma_M23}(i). Lemma %
\ref{lemma_M23}(ii) can be proved using similar argument as
\begin{align*}
& h_{n}\mathbb{P}\left[ \sup_{\vartheta _{1}\in \left\{ \Theta :\Vert
\vartheta _{1}-\vartheta _{2}\Vert _{2}<\varepsilon \right\}
}|q_{n1,\vartheta _{1}}(\chi )-q_{n1,\vartheta _{2}}(\chi )|^{2}\right]  \\
=& \mathbb{P}\left\{ h_{n}\mathbb{P}\left[ e_{n}(\chi )^{2}|\bar{\chi}\right]
\cdot \sup_{\vartheta _{1}\in \left\{ \Theta :\Vert \vartheta _{1}-\vartheta
_{2}\Vert _{2}<\varepsilon \right\} }|u(\vartheta _{1})-u(\vartheta
_{2})|\right\}  \\
\leq & c_{3}\mathbb{P}\left[ \sup_{\vartheta _{1}\in \left\{ \Theta :\Vert
\vartheta _{1}-\vartheta _{2}\Vert _{2}<\varepsilon \right\} }|u(\vartheta
_{1})-u(\vartheta _{2})|\right] \leq c_{4}\varepsilon
\end{align*}%
for some constants $c_{3},c_{4}>0$.
\end{proof}

\begin{lemma}[First-order Bias]
\label{lemma_C5} Suppose Assumptions A--C hold. Then
\begin{equation*}
\left. \frac{\partial \mathbb{P}[\bar{q}_{n1,\vartheta }(\bar{\chi})]}{%
\partial \vartheta }\right\vert _{\vartheta =\theta }=O\left(
(nh_{n})^{-1/3}\right) .
\end{equation*}
\end{lemma}

\begin{proof}[Proof of Lemma \protect\ref{lemma_C5}]
Applying argument similar to Section 5 (pp. 205--206) and Example 6.4 (pp.
213--215) of \cite{KimPollard1990}, we have\footnote{%
Recall that $\Vert \vartheta \Vert =1$ for all $\vartheta \in \Theta $.}
\begin{equation}
\frac{\partial \mathbb{P}[\bar{q}_{n1,\vartheta }(\bar{\chi})]}{\partial
\vartheta }=\vartheta ^{\prime }\theta (I+\vartheta \vartheta ^{\prime
})\int \mathds{1}\{\nu ^{\prime }\theta =0\}\kappa _{n}(T_{\vartheta }\nu
)\nu f_{\bar{\chi}}(T_{\vartheta }\nu |z_{2}>\sigma _{n},y_{2}=1)d\mu
_{\theta },  \label{eq:C5-1}
\end{equation}%
for $\vartheta $ near $\theta $, where the transformation $T_{\vartheta
}:=(I-\vartheta \vartheta ^{\prime })(I-\theta \theta ^{\prime })+\vartheta
\theta ^{\prime }$ maps the region $\{\bar{\chi}:\bar{\chi}^{\prime }\theta
\geq 0\}$ to $\{\bar{\chi}:\bar{\chi}^{\prime }\vartheta \geq 0\}$ (taking $%
\{\bar{\chi}:\bar{\chi}^{\prime }\theta =0\}$ to $\{\bar{\chi}:\bar{\chi}%
^{\prime }\vartheta =0\}$), $\mu _{\theta }$ is the surface measure on $\{%
\bar{\chi}:\bar{\chi}^{\prime }\theta =0\}$, and $\kappa _{n}(\cdot ):=%
\mathbb{P}[y_{31}|z_{2}>\sigma _{n},y_{2}=1,\bar{\chi}=\cdot ]$.

Note that by definition, $T_{\theta }\bar{\chi}=\bar{\chi}$ along $\{\bar{%
\chi}:\bar{\chi}^{\prime }\theta =0\}$, and hence we can use (\ref{eq:C5-1})
to obtain
\begin{equation}
\left. \frac{\partial \mathbb{P}[\bar{q}_{n1,\vartheta }(\bar{\chi})]}{%
\partial \vartheta }\right\vert _{\vartheta =\theta }=(I+\theta \theta
^{\prime })\int \mathds{1}\{\nu ^{\prime }\theta =0\}\kappa _{n}(\nu )\nu f_{%
\bar{\chi}}(\nu |z_{2}>\sigma _{n},y_{2}=1)d\mu _{\theta }.  \label{eq:C5-2}
\end{equation}%
Since $\kappa ^{+}(\bar{\chi}):=\lim_{n\rightarrow \infty }\kappa _{n}(\bar{%
\chi})=0$ along $\{\bar{\chi}:\bar{\chi}^{\prime }\theta =0\}$, (\ref%
{eq:C5-2}) implies $\left. \frac{\partial \mathbb{P}[\bar{q}_{n1,\vartheta }(%
\bar{\chi})]}{\partial \vartheta }\right\vert _{\vartheta =\theta
}\rightarrow 0$. The remaining task is to derive its convergence rate.
Recall that $y_{31}=\mathds{1}\{C\}-\mathds{1}\{D\}$ for any fixed $%
(y_{0},y_{2})$, and so letting $\Omega :=\{z_{2}>\sigma _{n},y_{2}=1,\bar{\chi}\}$, we can write%\footnote{%Here we ignore the scale normality to ease the exposition.}
\begin{align}
\kappa _{n}(\bar{\chi})=& \mathbb{P}[\mathds{1}\{C\}-\mathds{1}\{D\}|\Omega
]=\mathbb{P}\left\{ \mathbb{P}\left[ \mathds{1}\{C\}-\mathds{1}\{D\}|\alpha
,y_{0}=d_0,x^{T},z^{T}\right] |\Omega \right\}   \notag \\
=& \mathbb{P}\left\{ (1-F_{\epsilon |\alpha }(\alpha +\gamma
d_{0}+x_{1}^{\prime }\beta +\varpi z_{1}))F_{\epsilon |\alpha }(\alpha
+x_{2}^{\prime }\beta +\varpi z_{2})F_{\epsilon |\alpha }(\alpha +\gamma
+x_{3}^{\prime }\beta +\varpi z_{3})\right.   \notag \\
& \left. -F_{\epsilon |\alpha }(\alpha +\gamma d_{0}+x_{1}^{\prime }\beta
+\varpi z_{1})F_{\epsilon |\alpha }(\alpha +\gamma +x_{2}^{\prime }\beta
+\varpi z_{2})(1-F_{\epsilon |\alpha }(\alpha +\gamma +x_{3}^{\prime }\beta
+\varpi z_{3}))|\Omega \right\}   \notag \\
=& \mathbb{P}\left\{ (1-F_{\epsilon |\alpha }(\alpha +\gamma +x_{3}^{\prime
}\beta +\varpi z_{3}))F_{\epsilon |\alpha }(\alpha +\gamma +x_{3}^{\prime
}\beta +\varpi z_{3})\right.   \notag \\
& \left. \times \left( F_{\epsilon |\alpha }(\alpha +x_{2}^{\prime }\beta
+\varpi z_{2})-F_{\epsilon |\alpha }(\alpha +\gamma +x_{2}^{\prime }\beta
+\varpi z_{2})\right) |\Omega \right\} ,  \label{eq:C5-3}
\end{align}%
where $F_{\epsilon |\alpha }(\cdot )$ denotes the conditional (on $\alpha $)
CDF of $\epsilon _{t}$ and the third equality follows as we calculate $%
\kappa _{n}(\bar{\chi})$ along $\{\bar{\chi}:\bar{\chi}^{\prime }\theta =0\}$%
. Assume w.o.l.g. $\gamma <0$. Then Assumption C5 implies that there exists
some constant $c_{1}>0$ such that
\begin{align}
& F_{\epsilon |\alpha }(\alpha +x_{2}^{\prime }\beta +\varpi
z_{2})-F_{\epsilon |\alpha }(\alpha +\gamma +x_{2}^{\prime }\beta +\varpi
z_{2})\notag \\
 =&P(\alpha +\gamma +x_{2}^{\prime }\beta +\varpi z_{2}\leq \epsilon
_{2}<\alpha +x_{2}^{\prime }\beta +\varpi z_{2}|\alpha ,y_{0},x^{T},z^{T})
\notag \\
 \leq &P(\epsilon _{2}\geq \alpha +\gamma +x_{2}^{\prime }\beta +\varpi
\sigma _{n}|\alpha ,y_{0},x^{T},z^{T})\leq c_{1}(nh_{n})^{-1/3}
\label{eq:C5-4}
\end{align}%
for $n$ large enough. Plug (\ref{eq:C5-4}) into (\ref{eq:C5-3}) to conclude
that $\kappa _{n}(\bar{\chi})=O_{p}\left( (nh_{n})^{-1/3}\right) $. Then
combining this result with (\ref{eq:C5-2}) completes the proof.
\end{proof}

\begin{lemma}[Quadratic Approximation]
\label{lemma_M1} Suppose Assumptions A--C hold. Then for all $n$ large
enough and $\vartheta $ in a neighborhood of $\theta $,
\begin{equation*}
\mathbb{P}\left[ q_{n1,\vartheta }(\chi )-q_{n1,\theta }(\chi )\right] =%
\frac{1}{2}(\vartheta -\theta )^{\prime }V(\vartheta -\theta )+o\left( \Vert
\vartheta -\theta \Vert ^{2}\right) +O\left( (nh_{n})^{-1/3}\Vert \vartheta
-\theta \Vert \right)
\end{equation*}%
where $V:=-\int \mathds{1}\left\{ \nu ^{\prime }\theta =0\right\} \dot{\kappa%
}^{+}(\nu )^{\prime }\theta \nu \nu ^{\prime }f_{\bar{\chi}}^{+}(\nu
|y_{2}=1)d\mu _{\theta }P(y_{2}=1)$ and $\mu _{\theta }$ is the surface
measure on the boundary of $\{\bar{\chi}:\bar{\chi}^{\prime }\theta \geq 0\}$%
.
\end{lemma}

\begin{proof}[Proof of Lemma \protect\ref{lemma_M1}]
Note that
\begin{align}
\lim_{\sigma _{n}\rightarrow +\infty }\mathbb{P}\left[ q_{n1,\vartheta
}(\chi )\right] & =\lim_{\sigma _{n}\rightarrow +\infty }\mathbb{P}\left[
y_{31}\left( u(\vartheta )-u(\theta )\right) |z_{2}>\sigma _{n},y_{2}=1%
\right] \cdot P(y_{2}=1)  \notag \\
& =\mathbb{P}\left[ \bar{q}_{1\vartheta }^{+}(\bar{\chi})\right] \cdot
P(y_{2}=1).  \label{eq:lb2-limit}
\end{align}%
We first derive a second order expansion of $\lim_{\sigma _{n}\rightarrow
+\infty }\mathbb{P}\left[ q_{n1,\vartheta }(\chi )\right] $. Recall that
equation (\ref{eq:iden_ineq1}), together with Assumption A, implies $\text{%
sgn}\left\{ \mathbb{P}[y_{31}|z_{2}>\sigma _{n},y_{2}=1,\bar{\chi}]\right\} =%
\text{sgn}\left\{ \gamma y_{20}+x_{31}^{\prime }\beta +\varpi z_{31}\right\} $
holds as $\sigma _{n}\rightarrow +\infty $. Let $\mathcal{Z}_{\theta }:=\{%
\bar{\chi}:u(\vartheta )\neq u(\theta )\}$. Then using the same argument as
in the proof of Lemma 7 in \cite{honore-k}, we can write
\begin{equation*}
-\mathbb{P}\left[ \bar{q}_{1\vartheta }^{+}(\bar{\chi})\right] =\int_{%
\mathcal{Z}_{\theta }}\lim_{\sigma _{n}\rightarrow +\infty }|\mathbb{P}\left[
y_{31}|z_{2}>\sigma _{n},y_{2}=1,\bar{\chi}\right] |dF_{\bar{\chi}%
|y_{2}=1}^{+}>0
\end{equation*}%
for all $\vartheta \neq \theta $. Therefore, we have\footnote{%
It is easy to verify this using similar argument to (\ref{eq:iden_equiv}).}
\begin{equation}
\left. \frac{\partial \mathbb{P}\left[ \bar{q}_{1\vartheta }^{+}(\bar{\chi})%
\right] }{\partial \vartheta }\right\vert _{\vartheta =\theta }=0.
\label{eq:lb2-foc}
\end{equation}%
Furthermore, by applying argument similar to Example 6.4 of \cite%
{KimPollard1990}, we obtain
\begin{equation}
-\frac{\partial ^{2}\mathbb{P}\left[ \bar{q}_{1\vartheta }^{+}(\bar{\chi})%
\right] }{\partial \vartheta \partial \vartheta ^{\prime }}=\int \mathds{1}%
\left\{ \nu ^{\prime }\theta =0\right\} \dot{\kappa}^{+}(\nu )^{\prime
}\theta \nu \nu ^{\prime }f_{\bar{\chi}}^{+}(\nu |y_{2}=1)d\mu _{\theta }.
\label{eq:lb2-soc}
\end{equation}%
Combining (\ref{eq:lb2-limit}), (\ref{eq:lb2-foc}), and (\ref{eq:lb2-soc}),
we obtain
\begin{align}
\lim_{\sigma _{n}\rightarrow +\infty }\mathbb{P}\left[ q_{n1,\vartheta
}(\chi )-q_{n1,\theta }(\chi )\right] & =\mathbb{P}\left[ \bar{q}%
_{1\vartheta }^{+}(\bar{\chi})-\bar{q}_{1\theta }^{+}(\bar{\chi})\right]
\cdot P(y_{2}=1)  \notag \\
& =\frac{1}{2}(\vartheta -\theta )^{\prime }V(\vartheta -\theta )+o\left(
\Vert \vartheta -\theta \Vert ^{2}\right) ,  \label{eq:lb2-taylor1}
\end{align}%
where $V=-\int \mathds{1}\left\{ \nu ^{\prime }\theta =0\right\} \dot{\kappa}%
^{+}(\nu )^{\prime }\theta \nu \nu ^{\prime }f_{\bar{\chi}}^{+}(\nu
|y_{2}=1)d\mu _{\theta }P(y_{2}=1)$.

Next, applying similar argument to $\mathbb{P}\left[ q_{n1,\vartheta }(\chi )%
\right] $ yields
\begin{equation}
\mathbb{P}\left[ q_{n1,\vartheta }(\chi )-q_{n1,\theta }(\chi )\right]
=\left. (\vartheta -\theta )^{\prime }\frac{\partial \mathbb{P}\left[ \bar{q}%
_{n,\vartheta }(\bar{\chi})\right] }{\partial \vartheta }\right\vert
_{\vartheta =\theta }+\frac{1}{2}(\vartheta -\theta )^{\prime
}V_{n}(\vartheta -\theta )+o\left( \Vert \vartheta -\theta \Vert ^{2}\right)
,  \label{eq:lb2-taylor2}
\end{equation}%
where $V_{n}=-\int \mathds{1}\left\{ \nu ^{\prime }\theta =0\right\} \dot{%
\kappa}_{n}(\nu )^{\prime }\theta \nu \nu ^{\prime }f_{\bar{\chi}}(\nu
|z_{2}>\sigma _{n},y_{2}=1)d\mu _{\theta }P(y_{2}=1)$. Then under
Assumptions C3--C5, the desired result follows by combining (\ref%
{eq:lb2-taylor1}), (\ref{eq:lb2-taylor2}), and Lemma \ref{lemma_C5} as
\begin{align*}
& \mathbb{P}\left[ q_{n1,\vartheta }(\chi )-q_{n1,\theta }(\chi )\right] \\
=& \lim_{\sigma _{n}\rightarrow +\infty }\mathbb{P}\left[ q_{n1,\vartheta
}(\chi )-q_{n1,\theta }(\chi )\right] +\left( \mathbb{P}\left[
q_{n1,\vartheta }(\chi )-q_{n1,\theta }(\chi )\right] -\lim_{\sigma
_{n}\rightarrow +\infty }\mathbb{P}\left[ q_{n1,\vartheta }(\chi
)-q_{n1,\theta }(\chi )\right] \right) \\
=& \frac{1}{2}(\vartheta -\theta )^{\prime }V(\vartheta -\theta )+\frac{1}{2}%
(\vartheta -\theta )^{\prime }(V_{n}-V)(\vartheta -\theta )+o\left( \Vert
\vartheta -\theta \Vert ^{2}\right) +\left. (\vartheta -\theta )^{\prime }%
\frac{\partial \mathbb{P}\left[ \bar{q}_{n,\vartheta }(\bar{\chi})\right] }{%
\partial \vartheta }\right\vert _{\vartheta =\theta } \\
=& \frac{1}{2}(\vartheta -\theta )^{\prime }V(\vartheta -\theta )+o\left(
\Vert \vartheta -\theta \Vert ^{2}\right) +O\left( (nh_{n})^{-1/3}\Vert
\vartheta -\theta \Vert \right) .
\end{align*}
\end{proof}

\begin{lemma}[Convergence Rate of $\hat{\protect\theta}_{n}$]
\label{lemma_rate} Under Assumptions A--C, $\hat{\theta}_{n}-\theta=O_{p}%
\left((nh_{n})^{-1/3}\right)$.
\end{lemma}

\begin{proof}[Proof of Lemma \protect\ref{lemma_rate}]
Recall that $h_{n}q_{n1,\vartheta }(\chi )$ is uniformly bounded by
definition and $\mathbb{P}[q_{n1,\vartheta }(\chi )]$ is twice continuously
differentiable at $\theta $ for all $n$ large enough under Assumption C5.
Furthermore, we have shown in Theorem \ref{T:identify} that $%
\lim_{n\rightarrow \infty }\mathbb{P}[q_{n1,\vartheta }(\chi )]$ is uniquely
maximized at $\theta $ under Assumption A. Putting all these results and
Lemma \ref{lemma_M23} together enables us to apply Lemmas M and 1 of \cite%
{SeoOtsu2018} to obtain that there exist a sequence of random variables $%
R_{n}=O_{p}(1)$ and some positive constant $C$ such that
\begin{equation}
\left\vert \mathbb{P}_{n}\left( q_{n1,\vartheta }(\chi _{i})-q_{n1,\theta
}(\chi _{i})\right) -\mathbb{P}\left[ q_{n1,\vartheta }(\chi )-q_{n1,\theta
}(\chi )\right] \right\vert \leq \varepsilon \Vert \vartheta -\theta \Vert
^{2}+(nh_{n})^{-2/3}R_{n}^{2}  \label{eq:lb3-1}
\end{equation}%
holds for all $\vartheta \in \{\Theta :(nh_{n})^{-1/3}\leq \Vert \vartheta
-\theta \Vert \leq C\}$ and $\varepsilon >0$ as $n\rightarrow \infty $.
Then, assuming $\Vert \hat{\theta}_{n}-\theta \Vert \geq (nh_{n})^{-1/3}$,
we can take a positive constant $c$ such that for all $\varepsilon >0$
\begin{align}
& o_{p}\left( (nh_{n})^{-2/3}\right)  \notag \\
\leq & \mathbb{P}_{n}\left( q_{n1,\hat{\theta}_{n}}(\chi _{i})-q_{n1,\theta
}(\chi _{i})\right)  \notag \\
\leq & \mathbb{P}\left[ q_{n1,\hat{\theta}_{n}}(\chi )-q_{n1,\theta }(\chi )%
\right] +\varepsilon \Vert \hat{\theta}_{n}-\theta \Vert
^{2}+(nh_{n})^{-2/3}R_{n}^{2}  \notag \\
\leq & (-c+\varepsilon )\Vert \hat{\theta}_{n}-\theta \Vert ^{2}+o\left(
\Vert \hat{\theta}_{n}-\theta \Vert ^{2}\right) +O_{p}\left(
(nh_{n})^{-1/3}\Vert \hat{\theta}_{n}-\theta \Vert \right) +O_{p}\left(
(nh_{n})^{-2/3}\right) ,  \notag  \label{eq:lb3-2}
\end{align}%
where the first inequality is due to Assumption C1, the second inequality
uses (\ref{eq:lb3-1}), and the last inequality follows from Lemma \ref%
{lemma_M1}. As $\varepsilon $ can be arbitrarily small, taking some $%
\varepsilon <c$ justifies the convergence rate of $\hat{\theta}_{n}$ claimed
in Lemma \ref{lemma_rate}.
\end{proof}

The rate established in Lemma \ref{lemma_rate} enables us to consider the
following centered and normalized empirical process
\begin{equation}
Z_{n}(s):=n^{1/6}h_{n}^{2/3}\mathbb{G}_{n}\left(q_{n1,%
\theta+s(nh_{n})^{-1/3}}-q_{n1,\theta}\right)  \label{eq:ep}
\end{equation}
for $\Vert s\Vert\le K$ with some $K>0$. Lemma \ref{lemma_finite} below
yields a finite dimensional convergence result which characterizes the weak
convergence of $Z_{n}(s)$. Lemma \ref{lemma_sae} establishes the stochastic
asymptotic equicontinuity of $Z_{n}(s)$. With these results, the limiting
distribution of $\hat{\theta}_{n}$ then follows by the continuous mapping
theorem of an argmax element (see Theorem 2.7 of \cite{KimPollard1990}).

\begin{lemma}[Finite Dimensional Convergence]
\label{lemma_finite} Let $g_{n}$ be any finite dimensional projection of the
process $\{g_{n}(s)-\mathbb{P}[g_{n}(s)]\}$ for $\Vert s\Vert\le K$ with
some $K>0$, where
\begin{equation}
g_{n}(s):=n^{1/6}h_{n}^{2/3}\left(q_{n1,\theta+s(nh_{n})^{-1/3}}-q_{n1,%
\theta}\right).  \label{eq:lb4-0}
\end{equation}
Suppose Assumptions A--C hold. Then $\Sigma:=\lim_{n\rightarrow\infty}\text{%
Var}(\mathbb{G}_{n}g_{n})$ exists and $\mathbb{G}_{n}g_{n}\overset{d}{%
\rightarrow}N(0,\Sigma)$.
\end{lemma}

\begin{proof}[Proof of Lemma \protect\ref{lemma_finite}]
Let $A_{n,s}(\bar{\chi}):=u(\theta +s(nh_{n})^{-1/3})-u(\theta )$. Note that
we can write
\begin{align}
& P(|g_{n}(s)|\geq c_{1})  \notag \\
=& P(|y_{2}y_{31}\mathds{1}\{z_{2}>\sigma _{n}\}|\geq
c_{1}n^{-1/6}h_{n}^{1/3}||A_{n,s}(\bar{\chi})|=1)P(|A_{n,s}(\bar{\chi})|=1)
\notag \\
=& P(y_{2}=1,y_{1}\neq y_{3},z_{2}>\sigma _{n}||A_{n,s}(\bar{\chi}%
)|=1)P(|A_{n,s}(\bar{\chi})|=1)  \notag \\
\leq & c_{2}P(z_{2}>\sigma _{n}||A_{n,s}(\bar{\chi})|=1,y_{2}=1,y_{1}\neq
y_{3})P(|A_{n,s}(\bar{\chi})|=1)\leq c_{3}(nh_{n}^{-2})^{-1/3},
\label{eq:lb4-1}
\end{align}%
for some $c_{1},c_{2},c_{3}>0$ and $n$ large enough, where the last
inequality follows by Assumption C2 and the fact that $P(|A_{n,s}(\bar{\chi}%
)|=1)\asymp (nh_{n})^{-1/3}$ (see the argument to Example 6.4 of \cite%
{KimPollard1990}). Then, under Assumptions A--C and (\ref{eq:lb4-1}), the
finite dimensional convergence claimed in Lemma \ref{lemma_finite} follows
by first applying Lemma 2 and then Lemma C of \cite{SeoOtsu2018}.
\end{proof}

\begin{lemma}[Stochastic Asymptotic Equicontinuity]
\label{lemma_sae} For any $\eta>0$, there exists $\delta>0$ and a positive
integer $N_{\delta}$ large enough such that
\begin{equation}
\mathbb{P}\left[\sup_{(s_{1},s_{2}):\Vert s_{1}-s_{2}\Vert<\delta}\left|%
\mathbb{G}_{n}\left(g_{n}(s_{1})-g_{n}(s_{2})\right)\right|\right]\leq\eta
\label{eq:lb5-0}
\end{equation}
holds for all $n\geq N_{\delta}$, where $g_{n}(s)$ is defined in (\ref%
{eq:lb4-0}).
\end{lemma}

\begin{proof}[Proof of Lemma \protect\ref{lemma_sae}]
Denote $G_{n}=\sup_{\Vert s\Vert \leq K}|g_{n}(s)|$ as the envelope of the
class of functions $\mathcal{F}_{n}:=\{g_{n}(s):\Vert s\Vert \leq K\}$. By
Assumption C6, we have $G_{n}/n^{\varrho }\leq n^{1/6-\varrho
}h_{n}^{-1/3}=O(1)$ for $\varrho \leq \frac{1}{2}-\frac{\varepsilon}{3}$. Furthermore, note that
\begin{equation*}
\mathbb{P}G_{n}^{2}=(nh_{n})^{1/3}\cdot \mathbb{P}\left[ \sup_{\Vert s\Vert
\leq K}h_{n}\left\vert q_{n1,\theta +s(nh_{n})^{-1/3}}(\chi )-q_{n1,\theta
}(\chi )\right\vert ^{2}\right] \leq (nh_{n})^{1/3}\cdot
C_{3}K(nh_{n})^{-1/3}=C_{3}K
\end{equation*}%
for some $C_{3}>0$ and all $n$ large enough, where the inequality follows
from Lemma \ref{lemma_M23}(ii). Then with all these results and Lemma \ref%
{lemma_M23}, applying Lemma M' of \cite{SeoOtsu2018} proves (\ref{eq:lb5-0}).
\end{proof}

\begin{proof}[Proof of Theorem \ref{T:limiting_dist}]
Part (i) of Theorem 4.2 has been proved in Lemma \ref%
{lemma_rate}. This result implies that when deriving the asymptotic
distribution of $\hat{\theta}_{n}$, we can restrict our attention to the
empirical process defined in (\ref{eq:ep}). Note that the finite dimensional
convergence and stochastic asymptotic equicontinuity results obtained in
Lemmas \ref{lemma_finite} and \ref{lemma_sae}, respectively, guarantees that
$Z_{n}(s)\overset{d}{\rightarrow }Z(s)$ with expected values $s^{\prime }Vs/2
$ and covariance kernel $H(s_{1},s_{2})$ by Theorem 2.7 of \cite%
{KimPollard1990} (see also Theorem 1 of \cite{SeoOtsu2018}). To derive $%
H(s_{1},s_{2})$, we use Theorem 4.7 of \cite{KimPollard1990} to write
\begin{align}
H(s_{1},s_{2})& =\lim_{n\rightarrow \infty }\mathbb{P}\left[
g_{n}(s_{1})g_{n}(s_{2})\right]   \notag \\
& =\lim_{n\rightarrow \infty }(nh_{n})^{1/3}\mathbb{P}\left[ h_{n}\left(
q_{n1,\theta +s_{1}(nh_{n})^{-1/3}}-q_{n1,\theta }\right) \left(
q_{n1,\theta +s_{2}(nh_{n})^{-1/3}}-q_{n1,\theta }\right) \right]   \notag \\
& =\frac{1}{2}\left( L(s_{1})+L(s_{2})-L(s_{1}-s_{2})\right) ,
\label{eq:cov_ker0}
\end{align}%
where
\begin{equation*}
L(s):=\lim_{n\rightarrow \infty }(nh_{n})^{1/3}\mathbb{P}\left[ h_{n}\left(
q_{n1,\theta +s(nh_{n})^{-1/3}}-q_{n1,\theta }\right) ^{2}\right]
\end{equation*}%
and
\begin{equation*}
L(s_{1}-s_{2}):=\lim_{n\rightarrow \infty }(nh_{n})^{1/3}\mathbb{P}\left[
h_{n}\left( q_{n1,\theta +s_{1}(nh_{n})^{-1/3}}-q_{n1,\theta
+s_{2}(nh_{n})^{-1/3}}\right) ^{2}\right] .
\end{equation*}%
Note that
\begin{align*}
& L(s_{1}-s_{2}) \\
=& \lim_{n\rightarrow \infty }(nh_{n})^{1/3}\mathbb{P}\left[
h_{n}^{-1}y_{2}|y_{31}|\cdot \mathds{1}\{z_{2}>\sigma _{n}\}\left( u(\theta
+s_{1}(nh_{n})^{-1/3})-u(\theta +s_{2}(nh_{n})^{-1/3})\right) ^{2}\right]  \\
=& \lim_{n\rightarrow \infty }(nh_{n})^{1/3}\mathbb{P}\left[
|y_{31}|\left\vert u(\theta +s_{1}(nh_{n})^{-1/3})-u(\theta
+s_{2}(nh_{n})^{-1/3})\right\vert ^{2}|z_{2}>\sigma _{n},y_{2}=1\right]
P(y_{2}=1).
\end{align*}%
Using the same argument to Example 6.4 (p. 215) of \cite{KimPollard1990}, we
can decompose vector $\bar{\chi}$ into $\varpi ^{\prime }\theta +\bar{\chi}%
^{\perp }$ with $\bar{\chi}^{\perp }$ orthogonal to $\theta $ and write
\begin{equation}
L(s_{1}-s_{2})=\int |\bar{\chi}^{\perp \prime }(s_{1}-s_{2})|f^{+}(0,\bar{%
\chi}^{\perp }|y_{2}=1)d\bar{\chi}^{\perp }\cdot P(y_{2}=1),
\label{eq:cov_ker1}
\end{equation}%
where $f^{+}(\cdot ,\cdot |y_{2}=1)$ denotes the limit of the joint PDF $%
f(\cdot ,\cdot |z_{2}>\sigma _{n},y_{2}=1)$ of $(\varpi ,\bar{\chi}^{\perp })
$ conditional on $\{z_{2}>\sigma _{n},y_{2}=1\}$ as $n\rightarrow \infty $.
Taking $(s_{1},s_{2})=(s,0)$ in (\ref{eq:cov_ker1}) gives
\begin{equation}
L(s)=\int |\bar{\chi}^{\perp \prime }s|f^{+}(0,\bar{\chi}^{\perp }|y_{2}=1)d%
\bar{\chi}^{\perp }\cdot P(y_{2}=1).  \label{eq:cov_ker2}
\end{equation}%
Then plugging (\ref{eq:cov_ker1}) and (\ref{eq:cov_ker2}) into (\ref%
{eq:cov_ker0}) yields
\begin{equation}
H(s_{1},s_{2})=\frac{1}{2}\int \left( |\bar{\chi}^{\perp \prime }s_{1}|+|%
\bar{\chi}^{\perp \prime }s_{2}|-|\bar{\chi}^{\perp \prime
}(s_{1}-s_{2})|\right) f^{+}(0,\bar{\chi}^{\perp }|y_{2}=1)d\bar{\chi}%
^{\perp }\cdot P(y_{2}=1).  \label{eq:cov_kernel}
\end{equation}%
This completes the proof.
\end{proof}

\section{Additional Simulation Results}\label{appendixC}

The sensitivity check results for Designs 1 and 2 are presented in Tables \ref{T:D1_robust} and \ref{T:D2_robust}, respectively. In addition, we conduct supplementary Monte Carlo experiments (Designs 3--6) to examine the impact of auto-correlations of covariates on the performance of our estimator and compare the performance of our estimator with that of HK and OY. Note that, for the latter, we remove the time trend term and set $T=4$ to make both HK and OY estimators applicable. As discussed in Appendix \ref{appendix0}, neither HK nor OY allow for time trends or dummies, and OY requires at least $T=4$.

\begin{comment}
In this section, we first examine the impact of auto-correlations
of regressors on the performance of our estimators. We then compare
the performance of our estimator with that of HK and OY. However,
to facilitate this comparison, we remove the time trend term, as HK
and OY do not allow this term. In addition, the sensitivity results
for Designs 1 and 2 are collected in this Section.
\end{comment}

We introduce Design 3, which is akin to Design 1 but with the distinction that $x_t$ and $z_t$ are auto-correlated. The data generating process (DGP) is formulated as follows:
\begin{align*}
y_{i0} & = \mathds{1}\left\{ \alpha_{i} + \delta \times (0 - 2) + \beta_{1}x_{i0,1} + z_{i0} \geq \epsilon_{i0} \right\}, \\
y_{it} & = \mathds{1}\left\{ \alpha_{i} + \delta \times (t - 2) + \gamma y_{it-1} + \beta_{1}x_{it,1} + z_{it} \geq \epsilon_{it} \right\}, \text{ for } t \in \{ 1, 2, 3 \},
\end{align*}
where \( \gamma = \beta_{1} = 1 \) and $\delta=1/2$. With all other aspects remaining the same as Design 1, we consider two sets of regressors:
\begin{enumerate}
\item Autoregressive (AR) regressors:
\[
x_{i0,1} = u_{i0,1},   x_{it,1} = \frac{1}{2}x_{it-1,1} + \frac{\sqrt{3}}{2}u_{it,1},  z_{i0} = u_{i0,2},   z_{it} = \frac{1}{2}z_{it-1} + \frac{\sqrt{3}}{2}u_{it,2}, \text{ for } t = 1, 2, 3.
\]
\item Moving Average (MA) regressors:
\[
x_{i0,1} = u_{i0,1},  x_{it,1} = \frac{\sqrt{3}}{2}u_{it,1} + \frac{1}{2}u_{it-1,1},  z_{i0} = u_{i0,2},   z_{it} = \frac{\sqrt{3}}{2}u_{it,2} + \frac{1}{2}u_{it-1,2}, \text{ for } t = 1, 2, 3.
\]
\end{enumerate}
We let \( u_{it,1} \overset{d}{\sim} N(0,1) \). For \( z_{it} \), we explore different tail behaviors. In the first scenario, termed ``Norm'', \( u_{it,2} \) is distributed as \( u_{it,2} \overset{d}{\sim} N(0,1) \), and in the second scenario, termed ``Lap'', as \( u_{it,2} \overset{d}{\sim} \text{Laplace}(0,\sqrt{2}/2) \).  Consequently, this leads to four distinct combinations arising from two scenarios and two sets of regressors.

Design 4 closely parallels Design 1, with two key differences: the removal of the time trend term and setting \( T = 4 \). These changes make the HK and OY estimators applicable. The DGP for this design is then formulated as follows:
\begin{align*}
y_{i0} & = \mathds{1}\left\{ \alpha_{i} + \beta_{1}x_{i0,1} + z_{i0} \geq \epsilon_{i0} \right\}, \\
y_{it} & = \mathds{1}\left\{ \alpha_{i} + \gamma y_{it-1} + \beta_{1}x_{it,1} + z_{it} \geq \epsilon_{it} \right\}, \text{ for } t \in \{ 1, 2, 3, 4 \},
\end{align*}
where \( \gamma = \beta_{1} = 1 \). The process of generating regressors and the error term follows the same approach as in Designs 1 and 3, i.e., we consider independent, AR, and MA covariates. Consequently, we examine six different combinations arising from two scenarios and three distinct sets of regressors. In Design 4, we compare our estimator against several others: the parametric estimator in HK (assuming $\epsilon_{it}$ to be logistic), denoted as HK1; the semiparametric estimator in HK (distribution-free), denoted as HK2; and the estimator in OY, denoted as OY.

Design 5 mirrors Design 2, but with covariates being auto-correlated as in Design 3. Specifically, we set
\begin{align*}
y_{i0} & = \mathds{1}\left\{ \alpha_{i} + \delta \times (0 - 2) + \beta_{1}x_{i0,1} + \beta_{2}x_{i0,2} + z_{i0} \geq \epsilon_{i0} \right\}, \\
y_{it} & = \mathds{1}\left\{ \alpha_{i} + \delta \times (t - 2) + \gamma y_{it-1} + \beta_{1}x_{it,1} + \beta_{2}x_{it,2} + z_{it} \geq \epsilon_{it} \right\}, \text{ for } t \in \{ 1, 2, 3 \},
\end{align*}
where \( \gamma = \beta_{1} = \beta_{2} = 1 \) and \( \delta = 1/2 \). We consider the same AR and MA DGP for \( x_{it,1}, x_{it,2} \) and \( z_{it} \). Again, we examine two types of \( z \) distributions: normal and Laplace. This results in a total of four cases.

The final design, Design 6, closely resembles Design 5, with the exceptions of omitting the time trend term and setting \( T = 4 \) to allow for the application of HK and OY. The model is then defined as:
\begin{align*}
y_{i0} & = \mathds{1}\left\{ \alpha_{i} + \beta_{1}x_{i0,1} + \beta_{2}x_{i0,2} + z_{i0} \geq \epsilon_{i0} \right\}, \\
y_{it} & = \mathds{1}\left\{ \alpha_{i} + \gamma y_{it-1} + \beta_{1}x_{it,1} + \beta_{2}x_{it,2} + z_{it} \geq \epsilon_{it} \right\}, \text{ for } t \in \{ 1, 2, 3, 4 \},
\end{align*}
where \( \gamma = \beta_{1} = \beta_{2} = 1 \). All scenarios considered in Design 5 are included in Design 6, along with the one in which \( x_{it,1}, x_{it,2} \), and \( z_{it} \) are serially independent. Therefore, Design 6 encompasses a total of six cases. Similar to Design 4, we compare the performance of our estimator, HK1, HK2, and OY.

Following the recommendation in HK, we adopt the bandwidth \( h_{n} = c \cdot n^{-1/6} \) for HK's estimators in Design 4. Experiments are conducted with \( c = 1, 2, 3, 4 \), and we report the simulation results corresponding to \( c = 3 \). This choice is based on the smallest bias and relatively smaller root mean square errors of the HK estimators of \( \gamma \) at this value. For Design 6, we set \( h_{n} = 3 \cdot n^{-1/7} \) for similar reasons. For the OY estimator, we set the bandwidth as \( h_{n} = n^{-1/4} \cdot (\log n)^{-1} \), following their recommendation. Our MS estimation uses
\[
\sigma_{n} = \widehat{\text{std}(z_{i2})} \cdot \sqrt{\log n^{*}/2.95},
\]
where \( \widehat{\text{std}(z_{i2})} \) is the sample standard deviation of \( z_{2} \), and \( n^{*} \) represents the number of ``switchers'' (i.e., observations with \( y_{3} \neq y_{1} \)). Sample sizes of \( n = 5000, 10000, \) and \( 20000 \) are considered. All results are based on 1000 replications for each sample size. We report the mean bias (MBIAS) and the root mean square errors (RMSE) of the estimates. The results are detailed in Tables \ref{T:D1_robust} to \ref{T:last}, which are titled to reflect their content and listed below for easy reference.

We briefly summarize our findings from these results here:
\begin{enumerate}
\item The simulation results of our estimator are not very sensitive to the choice of tuning parameters across all cases.
\item Despite serial correlations in covariates, our approach performs well, with some degradation observed in estimator performance, indicated by larger bias and RMSE. This deterioration arises from two factors: the reduced variation in covariates (especially for $z_{13}$ conditional on $\vert z_2\vert>\sigma_n$) and the challenge in selecting $\sigma_{n}$ due to differing tail behaviors on the left and right sides. As anticipated, a weaker serial correlation corresponds to better performance. For instance, $(x_{it},z_{it})$ in the AR setting has a serial correlation of 0.5, while $(x_{it},z_{it})$ in the MA setting exhibits a slightly lower serial correlation, approximately 0.4, where our estimator performs comparatively better.
\item  HK1 outperforms all semiparametric estimators, as expected, given its efficiency in low-dimensional cases with scaled logistic errors. Theoretically, OY converges faster than HK2 and our estimator (at least for $\beta$), which aligns with the simulation results. Our estimator outperforms HK2 across all designs with independent and MA regressors. However, in some designs with AR regressors (stronger serial correlations) and normally distributed $z$ (thin tails), it exhibits slightly larger RMSE. We anticipate that the advantage of our estimator over HK2 will become more obvious in higher dimensional settings. This is because our estimator, like OY, avoids the curse of dimensionality suffered by HK2.
\end{enumerate}
Importantly, it should be noted that HK1, HK2, and OY estimators are not suitable in the presence of time trends and dummies, commonly found in empirical applications. Our approach thus stands out as a valuable tool in these situations.

\begin{comment}
In summary, our results are not very sensitive to the choice of tuning parameters across all cases. However, our estimator's performance is not as robust with dependent regressors, particularly AR regressors, due to two factors: the reduced variation of \( z_{31} \) conditional on \( |z_{2}| > \sigma_{n} \) and the challenge in selecting \( \sigma_{n} \) due to differing tail behaviors on the left and right sides. Notably, HK1 outperforms all semiparametric estimators, as expected, given its efficiency in low dimensional cases. Theoretically, OY's estimator converges faster than HK2 and our estimator, with HK2's convergence rate being comparable to ours in this low-dimensional case.\footnote{To be more precise, this statement is valid only for certain specific bandwidths.} The simulation results align with this theoretical expectation. Lastly, it is important to note that HK1, HK2, and OY estimators are not suitable in the presence of dummies and time trends, commonly found in empirical applications. Our approach serves as a valuable tool in these situations.
\end{comment}

% Table generated by Excel2LaTeX from sheet 'Sheet1'
\begin{table}[H]
\small
  \centering
  %\caption{Add caption}
    \begin{tabular}{ll}
    \toprule
    \multicolumn{2}{c}{List of Tables} \\
    \midrule
    Table \ref{T:D1_robust} & Simulation Results of Design 1 (Sensitivity Check) \\
    Table \ref{T:D2_robust} & Simulation Results of Design 2: (Sensitivity Check) \\
    Table \ref{T:D.3} & Simulation Results of Design 3 with AR Regressors \\
    Table \ref{T:D.4} & Simulation Results of Design 3 with AR Regressors (Sensitivity Check) \\
    Table \ref{T:D.5} & Simulation Results of Design 3 with MA Regressors \\
    Table \ref{T:D.6} & Simulation Results of Design 3 with MA Regressors (Sensitivity Check) \\
    Table \ref{T:D.7} & Simulation Results of Design 4 with Independent Regressors \\
    Table \ref{T:D.8} & \multicolumn{1}{p{37.18em}}{Simulation Results of Design 4 with Independent Regressors (Sensitivity Check)} \\
    Table \ref{T:D.9} & Simulation Results of Design 4 with AR Regressors \\
    Table \ref{T:D.10} & Simulation Results of Design 4 with AR Regressors (Sensitivity Check) \\
    Table \ref{T:D.11} & Simulation Results of Design 4 with MA Regressors \\
    Table \ref{T:D.12} & Simulation Results of Design 4 with MA Regressors (Sensitivity Check) \\
    Table \ref{T:D.13} & Simulation Results of Design 5 with AR Regressors \\
    Table \ref{T:D.14} & Simulation Results of Design 5 with AR Regressors: (Sensitivity Check) \\
    Table \ref{T:D.15} & Simulation Results of Design 5 with MA Regressors \\
    Table \ref{T:D.16} & Simulation Results of Design 5 with MA Regressors: (Sensitivity Check) \\
    Table \ref{T:D.17} & Simulation Results of Design 6 with Independent Regressors \\
    Table \ref{T:D.18} & Simulation Results of Design 6 with Independent Regressors (Sensitivity Check) \\
    Table \ref{T:D.19} & Simulation Results of Design 6 with AR Regressors \\
    Table \ref{T:D.20} & Simulation Results of Design 6 with AR Regressors (Sensitivity Check) \\
    Table \ref{T:D.21} & Simulation Results of Design 6 with MA Regressors \\
Table \ref{T:last} & Simulation Results of Design 6 with MA Regressors (Sensitivity Check) \\
    \bottomrule
    \end{tabular}
\end{table}

\begin{table}[H]
\small
\centering \caption{Simulation Results of Design 1 (Sensitivity Check)}
\label{T:D1_robust} %
\begin{tabular}{r|cc|cc|cc}
\hline  \hline
 & \multicolumn{2}{c|}{$\beta_{1}$} & \multicolumn{2}{c|}{$\gamma$} & \multicolumn{2}{c}{$\delta$}\tabularnewline
 & MBIAS  & RMSE  & MBIAS  & RMSE  & MBIAS  & RMSE \tabularnewline
\hline
\multicolumn{7}{c}{Panel A: $\sigma_{n}=0.9\cdot\widehat{\text{std}(z_{i2})}\sqrt{\log n^{*}/2.95}$}\tabularnewline
\hline
$n_{1}$  & 0.071  & 0.331  & -0.042  & 0.476  & 0.057  & 0.198 \tabularnewline
Norm $n_{2}$  & 0.055  & 0.257  & -0.037  & 0.368  & 0.050  & 0.149 \tabularnewline
$n_{3}$  & 0.032  & 0.193  & -0.080  & 0.294  & 0.044  & 0.121 \tabularnewline
\hline
$n_{1}$  & 0.043  & 0.237  & -0.036  & 0.366  & 0.030  & 0.150 \tabularnewline
Lap $n_{2}$  & 0.031  & 0.177  & -0.046  & 0.286  & 0.023  & 0.111 \tabularnewline
$n_{3}$  & 0.015  & 0.139  & -0.068  & 0.238  & 0.015  & 0.087 \tabularnewline
\hline
\multicolumn{7}{c}{Panel B: $\sigma_{n}=1.1\cdot\widehat{\text{std}(z_{i2})}\sqrt{\log n^{*}/2.95}$}\tabularnewline
\hline
$n_{1}$  & 0.146  & 0.468  & 0.011  & 0.606  & 0.080  & 0.262 \tabularnewline
Norm $n_{2}$  & 0.107  & 0.375  & 0.022  & 0.507  & 0.070  & 0.223 \tabularnewline
$n_{3}$  & 0.069  & 0.289  & -0.033  & 0.415  & 0.048  & 0.163 \tabularnewline
\hline
$n_{1}$  & 0.037  & 0.267  & -0.014  & 0.407  & 0.024  & 0.160 \tabularnewline
Lap $n_{2}$  & 0.020  & 0.195  & -0.042  & 0.314  & 0.019  & 0.122 \tabularnewline
$n_{3}$  & 0.027  & 0.156  & -0.030  & 0.268  & 0.014  & 0.098 \tabularnewline
\hline
\hline
\multicolumn{7}{l}{Note: $n_{1}=5000,n_{2}=10000,n_{3}=20000$.}\tabularnewline
\end{tabular}
\end{table}

\begin{table}[H]
\small
\centering \caption{Simulation Results of Design 2: (Sensitivity Check)}
\label{T:D2_robust} \begin{adjustbox}{width=\textwidth} %
\begin{tabular}{r|cc|cc|cc|cc}
\hline  \hline
 & \multicolumn{2}{c|}{$\beta_{1}$} & \multicolumn{2}{c|}{$\beta_{2}$} & \multicolumn{2}{c|}{$\gamma$} & \multicolumn{2}{c}{$\delta$}\tabularnewline
 & MBIAS  & RMSE  & MBIAS  & RMSE  & MBIAS  & RMSE  & MBIAS  & RMSE \tabularnewline
\hline
\multicolumn{9}{c}{Panel A: $\sigma_{n}=0.9\cdot\widehat{\text{std}(z_{i2})}\sqrt{\log n^{*}/2.95}$}\tabularnewline
\hline
       $n_1$  & 0.123 & 0.414 & 0.122 & 0.406 & -0.028  & 0.468  & 0.083  & 0.201 \\
 Norm  $n_2$  & 0.073 & 0.305 & 0.063 & 0.302 & -0.064  & 0.381  & 0.056  & 0.153 \\
       $n_3$  & 0.048 & 0.224 & 0.038 & 0.235 & -0.062  & 0.298  & 0.044  & 0.119 \\ \hline
       $n_1$  & 0.054 & 0.279 & 0.065 & 0.300 & -0.017  & 0.364  & 0.042  & 0.141 \\
 Lap   $n_2$  & 0.024 & 0.205 & 0.026 & 0.223 & -0.038  & 0.278  & 0.025  & 0.104 \\
       $n_3$  & 0.023 & 0.173 & 0.021 & 0.171 & -0.037  & 0.222  & 0.021  & 0.087 \\
\hline
\multicolumn{9}{c}{Panel B: $\sigma_{n}=1.1\cdot\widehat{\text{std}(z_{i2})}\sqrt{\log n^{*}/2.95}$}\tabularnewline
\hline
       $n_1$  & 0.180 & 0.559 & 0.178 & 0.537 & 0.051  & 0.634  & 0.100  & 0.263 \\
 Norm  $n_2$  & 0.102 & 0.412 & 0.126 & 0.434 & 0.021  & 0.497  & 0.063  & 0.203 \\
       $n_3$  & 0.070 & 0.333 & 0.078 & 0.328 & -0.015  & 0.412  & 0.050  & 0.162 \\\hline
       $n_1$  & 0.088 & 0.326 & 0.078 & 0.341 & 0.001  & 0.419  & 0.045  & 0.164 \\
 Lap   $n_2$  & 0.029 & 0.239 & 0.044 & 0.252 & -0.017  & 0.328  & 0.026  & 0.119 \\
       $n_3$  & 0.023 & 0.194 & 0.022 & 0.195 & -0.022  & 0.250  & 0.017  & 0.097 \\
\hline
\hline
\multicolumn{9}{l}{Note: $n_{1}=5000,n_{2}=10000,n_{3}=20000$.}\tabularnewline
\end{tabular}\end{adjustbox}
\end{table}

\begin{table}[H]
\small
\caption{Simulation Results of Design 3 with AR Regressors}
\centering %
\label{T:D.3}
\begin{tabular}{r|cc|cc|cc}
\hline  \hline
 & \multicolumn{2}{c|}{$\beta_{1}$} & \multicolumn{2}{c|}{$\gamma$} & \multicolumn{2}{c}{$\delta$}\tabularnewline
 & MBIAS  & RMSE  & MBIAS  & RMSE  & MBIAS  & RMSE \tabularnewline
\hline
$n_{1}$  & 0.165  & 0.540  & -0.006  & 0.602  & 0.109  & 0.285 \tabularnewline
Norm $n_{2}$  & 0.130  & 0.409  & -0.034  & 0.522  & 0.085  & 0.216 \tabularnewline
$n_{3}$  & 0.079  & 0.316  & -0.072  & 0.392  & 0.076  & 0.175 \tabularnewline
\hline
$n_{1}$  & 0.031  & 0.297  & -0.055  & 0.457  & 0.037  & 0.164 \tabularnewline
Lap $n_{2}$  & 0.027  & 0.228  & -0.098  & 0.366  & 0.034  & 0.131 \tabularnewline
$n_{3}$  & 0.015  & 0.175  & -0.095  & 0.307  & 0.027  & 0.105 \tabularnewline
\hline
\hline
\multicolumn{7}{l}{Note: $n_{1}=5000,n_{2}=10000,n_{3}=20000$.}\tabularnewline
\end{tabular}
\end{table}

\begin{table}[H]
\small
\centering \caption{Simulation Results of Design 3 with AR Regressors (Sensitivity Check)}
\label{T:D.4}
\begin{tabular}{r|cc|cc|cc}
\hline  \hline
 & \multicolumn{2}{c|}{$\beta_{1}$} & \multicolumn{2}{c|}{$\gamma$} & \multicolumn{2}{c}{$\delta$}\tabularnewline
 & MBIAS  & RMSE  & MBIAS  & RMSE  & MBIAS  & RMSE \tabularnewline
\hline
\multicolumn{7}{c}{Panel A: $\sigma_{n}=0.9\cdot\widehat{\text{std}(z_{i2})}\sqrt{\log n^{*}/2.95}$}\tabularnewline
\hline
$n_{1}$  & 0.136  & 0.461  & -0.014  & 0.537  & 0.092  & 0.249 \tabularnewline
Norm $n_{2}$  & 0.101  & 0.351  & -0.053  & 0.440  & 0.075  & 0.187 \tabularnewline
$n_{3}$  & 0.058  & 0.269  & -0.108  & 0.356  & 0.069  & 0.154 \tabularnewline
\hline
$n_{1}$  & 0.033  & 0.280  & -0.077  & 0.431  & 0.036  & 0.150 \tabularnewline
Lap $n_{2}$  & 0.024  & 0.210  & -0.110  & 0.344  & 0.034  & 0.122 \tabularnewline
$n_{3}$  & 0.004  & 0.155  & -0.109  & 0.283  & 0.030  & 0.101 \tabularnewline
\hline
\multicolumn{7}{c}{Panel B: $\sigma_{n}=1.1\cdot\widehat{\text{std}(z_{i2})}\sqrt{\log n^{*}/2.95}$}\tabularnewline
\hline
$n_{1}$  & 0.242  & 0.644  & 0.046  & 0.714  & 0.135  & 0.341 \tabularnewline
Norm $n_{2}$  & 0.166  & 0.504  & 0.015  & 0.618  & 0.110  & 0.275 \tabularnewline
$n_{3}$  & 0.127  & 0.385  & -0.031  & 0.465  & 0.084  & 0.206 \tabularnewline
\hline
$n_{1}$  & 0.038  & 0.325  & -0.023  & 0.492  & 0.042  & 0.171 \tabularnewline
Lap $n_{2}$  & 0.037  & 0.248  & -0.061  & 0.382  & 0.038  & 0.142 \tabularnewline
$n_{3}$  & 0.020  & 0.192  & -0.073  & 0.316  & 0.025  & 0.110 \tabularnewline
\hline
\hline
\multicolumn{7}{l}{Note: $n_{1}=5000,n_{2}=10000,n_{3}=20000$.}\tabularnewline
\end{tabular}
\end{table}

\begin{table}[H]
\small
\caption{Simulation Results of Design 3 with MA Regressors}
\centering %
\label{T:D.5}
\begin{tabular}{r|cc|cc|cc}
\hline  \hline
 & \multicolumn{2}{c|}{$\beta_{1}$} & \multicolumn{2}{c|}{$\gamma$} & \multicolumn{2}{c}{$\delta$}\tabularnewline
 & MBIAS  & RMSE  & MBIAS  & RMSE  & MBIAS  & RMSE \tabularnewline
\hline
\hline
$n_{1}$  & 0.121  & 0.424  & -0.018  & 0.529  & 0.078  & 0.245 \tabularnewline
Norm $n_{2}$  & 0.090  & 0.332  & -0.054  & 0.408  & 0.069  & 0.190 \tabularnewline
$n_{3}$  & 0.063  & 0.262  & -0.062  & 0.343  & 0.055  & 0.156 \tabularnewline
\hline
$n_{1}$  & 0.029  & 0.245  & -0.059  & 0.410  & 0.030  & 0.156 \tabularnewline
Lap $n_{2}$  & 0.033  & 0.189  & -0.070  & 0.326  & 0.034  & 0.128 \tabularnewline
$n_{3}$  & 0.009  & 0.149  & -0.083  & 0.260  & 0.028  & 0.100 \tabularnewline
\hline
\hline
\multicolumn{7}{l}{Note: $n_{1}=5000,n_{2}=10000,n_{3}=20000$.}\tabularnewline
\end{tabular}
\end{table}

\begin{table}[H]
\small
\centering \caption{Simulation Results of Design 3 with MA Regressors (Sensitivity Check)}
\label{T:D.6}
\begin{tabular}{r|cc|cc|cc}
\hline  \hline
 & \multicolumn{2}{c|}{$\beta_{1}$} & \multicolumn{2}{c|}{$\gamma$} & \multicolumn{2}{c}{$\delta$}\tabularnewline
 & MBIAS  & RMSE  & MBIAS  & RMSE  & MBIAS  & RMSE \tabularnewline
\hline
\multicolumn{7}{c}{Panel A: $\sigma_{n}=0.9\cdot\widehat{\text{std}(z_{i2})}\sqrt{\log n^{*}/2.95}$}\tabularnewline
\hline
$n_{1}$  & 0.112  & 0.394  & -0.040  & 0.461  & 0.085  & 0.236 \tabularnewline
Norm $n_{2}$  & 0.085  & 0.296  & -0.084  & 0.374  & 0.070  & 0.170 \tabularnewline
$n_{3}$  & 0.045  & 0.220  & -0.087  & 0.297  & 0.050  & 0.135 \tabularnewline
\hline
$n_{1}$  & 0.021  & 0.222  & -0.067  & 0.369  & 0.026  & 0.146 \tabularnewline
Lap $n_{2}$  & 0.021  & 0.176  & -0.079  & 0.315  & 0.029  & 0.118 \tabularnewline
$n_{3}$  & 0.012  & 0.147  & -0.099  & 0.249  & 0.028  & 0.094 \tabularnewline
\hline
\multicolumn{7}{c}{Panel B: $\sigma_{n}=1.1\cdot\widehat{\text{std}(z_{i2})}\sqrt{\log n^{*}/2.95}$}\tabularnewline
\hline
$n_{1}$  & 0.175  & 0.529  & 0.022  & 0.620  & 0.097  & 0.291 \tabularnewline
Norm $n_{2}$  & 0.108  & 0.380  & -0.013  & 0.491  & 0.081  & 0.229 \tabularnewline
$n_{3}$  & 0.078  & 0.317  & -0.047  & 0.388  & 0.070  & 0.187 \tabularnewline
\hline
$n_{1}$  & 0.031  & 0.262  & -0.039  & 0.442  & 0.022  & 0.169 \tabularnewline
Lap $n_{2}$  & 0.036  & 0.215  & -0.043  & 0.344  & 0.029  & 0.131 \tabularnewline
$n_{3}$  & 0.018  & 0.165  & -0.070  & 0.274  & 0.024  & 0.108 \tabularnewline
\hline
\hline
\multicolumn{7}{l}{Note: $n_{1}=5000,n_{2}=10000,n_{3}=20000$.}\tabularnewline
\end{tabular}
\end{table}

\begin{table}[H]
\small
\caption{Simulation Results of Design 4 with Independent Regressors}
\label{T:D.7}
\centering
\begin{tabular}{r|cc|cc}
\hline  \hline
 & \multicolumn{2}{c|}{$\beta_{1}$} & \multicolumn{2}{c}{$\gamma$}\tabularnewline
 & MBIAS  & RMSE  & MBIAS  & RMSE \tabularnewline
\hline
$n_{1}$  & 0.041  & 0.249  & -0.046  & 0.364 \tabularnewline
HK1 $n_{1}$  & 0.006  & 0.144  & -0.037  & 0.229 \tabularnewline
HK2 $n_{1}$ & 0.074  & 0.339  & 0.009  & 0.440 \tabularnewline
OY \  $n_{1}$  & 0.022  & 0.220  & -0.036  & 0.285 \tabularnewline
\cline{2-5}
$n_{2}$  & 0.039  & 0.214  & -0.045  & 0.291 \tabularnewline
Norm HK1 $n_{2}$ & 0.003  & 0.122  & -0.028  & 0.184 \tabularnewline
HK2 $n_{2}$ & 0.062  & 0.278  & 0.017  & 0.378 \tabularnewline
OY \  $n_{2}$  & 0.017  & 0.171  & -0.024  & 0.238 \tabularnewline
\cline{2-5}
$n_{3}$  & 0.027  & 0.162  & -0.063  & 0.238 \tabularnewline
HK1 $n_{3}$ & 0.004  & 0.093  & -0.028  & 0.139 \tabularnewline
HK2 $n_{3}$ & 0.046  & 0.239  & 0.009  & 0.306 \tabularnewline
OY \  $n_{3}$ & 0.011  & 0.130  & -0.030  & 0.186 \tabularnewline
\hline
$n_{1}$  & 0.025  & 0.192  & -0.036  & 0.275 \tabularnewline
HK1 $n_{1}$  & 0.011  & 0.162  & -0.036  & 0.246 \tabularnewline
HK2$n_{1}$  & 0.069  & 0.347  & 0.014  & 0.452 \tabularnewline
OY \  $n_{1}$  & 0.028  & 0.216  & -0.046  & 0.304 \tabularnewline
\cline{2-5}
$n_{2}$  & 0.015  & 0.150  & -0.042  & 0.211 \tabularnewline
Lap HK1 $n_{2}$  & 0.010  & 0.129  & -0.030  & 0.191 \tabularnewline
HK2 $n_{2}$ & 0.042  & 0.269  & -0.010  & 0.380 \tabularnewline
OY \  $n_{2}$  & 0.013  & 0.160  & -0.034  & 0.250 \tabularnewline
\cline{2-5}
$n_{3}$  & 0.015  & 0.123  & -0.045  & 0.189 \tabularnewline
HK1 $n_{3}$  & 0.002  & 0.099  & -0.026  & 0.151 \tabularnewline
HK2 $n_{3}$  & 0.030  & 0.239  & 0.008  & 0.322 \tabularnewline
OY \  $n_{3}$  & 0.013  & 0.129  & -0.016  & 0.207 \tabularnewline
\hline
\hline
\multicolumn{5}{l}{Note: $n_{1}=5000,n_{2}=10000,n_{3}=20000$.}\tabularnewline
\end{tabular}
\end{table}

\begin{table}[H]
\small
\centering \caption{Simulation Results of Design 4 with Independent Regressors (Sensitivity
Check)}
\label{T:D.8}
\begin{tabular}{r|cc|cc}
\hline  \hline
 & \multicolumn{2}{c|}{$\beta_{1}$} & \multicolumn{2}{c}{$\gamma$}\tabularnewline
 & MBIAS  & RMSE  & MBIAS  & RMSE \tabularnewline
\hline
\multicolumn{5}{c}{Panel A: $\sigma_{n}=0.9\cdot\widehat{\text{std}(z_{i2})}\sqrt{\log n^{*}/2.95}$}\tabularnewline
\hline
$n_{1}$  & 0.029  & 0.219  & -0.063  & 0.312 \tabularnewline
Norm $n_{2}$  & 0.028  & 0.184  & -0.070  & 0.252 \tabularnewline
$n_{3}$  & 0.016  & 0.138  & -0.072  & 0.217 \tabularnewline
\hline
$n_{1}$  & 0.018  & 0.182  & -0.042  & 0.261 \tabularnewline
Lap $n_{2}$  & 0.018  & 0.143  & -0.055  & 0.210 \tabularnewline
$n_{3}$  & 0.012  & 0.114  & -0.057  & 0.174 \tabularnewline
\hline
\multicolumn{5}{c}{Panel B: $\sigma_{n}=1.1\cdot\widehat{\text{std}(z_{i2})}\sqrt{\log n^{*}/2.95}$}\tabularnewline
\hline
$n_{1}$  & 0.053  & 0.290  & -0.031  & 0.390 \tabularnewline
Norm $n_{2}$  & 0.043  & 0.233  & -0.035  & 0.326 \tabularnewline
$n_{3}$  & 0.040  & 0.199  & -0.037  & 0.273 \tabularnewline
\hline
$n_{1}$  & 0.032  & 0.200  & -0.035  & 0.279 \tabularnewline
Lap $n_{2}$  & 0.018  & 0.157  & -0.032  & 0.223 \tabularnewline
$n_{3}$  & 0.014  & 0.125  & -0.034  & 0.195 \tabularnewline
\hline
\hline
\multicolumn{5}{l}{Note: $n_{1}=5000,n_{2}=10000,n_{3}=20000$.}\tabularnewline
\end{tabular}
\end{table}

\begin{table}[H]
\small
\caption{Simulation Results of Design 4 with AR Regressors}
\centering %
\label{T:D.9}
\begin{tabular}{r|cc|cc}
\hline  \hline
 & \multicolumn{2}{c|}{$\beta_{1}$} & \multicolumn{2}{c}{$\gamma$}\tabularnewline
 & MBIAS  & RMSE  & MBIAS  & RMSE \tabularnewline
\hline
\hline
$n_{1}$  & 0.098  & 0.388  & -0.072  & 0.434 \tabularnewline
HK1 $n_{1}$  & 0.010  & 0.140  & -0.036  & 0.203 \tabularnewline
HK2 $n_{1}$  & 0.071  & 0.370  & 0.030  & 0.436 \tabularnewline
OY \  $n_{1}$  & 0.043  & 0.292  & 0.003  & 0.353 \tabularnewline
\cline{2-5}
$n_{2}$  & 0.051  & 0.279  & -0.108  & 0.348 \tabularnewline
Norm HK1 $n_{2}$  & 0.003  & 0.112  & -0.039  & 0.161 \tabularnewline
HK2 $n_{2}$  & 0.072  & 0.320  & 0.028  & 0.371 \tabularnewline
OY \  $n_{2}$  & 0.036  & 0.229  & -0.006  & 0.255 \tabularnewline
\cline{2-5}
$n_{3}$  & 0.045  & 0.218  & -0.107  & 0.291 \tabularnewline
HK1 $n_{3}$  & 0.005  & 0.091  & -0.037  & 0.129 \tabularnewline
HK2 $n_{3}$  & 0.042  & 0.239  & -0.012  & 0.288 \tabularnewline
OY \  $n_{3}$  & 0.020  & 0.166  & -0.024  & 0.207 \tabularnewline
\hline
$n_{1}$  & 0.031  & 0.235  & -0.093  & 0.339 \tabularnewline
HK1 $n_{1}$  & 0.003  & 0.156  & -0.039  & 0.216 \tabularnewline
HK2 $n_{1}$  & 0.061  & 0.347  & 0.030  & 0.436 \tabularnewline
OY \  $n_{1}$  & 0.033  & 0.271  & 0.009  & 0.362 \tabularnewline
\cline{2-5}
$n_{2}$  & 0.007  & 0.176  & -0.097  & 0.270 \tabularnewline
Lap HK1 $n_{2}$  & 0.010  & 0.125  & -0.021  & 0.167 \tabularnewline
HK2 $n_{2}$  & 0.078  & 0.318  & 0.036  & 0.389 \tabularnewline
OY \  $n_{2}$  & 0.025  & 0.216  & 0.005  & 0.271 \tabularnewline
\cline{2-5}
$n_{3}$  & 0.005  & 0.142  & -0.108  & 0.229 \tabularnewline
HK1 $n_{3}$  & 0.001  & 0.096  & -0.028  & 0.135 \tabularnewline
HK2 $n_{3}$  & 0.039  & 0.260  & 0.026  & 0.325 \tabularnewline
OY \  $n_{3}$  & 0.021  & 0.175  & -0.016  & 0.213 \tabularnewline
\hline
\hline
\multicolumn{5}{l}{Note: $n_{1}=5000,n_{2}=10000,n_{3}=20000$.}\tabularnewline
\end{tabular}
\end{table}

\begin{table}[H]
\small
\centering \caption{Simulation Results of Design 4 with AR Regressors (Sensitivity Check)}
\label{T:D.10}
\begin{tabular}{r|cc|cc}
\hline  \hline
 & \multicolumn{2}{c|}{$\beta_{1}$} & \multicolumn{2}{c}{$\gamma$}\tabularnewline
 & MBIAS  & RMSE  & MBIAS  & RMSE \tabularnewline
\hline
\multicolumn{5}{c}{Panel A: $\sigma_{n}=0.9\cdot\widehat{\text{std}(z_{i2})}\sqrt{\log n^{*}/2.95}$}\tabularnewline
\hline
$n_{1}$  & 0.071  & 0.332  & -0.117  & 0.387 \tabularnewline
Norm $n_{2}$  & 0.036  & 0.237  & -0.117  & 0.313 \tabularnewline
$n_{3}$  & 0.027  & 0.187  & -0.134  & 0.265 \tabularnewline
\hline
$n_{1}$  & 0.016  & 0.216  & -0.115  & 0.329 \tabularnewline
Lap $n_{2}$  & 0.007  & 0.160  & -0.113  & 0.261 \tabularnewline
$n_{3}$  & 0.008  & 0.134  & -0.115  & 0.226 \tabularnewline
\hline
\multicolumn{5}{c}{Panel B: $\sigma_{n}=1.1\cdot\widehat{\text{std}(z_{i2})}\sqrt{\log n^{*}/2.95}$}\tabularnewline
\hline
$n_{1}$  & 0.123  & 0.445  & -0.039  & 0.503 \tabularnewline
Norm $n_{2}$  & 0.078  & 0.336  & -0.069  & 0.392 \tabularnewline
$n_{3}$  & 0.055  & 0.266  & -0.093  & 0.336 \tabularnewline
$n_{1}$  & 0.031  & 0.244  & -0.082  & 0.362 \tabularnewline
Lap $n_{2}$  & 0.020  & 0.187  & -0.077  & 0.282 \tabularnewline
$n_{3}$  & 0.009  & 0.154  & -0.082  & 0.242 \tabularnewline
\hline
\hline
\multicolumn{5}{l}{Note: $n_{1}=5000,n_{2}=10000,n_{3}=20000$.}\tabularnewline
\end{tabular}
\end{table}

\begin{table}[H]
\small
\caption{Simulation Results of Design 4 with MA Regressors}
\centering %
\label{T:D.11}
\begin{tabular}{r|cc|cc}
\hline  \hline
 & \multicolumn{2}{c|}{$\beta_{1}$} & \multicolumn{2}{c}{$\gamma$}\tabularnewline
 & MBIAS  & RMSE  & MBIAS  & RMSE \tabularnewline
\hline
\hline
$n_{1}$  & 0.062  & 0.303  & -0.078  & 0.385 \tabularnewline
HK1 $n_{1}$  & 0.006  & 0.146  & -0.020  & 0.214 \tabularnewline
HK2 $n_{1}$  & 0.080  & 0.361  & 0.074  & 0.480 \tabularnewline
OY \  $n_{1}$  & 0.044  & 0.279  & 0.039  & 0.359 \tabularnewline
\cline{2-5}
$n_{2}$  & 0.041  & 0.227  & -0.087  & 0.302 \tabularnewline
Norm HK1 $n_{2}$  & 0.003  & 0.109  & -0.030  & 0.167 \tabularnewline
HK2 $n_{2}$  & 0.040  & 0.273  & 0.028  & 0.372 \tabularnewline
OY \  $n_{2}$  & 0.030  & 0.208  & 0.017  & 0.273 \tabularnewline
\cline{2-5}
$n_{3}$  & 0.029  & 0.178  & -0.089  & 0.254 \tabularnewline
HK1 $n_{3}$  & 0.003  & 0.085  & -0.023  & 0.133 \tabularnewline
HK2 $n_{3}$  & 0.037  & 0.228  & 0.015  & 0.315 \tabularnewline
OY \  $n_{3}$  & 0.013  & 0.154  & -0.018  & 0.217 \tabularnewline
\hline
$n_{1}$  & 0.012  & 0.190  & -0.077  & 0.305 \tabularnewline
HK1 $n_{1}$  & 0.011  & 0.153  & -0.007  & 0.236 \tabularnewline
HK2 $n_{1}$  & 0.078  & 0.364  & 0.052  & 0.489 \tabularnewline
OY \  $n_{1}$  & 0.033  & 0.270  & 0.030  & 0.378 \tabularnewline
\cline{2-5}
$n_{2}$  & 0.005  & 0.147  & -0.092  & 0.253 \tabularnewline
Lap HK1 $n_{2}$  & 0.002  & 0.118  & -0.013  & 0.184 \tabularnewline
HK2 $n_{2}$  & 0.051  & 0.300  & 0.049  & 0.425 \tabularnewline
OY \  $n_{2}$  & 0.018  & 0.200  & 0.022  & 0.293 \tabularnewline
\cline{2-5}
$n_{3}$  & 0.013  & 0.119  & -0.081  & 0.214 \tabularnewline
HK1 $n_{3}$  & 0.003  & 0.093  & -0.014  & 0.141 \tabularnewline
HK2 $n_{3}$  & 0.052  & 0.250  & 0.052  & 0.350 \tabularnewline
OY \  $n_{3}$  & 0.018  & 0.155  & 0.002  & 0.226 \tabularnewline
\hline
\hline
\multicolumn{5}{l}{Note: $n_{1}=5000,n_{2}=10000,n_{3}=20000$.}\tabularnewline
\end{tabular}
\end{table}

\begin{table}[H]
\small
\centering \caption{Simulation Results of Design 4 with MA Regressors (Sensitivity Check)}
\label{T:D.12}
\begin{tabular}{r|cc|cc}
\hline  \hline
 & \multicolumn{2}{c|}{$\beta_{1}$} & \multicolumn{2}{c}{$\gamma$}\tabularnewline
 & MBIAS  & RMSE  & MBIAS  & RMSE \tabularnewline
\hline
\multicolumn{5}{c}{Panel A: $\sigma_{n}=0.9\cdot\widehat{\text{std}(z_{i2})}\sqrt{\log n^{*}/2.95}$}\tabularnewline
\hline
$n_{1}$  & 0.043  & 0.252  & -0.098  & 0.359 \tabularnewline
Norm $n_{2}$  & 0.032  & 0.210  & -0.091  & 0.276 \tabularnewline
$n_{3}$  & 0.021  & 0.161  & -0.106  & 0.230 \tabularnewline
\hline
$n_{1}$  & 0.013  & 0.180  & -0.086  & 0.286 \tabularnewline
Lap $n_{2}$  & 0.006  & 0.140  & -0.096  & 0.244 \tabularnewline
$n_{3}$  & 0.008  & 0.111  & -0.092  & 0.203 \tabularnewline
\hline
\multicolumn{5}{c}{Panel B: $\sigma_{n}=1.1\cdot\widehat{\text{std}(z_{i2})}\sqrt{\log n^{*}/2.95}$}\tabularnewline
\hline
$n_{1}$  & 0.088  & 0.375  & -0.044  & 0.450 \tabularnewline
Norm $n_{2}$  & 0.050  & 0.270  & -0.070  & 0.342 \tabularnewline
$n_{3}$  & 0.042  & 0.215  & -0.071  & 0.278 \tabularnewline
\hline
$n_{1}$  & 0.014  & 0.203  & -0.070  & 0.333 \tabularnewline
Lap $n_{2}$  & 0.012  & 0.164  & -0.074  & 0.264 \tabularnewline
$n_{3}$  & 0.012  & 0.125  & -0.070  & 0.221 \tabularnewline
\hline
\hline
\multicolumn{5}{l}{Note: $n_{1}=5000,n_{2}=10000,n_{3}=20000$.}\tabularnewline
\end{tabular}
\end{table}

\begin{table}[H]
\small
\centering \caption{Simulation Results of Design 5 with AR Regressors} \label{T:D.13}
\begin{adjustbox}{width=\textwidth} %
\begin{tabular}{r|cc|cc|cc|cc}
\hline  \hline
 & \multicolumn{2}{c|}{$\beta_{1}$} & \multicolumn{2}{c|}{$\beta_{2}$} & \multicolumn{2}{c|}{$\gamma$} & \multicolumn{2}{c}{$\delta$}\tabularnewline
 & MBIAS  & RMSE  & MBIAS  & RMSE  & MBIAS  & RMSE  & MBIAS  & RMSE \tabularnewline
\hline

       $n_1$ & 0.197 & 0.639 & 0.248 & 0.664 & 0.049  & 0.668  & 0.138  & 0.313 \\
 Norm  $n_2$ & 0.144 & 0.515 & 0.155 & 0.517 & -0.001  & 0.563  & 0.093  & 0.235 \\
       $n_3$ & 0.131 & 0.426 & 0.105 & 0.409 & -0.047  & 0.440  & 0.082  & 0.197 \\\hline
       $n_1$ & 0.075 & 0.408 & 0.051 & 0.386 & -0.014  & 0.475  & 0.047  & 0.167 \\
 Lap   $n_2$ & 0.020 & 0.301 & 0.045 & 0.302 & -0.056  & 0.372  & 0.034  & 0.134 \\
       $n_3$ & 0.000 & 0.232 & -0.002 & 0.236 & -0.073  & 0.296  & 0.025  & 0.109 \\
\hline
\hline
\multicolumn{9}{l}{Note: $n_{1}=5000,n_{2}=10000,n_{3}=20000$.}\tabularnewline
\end{tabular}\end{adjustbox}
\end{table}

\begin{table}[H]
\small
\centering \caption{Simulation Results of Design 5 with AR Regressors: (Sensitivity Check)} \label{T:D.14}
\begin{adjustbox}{width=\textwidth} %
\begin{tabular}{r|cc|cc|cc|cc}
\hline  \hline
 & \multicolumn{2}{c|}{$\beta_{1}$} & \multicolumn{2}{c|}{$\beta_{2}$} & \multicolumn{2}{c|}{$\gamma$} & \multicolumn{2}{c}{$\delta$}\tabularnewline
 & MBIAS  & RMSE  & MBIAS  & RMSE  & MBIAS  & RMSE  & MBIAS  & RMSE \tabularnewline
\hline
\multicolumn{9}{c}{Panel A: $\sigma_{n}=0.9\cdot\widehat{\text{std}(z_{i2})}\sqrt{\log n^{*}/2.95}$}\tabularnewline
\hline
       $n_1$ & 0.196 & 0.585 & 0.221 & 0.603 & 0.006  & 0.611  & 0.127  & 0.276 \\
 Norm  $n_2$ & 0.125 & 0.436 & 0.125 & 0.440 & -0.061  & 0.478  & 0.092  & 0.211 \\
       $n_3$ & 0.097 & 0.351 & 0.083 & 0.340 & -0.083  & 0.373  & 0.071  & 0.167 \\\hline
       $n_1$ & 0.048 & 0.367 & 0.039 & 0.354 & -0.049  & 0.436  & 0.045  & 0.158 \\
 Lap   $n_2$ & 0.020 & 0.268 & 0.050 & 0.286 & -0.082  & 0.342  & 0.037  & 0.126 \\
       $n_3$ & 0.006 & 0.217 & -0.000 & 0.212 & -0.085  & 0.282  & 0.028  & 0.098 \\
\hline
\multicolumn{9}{c}{Panel B: $\sigma_{n}=1.1\cdot\widehat{\text{std}(z_{i2})}\sqrt{\log n^{*}/2.95}$}\tabularnewline
\hline
       $n_1$ & 0.262 & 0.756 & 0.274 & 0.752 & 0.080  & 0.761  & 0.158  & 0.358 \\
 Norm  $n_2$ & 0.203 & 0.627 & 0.207 & 0.636 & 0.043  & 0.655  & 0.113  & 0.284 \\
       $n_3$ & 0.155 & 0.513 & 0.161 & 0.506 & 0.010  & 0.554  & 0.098  & 0.237 \\\hline
       $n_1$ & 0.086 & 0.443 & 0.054 & 0.414 & 0.029  & 0.525  & 0.051  & 0.177 \\
 Lap   $n_2$ & 0.019 & 0.315 & 0.044 & 0.329 & -0.034  & 0.399  & 0.031  & 0.134 \\
       $n_3$ & 0.015 & 0.251 & 0.014 & 0.259 & -0.053  & 0.313  & 0.027  & 0.111 \\
\hline
\hline
\multicolumn{9}{l}{Note: $n_{1}=5000,n_{2}=10000,n_{3}=20000$.}\tabularnewline
\end{tabular}\end{adjustbox}
\end{table}

\begin{table}[H]
\small
\centering \caption{Simulation Results of Design 5 with MA Regressors} \label{T:D.15}
\begin{adjustbox}{width=\textwidth} %
\begin{tabular}{r|cc|cc|cc|cc}
\hline \hline
 & \multicolumn{2}{c|}{$\beta_{1}$} & \multicolumn{2}{c|}{$\beta_{2}$} & \multicolumn{2}{c|}{$\gamma$} & \multicolumn{2}{c}{$\delta$}\tabularnewline
 & MBIAS  & RMSE  & MBIAS  & RMSE  & MBIAS  & RMSE  & MBIAS  & RMSE \tabularnewline
\hline
       $n_1$ & 0.182 & 0.550 & 0.171 & 0.547 & 0.047  & 0.582  & 0.096  & 0.275 \\
 Norm  $n_2$ & 0.116 & 0.424 & 0.128 & 0.438 & -0.028  & 0.473  & 0.066  & 0.209 \\
       $n_3$ & 0.099 & 0.344 & 0.087 & 0.333 & -0.050  & 0.371  & 0.063  & 0.164 \\\hline
       $n_1$ & 0.050 & 0.324 & 0.024 & 0.320 & -0.034  & 0.396  & 0.037  & 0.158 \\
 Lap   $n_2$ & 0.021 & 0.254 & 0.029 & 0.251 & -0.048  & 0.334  & 0.028  & 0.122 \\
       $n_3$ & 0.004 & 0.198 & -0.003 & 0.200 & -0.063  & 0.268  & 0.022  & 0.100 \\
\hline
\hline
\multicolumn{9}{l}{Note: $n_{1}=5000,n_{2}=10000,n_{3}=20000$.}\tabularnewline
\end{tabular}\end{adjustbox}
\end{table}

\begin{table}[H]
\small
\centering \caption{Simulation Results of Design 5 with MA Regressors: (Sensitivity Check)} \label{T:D.16}
\begin{adjustbox}{width=\textwidth} %
\begin{tabular}{r|cc|cc|cc|cc}
\hline \hline
 & \multicolumn{2}{c|}{$\beta_{1}$} & \multicolumn{2}{c|}{$\beta_{2}$} & \multicolumn{2}{c|}{$\gamma$} & \multicolumn{2}{c}{$\delta$}\tabularnewline
 & MBIAS  & RMSE  & MBIAS  & RMSE  & MBIAS  & RMSE  & MBIAS  & RMSE \tabularnewline
\hline
\multicolumn{9}{c}{Panel A: $\sigma_{n}=0.9\cdot\widehat{\text{std}(z_{i2})}\sqrt{\log n^{*}/2.95}$}\tabularnewline
\hline
       $n_1$ & 0.148 & 0.476 & 0.137 & 0.467 & 0.017  & 0.514  & 0.086  & 0.235 \\
 Norm  $n_2$ & 0.097 & 0.356 & 0.098 & 0.366 & -0.052  & 0.404  & 0.063  & 0.179 \\
       $n_3$ & 0.068 & 0.270 & 0.060 & 0.263 & -0.084  & 0.318  & 0.056  & 0.137 \\\hline
       $n_1$ & 0.036 & 0.297 & 0.030 & 0.283 & -0.053  & 0.367  & 0.032  & 0.144 \\
 Lap   $n_2$ & 0.011 & 0.227 & 0.020 & 0.237 & -0.079  & 0.323  & 0.030  & 0.114 \\
       $n_3$ & 0.005 & 0.181 & 0.000 & 0.183 & -0.076  & 0.243  & 0.026  & 0.092 \\
\hline
\multicolumn{9}{c}{Panel B: $\sigma_{n}=1.1\cdot\widehat{\text{std}(z_{i2})}\sqrt{\log n^{*}/2.95}$}\tabularnewline
\hline
       $n_1$ & 0.231 & 0.637 & 0.209 & 0.636 & 0.085  & 0.669  & 0.121  & 0.326 \\
 Norm  $n_2$ & 0.146 & 0.497 & 0.163 & 0.521 & 0.041  & 0.571  & 0.079  & 0.251 \\
       $n_3$ & 0.104 & 0.412 & 0.113 & 0.412 & -0.022  & 0.432  & 0.076  & 0.202 \\\hline
       $n_1$ & 0.053 & 0.351 & 0.046 & 0.343 & -0.007  & 0.436  & 0.035  & 0.174 \\
 Lap   $n_2$ & 0.023 & 0.266 & 0.030 & 0.281 & -0.027  & 0.357  & 0.026  & 0.132 \\
       $n_3$ & 0.005 & 0.216 & 0.004 & 0.219 & -0.052  & 0.289  & 0.022  & 0.104 \\
\hline
\hline
\multicolumn{9}{l}{Note: $n_{1}=5000,n_{2}=10000,n_{3}=20000$.}\tabularnewline
\end{tabular}\end{adjustbox}
\end{table}

\begin{table}[H]
\small
\centering \caption{Simulation Results of Design 6 with Independent Regressors} \label{T:D.17}
\begin{tabular}{r|cc|cc|cc}
\hline \hline
 & \multicolumn{2}{c|}{$\beta_{1}$} & \multicolumn{2}{c|}{$\beta_{2}$} & \multicolumn{2}{c}{$\gamma$}\tabularnewline
 & MBIAS  & RMSE  & MBIAS  & RMSE  & MBIAS  & RMSE  \tabularnewline
\hline
       $n_1$     & 0.063 & 0.316  & 0.061 & 0.314 & -0.021  & 0.373  \\
 HK1 $n_1$     & 0.028 & 0.176  & 0.019 & 0.174 & -0.078  & 0.232  \\
 HK2 $n_1$     & 0.110 & 0.379  & 0.104 & 0.376 & -0.009  & 0.421  \\
 OY \  $n_1$     & 0.043 & 0.257  & 0.032 & 0.258 & -0.021  & 0.296  \\ \cline{2-7}
$n_2$     & 0.052 & 0.253  & 0.049 & 0.258 & -0.048  & 0.295  \\
  Norm  HK1 $n_2$     & 0.010 & 0.145  & 0.015 & 0.139 & -0.061  & 0.180  \\
 HK2 $n_2$     & 0.082 & 0.326  & 0.077 & 0.323 & 0.006  & 0.373  \\
 OY \  $n_2$     & 0.022 & 0.190  & 0.018 & 0.194 & -0.018  & 0.238  \\\cline{2-7}
       $n_3$     & 0.039 & 0.195  & 0.028 & 0.197 & -0.057  & 0.230  \\
 HK1 $n_3$     & 0.009 & 0.109  & 0.011 & 0.111 & -0.055  & 0.147  \\
 HK2 $n_3$     & 0.049 & 0.267  & 0.044 & 0.248 & -0.037  & 0.305  \\
 OY \  $n_3$     & 0.019 & 0.154  & 0.009 & 0.152 & -0.014  & 0.191  \\\hline
       $n_1$     & 0.027 & 0.217  & 0.030 & 0.228 & -0.032  & 0.272  \\
 HK1 $n_1$     & 0.030 & 0.197  & 0.021 & 0.194 & -0.065  & 0.243  \\
 HK2 $n_1$     & 0.095 & 0.382  & 0.087 & 0.384 & -0.023  & 0.432  \\
 OY  \ $n_1$     & 0.026 & 0.244  & 0.017 & 0.250 & -0.007  & 0.306  \\\cline{2-7}
  $n_2$     & 0.020 & 0.184  & 0.017 & 0.175 & -0.037  & 0.219  \\
  Lap HK1 $n_2$     & 0.026 & 0.151  & 0.019 & 0.153 & -0.054  & 0.195  \\
 HK2 $n_2$     & 0.065 & 0.317  & 0.062 & 0.323 & -0.018  & 0.370  \\
 OY \  $n_2$     & 0.018 & 0.194  & 0.024 & 0.188 & -0.012  & 0.244  \\\cline{2-7}
       $n_3$     & 0.012 & 0.135  & 0.020 & 0.145 & -0.038  & 0.181  \\
 HK1 $n_3$     & 0.010 & 0.127  & 0.010 & 0.123 & -0.056  & 0.156  \\
 HK2 $n_3$     & 0.062 & 0.283  & 0.050 & 0.276 & -0.014  & 0.300  \\
 OY \  $n_3$     & 0.007 & 0.150  & 0.003 & 0.153 & -0.028  & 0.202  \\

\hline
\hline
\multicolumn{7}{l}{Note: $n_{1}=5000,n_{2}=10000,n_{3}=20000$.}\tabularnewline
\end{tabular}
\end{table}

\begin{table}[H]
\small
\centering \caption{Simulation Results of Design 6 with Independent Regressors (Sensitivity Check)} \label{T:D.18}
\begin{tabular}{r|cc|cc|cc}
\hline \hline
 & \multicolumn{2}{c|}{$\beta_{1}$} & \multicolumn{2}{c|}{$\beta_{2}$} & \multicolumn{2}{c}{$\gamma$}\tabularnewline
 & MBIAS  & RMSE  & MBIAS  & RMSE  & MBIAS  & RMSE  \tabularnewline
\hline
\multicolumn{7}{c}{Panel A: $\sigma_{n}=0.9\cdot\widehat{\text{std}(z_{i2})}\sqrt{\log n^{*}/2.95}$}\tabularnewline
\hline
       $n_1$     & 0.050 & 0.272  & 0.003 & 0.153 & 0.048  & 0.270  \\
 Norm  $n_2$     & 0.034 & 0.216  & 0.003 & 0.153 & 0.034  & 0.221  \\
       $n_3$     & 0.036 & 0.173  & 0.003 & 0.153 & 0.027  & 0.173  \\\hline
       $n_1$     & 0.025 & 0.210  & 0.003 & 0.153 & 0.027  & 0.215  \\
 Lap   $n_2$     & 0.011 & 0.166  & 0.003 & 0.153 & 0.009  & 0.169  \\
       $n_3$     & 0.010 & 0.130  & 0.003 & 0.153 & 0.016  & 0.133  \\
\hline
\multicolumn{7}{c}{Panel B: $\sigma_{n}=1.1\cdot\widehat{\text{std}(z_{i2})}\sqrt{\log n^{*}/2.95}$}\tabularnewline
\hline
       $n_1$     & 0.064 & 0.370  & 0.003 & 0.153 & 0.071  & 0.356  \\
 Norm  $n_2$     & 0.061 & 0.288  & 0.003 & 0.153 & 0.056  & 0.296  \\
       $n_3$     & 0.037 & 0.222  & 0.003 & 0.153 & 0.037  & 0.230  \\ \hline
       $n_1$     & 0.030 & 0.239  & 0.003 & 0.153 & 0.028  & 0.233  \\
 Lap   $n_2$     & 0.019 & 0.185  & 0.003 & 0.153 & 0.017  & 0.179  \\
       $n_3$     & 0.011 & 0.151  & 0.003 & 0.153 & 0.015  & 0.153  \\

\hline
\hline
\multicolumn{7}{l}{Note: $n_{1}=5000,n_{2}=10000,n_{3}=20000$.}\tabularnewline
\end{tabular}
\end{table}

\begin{table}[H]
\small
\centering \caption{Simulation Results of Design 6 with AR Regressors} \label{T:D.19}
\begin{tabular}{r|cc|cc|cc}
\hline \hline
 & \multicolumn{2}{c|}{$\beta_{1}$} & \multicolumn{2}{c|}{$\beta_{2}$} & \multicolumn{2}{c}{$\gamma$}\tabularnewline
 & MBIAS  & RMSE  & MBIAS  & RMSE  & MBIAS  & RMSE  \tabularnewline
\hline

       $n_1$     & 0.108 & 0.451  & 0.102 & 0.453 & -0.059  & 0.444  \\
 HK1 $n_1$     & 0.014 & 0.162  & 0.014 & 0.156 & -0.084  & 0.190  \\
 HK2 $n_1$     & 0.086 & 0.382  & 0.087 & 0.377 & -0.010  & 0.373  \\
 OY \  $n_1$     & 0.057 & 0.346  & 0.070 & 0.334 & 0.054  & 0.366  \\\cline{2-7}
 $n_2$     & 0.069 & 0.340  & 0.072 & 0.353 & -0.084  & 0.349  \\
 Norm  HK1 $n_2$     & 0.010 & 0.121  & 0.001 & 0.122 & -0.069  & 0.147  \\
 HK2 $n_2$     & 0.074 & 0.303  & 0.065 & 0.309 & -0.002  & 0.315  \\
 OY \  $n_2$     & 0.045 & 0.252  & 0.042 & 0.260 & 0.018  & 0.268  \\\cline{2-7}
       $n_3$     & 0.056 & 0.282  & 0.049 & 0.276 & -0.085  & 0.299  \\
 HK1 $n_3$     & 0.009 & 0.095  & 0.010 & 0.097 & -0.061  & 0.117  \\
 HK2 $n_3$     & 0.043 & 0.253  & 0.039 & 0.255 & -0.018  & 0.251  \\
 OY \  $n_3$     & 0.008 & 0.194  & 0.026 & 0.191 & 0.001  & 0.204  \\\hline
       $n_1$     & 0.020 & 0.292  & 0.028 & 0.296 & -0.063  & 0.331  \\
 HK1 $n_1$     & 0.019 & 0.165  & 0.021 & 0.171 & -0.059  & 0.191  \\
 HK2 $n_1$     & 0.100 & 0.397  & 0.102 & 0.411 & 0.017  & 0.404  \\
 OY \  $n_1$     & 0.029 & 0.320  & 0.023 & 0.308 & 0.026  & 0.332  \\\cline{2-7}
  $n_2$     & 0.017 & 0.224  & 0.011 & 0.231 & -0.072  & 0.280  \\
 Lap  HK1 $n_2$     & 0.015 & 0.131  & 0.003 & 0.133 & -0.063  & 0.160  \\
 HK2 $n_2$     & 0.073 & 0.331  & 0.041 & 0.301 & -0.023  & 0.329  \\
 OY \  $n_2$     & 0.032 & 0.259  & 0.026 & 0.250 & 0.001  & 0.271  \\\cline{2-7}
       $n_3$     & 0.020 & 0.178  & 0.017 & 0.187 & -0.098  & 0.232  \\
 HK1 $n_3$     & 0.016 & 0.106  & 0.014 & 0.105 & -0.057  & 0.129  \\
 HK2 $n_3$     & 0.049 & 0.266  & 0.047 & 0.265 & -0.032  & 0.272  \\
 OY \  $n_3$     & 0.007 & 0.189  & 0.005 & 0.193 & -0.004  & 0.209  \\
\hline
\hline
\multicolumn{7}{l}{Note: $n_{1}=5000,n_{2}=10000,n_{3}=20000$.}\tabularnewline
\end{tabular}
\end{table}

\begin{table}[H]
\small
\centering \caption{Simulation Results of Design 6 with AR Regressors (Sensitivity Check)} \label{T:D.20}
\begin{tabular}{r|cc|cc|cc}
\hline \hline
 & \multicolumn{2}{c|}{$\beta_{1}$} & \multicolumn{2}{c|}{$\beta_{2}$} & \multicolumn{2}{c}{$\gamma$}\tabularnewline
 & MBIAS  & RMSE  & MBIAS  & RMSE  & MBIAS  & RMSE  \tabularnewline
\hline
\multicolumn{7}{c}{Panel A: $\sigma_{n}=0.9\cdot\widehat{\text{std}(z_{i2})}\sqrt{\log n^{*}/2.95}$}\tabularnewline
\hline
       $n_1$     & 0.082 & 0.383  & 0.093 & 0.396 & -0.093  & 0.391  \\
 Norm  $n_2$     & 0.045 & 0.278  & 0.048 & 0.293 & -0.114  & 0.309  \\
       $n_3$     & 0.038 & 0.236  & 0.037 & 0.233 & -0.102  & 0.261  \\\hline
       $n_1$     & 0.024 & 0.279  & 0.027 & 0.268 & -0.096  & 0.306  \\
 Lap   $n_2$     & 0.004 & 0.208  & 0.004 & 0.218 & -0.094  & 0.262  \\
       $n_3$     & 0.009 & 0.170  & 0.004 & 0.173 & -0.112  & 0.221  \\
\hline
\multicolumn{7}{c}{Panel B: $\sigma_{n}=1.1\cdot\widehat{\text{std}(z_{i2})}\sqrt{\log n^{*}/2.95}$}\tabularnewline
\hline
       $n_1$     & 0.132 & 0.525  & 0.138 & 0.532 & -0.018  & 0.515  \\
 Norm  $n_2$     & 0.105 & 0.414  & 0.110 & 0.434 & -0.052  & 0.417  \\
       $n_3$     & 0.082 & 0.351  & 0.050 & 0.349 & -0.064  & 0.341  \\\hline
       $n_1$     & 0.033 & 0.326  & 0.022 & 0.322 & -0.040  & 0.360  \\
 Lap   $n_2$     & 0.015 & 0.240  & 0.016 & 0.252 & -0.060  & 0.292  \\
       $n_3$     & 0.026 & 0.198  & 0.016 & 0.200 & -0.081  & 0.241  \\
\hline
\hline
\multicolumn{7}{l}{Note: $n_{1}=5000,n_{2}=10000,n_{3}=20000$.}\tabularnewline
\end{tabular}
\end{table}

\begin{table}[H]
\small
\centering \caption{Simulation Results of Design 6 with MA Regressors} \label{T:D.21}
\begin{tabular}{r|cc|cc|cc}
\hline \hline
 & \multicolumn{2}{c|}{$\beta_{1}$} & \multicolumn{2}{c|}{$\beta_{2}$} & \multicolumn{2}{c}{$\gamma$}\tabularnewline
 & MBIAS  & RMSE  & MBIAS  & RMSE  & MBIAS  & RMSE  \tabularnewline
\hline

      $n_1$     & 0.087 & 0.376  & 0.078 & 0.369 & -0.033  & 0.373  \\
 HK1 $n_1$     & 0.012 & 0.156  & 0.006 & 0.152 & -0.061  & 0.196  \\
 HK2 $n_1$     & 0.077 & 0.362  & 0.070 & 0.351 & 0.012  & 0.399  \\
 OY \  $n_1$     & 0.059 & 0.322  & 0.052 & 0.306 & 0.056  & 0.347  \\\cline{2-7}
 $n_2$     & 0.054 & 0.283  & 0.055 & 0.286 & -0.077  & 0.303  \\
 Norm  HK1 $n_2$     & 0.005 & 0.120  & 0.002 & 0.123 & -0.060  & 0.159  \\
 HK2 $n_2$     & 0.059 & 0.302  & 0.038 & 0.299 & -0.002  & 0.333  \\
 OY \  $n_2$     & 0.032 & 0.231  & 0.032 & 0.235 & 0.039  & 0.275  \\\cline{2-7}
       $n_3$     & 0.036 & 0.216  & 0.032 & 0.229 & -0.069  & 0.257  \\
 HK1 $n_3$     & 0.008 & 0.098  & 0.007 & 0.096 & -0.040  & 0.123  \\
 HK2 $n_3$     & 0.050 & 0.257  & 0.049 & 0.254 & 0.007  & 0.284  \\
 OY \  $n_3$     & 0.017 & 0.180  & 0.020 & 0.176 & 0.005  & 0.207  \\\hline
       $n_1$     & 0.031 & 0.255  & 0.017 & 0.239 & -0.078  & 0.310  \\
 HK1 $n_1$     & 0.020 & 0.166  & 0.014 & 0.178 & -0.043  & 0.216  \\
 HK2 $n_1$     & 0.080 & 0.375  & 0.076 & 0.379 & 0.047  & 0.458  \\
 OY \  $n_1$     & 0.027 & 0.314  & 0.032 & 0.309 & 0.078  & 0.384  \\\cline{2-7}
  $n_2$     & 0.006 & 0.194  & 0.004 & 0.195 & -0.059  & 0.247  \\
 Lap  HK1 $n_2$     & 0.002 & 0.133  & 0.002 & 0.134 & -0.041  & 0.169  \\
 HK2 $n_2$     & 0.036 & 0.301  & 0.040 & 0.304 & 0.016  & 0.345  \\
 OY \  $n_2$     & 0.021 & 0.242  & 0.021 & 0.239 & 0.035  & 0.278  \\\cline{2-7}
       $n_3$     & 0.007 & 0.156  & 0.007 & 0.158 & -0.065  & 0.203  \\
 HK1 $n_3$     & 0.007 & 0.109  & 0.004 & 0.105 & -0.045  & 0.141  \\
 HK2 $n_3$     & 0.044 & 0.260  & 0.054 & 0.265 & 0.005  & 0.302  \\
 OY \  $n_3$     & 0.017 & 0.181  & 0.015 & 0.192 & 0.006  & 0.224  \\
\hline
\hline
\multicolumn{7}{l}{Note: $n_{1}=5000,n_{2}=10000,n_{3}=20000$.}\tabularnewline
\end{tabular}
\end{table}

\begin{table}[H]
\small
\centering \caption{Simulation Results of Design 6 with MA Regressors (Sensitivity Check)}\label{T:last}
\begin{tabular}{r|cc|cc|cc}
\hline \hline
 & \multicolumn{2}{c|}{$\beta_{1}$} & \multicolumn{2}{c|}{$\beta_{2}$} & \multicolumn{2}{c}{$\gamma$}\tabularnewline
 & MBIAS  & RMSE  & MBIAS  & RMSE  & MBIAS  & RMSE  \tabularnewline
\hline
\multicolumn{7}{c}{Panel A: $\sigma_{n}=0.9\cdot\widehat{\text{std}(z_{i2})}\sqrt{\log n^{*}/2.95}$}\tabularnewline
\hline
       $n_1$     & 0.050 & 0.272  & 0.048 & 0.270 & -0.053  & 0.336  \\
 Norm  $n_2$     & 0.034 & 0.216  & 0.034 & 0.221 & -0.075  & 0.259  \\
       $n_3$     & 0.036 & 0.173  & 0.027 & 0.173 & -0.060  & 0.209  \\ \hline
       $n_1$     & 0.025 & 0.210  & 0.027 & 0.215 & -0.046  & 0.261  \\
 Lap   $n_2$     & 0.011 & 0.166  & 0.009 & 0.169 & -0.048  & 0.211  \\
       $n_3$     & 0.010 & 0.130  & 0.016 & 0.133 & -0.056  & 0.176  \\
\hline
\multicolumn{7}{c}{Panel B: $\sigma_{n}=1.1\cdot\widehat{\text{std}(z_{i2})}\sqrt{\log n^{*}/2.95}$}\tabularnewline
\hline
       $n_1$     & 0.064 & 0.370  & 0.071 & 0.356 & -0.006  & 0.417  \\
 Norm  $n_2$     & 0.061 & 0.288  & 0.056 & 0.296 & -0.026  & 0.343  \\
       $n_3$     & 0.037 & 0.222  & 0.037 & 0.230 & -0.042  & 0.262  \\ \hline
       $n_1$     & 0.030 & 0.239  & 0.028 & 0.233 & -0.025  & 0.282  \\
 Lap   $n_2$     & 0.019 & 0.185  & 0.017 & 0.179 & -0.029  & 0.234  \\
       $n_3$     & 0.011 & 0.151  & 0.015 & 0.153 & -0.032  & 0.183  \\

\hline
\hline
\multicolumn{7}{l}{Note: $n_{1}=5000,n_{2}=10000,n_{3}=20000$.}\tabularnewline
\end{tabular}
\end{table}

\end{document}